\documentclass[12pt,a4paper]{article}                        

\usepackage{amsmath, amssymb, bm, enumitem, amsthm}
\usepackage{chngcntr}
\usepackage{enumitem, amsthm, xcolor, subcaption}
\usepackage{caption}
\usepackage[a4paper, top=30mm, bottom=30mm, left=25mm, right=25mm]{geometry}
\usepackage{grffile}
\usepackage{graphicx}
\usepackage{multirow}   
\usepackage{siunitx}
\usepackage{rotating}    
\usepackage{placeins}
\usepackage{titlesec}
\usepackage{appendix}
\usepackage{booktabs}        
\usepackage{threeparttable}  
\usepackage{xcolor,colortbl}   
\usepackage{array}            
\usepackage{hyperref}   
\hypersetup{hidelinks}

\usepackage{etoolbox}
\usepackage{aliascnt}
\usepackage{cleveref}
\usepackage{chngcntr}
\usepackage{soul}     

\usepackage[
  backend=biber,
  style=numeric-comp,
  sorting=none,
  date=year,
  giveninits=true,
  maxbibnames=999,
  url=true,
  doi=false,
  isbn=false,
  eprint=false
]{biblatex}

\DeclareFieldFormat
  [article,incollection,inproceedings,book,inbook,thesis,report,online,unpublished]
  {title}{#1}
\DeclareFieldFormat[article]{journaltitle}{\mkbibemph{#1}}
\DeclareFieldFormat[book]{title}{\mkbibemph{#1}}
\DeclareFieldFormat{labelnumberwidth}{#1.}
\DeclareNameAlias{author}{family-given}
\DeclareNameAlias{editor}{family-given}

\DeclareDelimFormat[bib]{multinamedelim}{\addcomma\space}
\DeclareDelimFormat[bib]{finalnamedelim}{\addcomma\space}
\DeclareFieldFormat{pages}{#1}
\DeclareFieldFormat{page}{#1}
\DeclareFieldFormat{pagetotal}{#1}
\DeclareFieldFormat{url}{\url{#1}}
\renewbibmacro{in:}{}

\DeclareBibliographyDriver{article}{%
  \usebibmacro{bibindex}%
  \usebibmacro{begentry}%
  \usebibmacro{author/translator+others}%
  \setunit{\labelnamepunct}\newblock
  \printfield{title}%
  \newunit\newblock
  \iffieldundef{shortjournal}
    {\printfield{journaltitle}}
    {\mkbibemph{\thefield{shortjournal}}}%
  \setunit{\adddot\space}%
  \printdate
  \setunit{\addsemicolon}%
  \printfield{volume}%
  \iffieldundef{number}
    {}
    {\mkbibparens{\printfield{number}}}%
  \setunit{\addcolon}%
  \printfield{pages}%
  \newunit\newblock
  \usebibmacro{finentry}%
}
\DeclareBibliographyDriver{book}{%
  \usebibmacro{bibindex}%
  \usebibmacro{begentry}%
  \usebibmacro{author/editor+others/translator+others}%
  \setunit{\labelnamepunct}\newblock
  \printfield{title}%
  \newunit
  \printfield{edition}%
  \newunit\newblock
  \printlist{location}%
  \setunit{\addcolon\space}%
  \printlist{publisher}%
  \setunit{\addsemicolon\space}%
  \printdate
  \iffieldundef{pagetotal}
    {}
    {\setunit{\adddot\space}\printfield{pagetotal}\addspace p.}%
  \newunit\newblock
  \usebibmacro{finentry}%
}
\DeclareBibliographyDriver{online}{%
  \usebibmacro{bibindex}%
  \usebibmacro{begentry}%
  \usebibmacro{author/editor+others/translator+others}%
  \setunit{\labelnamepunct}\newblock
  \printfield{title}%
  \newunit\newblock
  \printfield{note}%
  \newunit\newblock
  \printtext{Available from \printfield{url}}%
  \newunit\newblock
  \usebibmacro{finentry}%
}

\newlist{assumptionlist}{enumerate}{1}
\setlist[assumptionlist]{label=\textbf{A\arabic*.},ref=A\arabic*}
\crefname{assumptionlisti}{assumption}{assumptions}   
\Crefname{assumptionlisti}{Assumption}{Assumptions}   
\newlist{assumprime}{enumerate}{1}
\setlist[assumprime]{label=\textbf{$\text{A2}^\prime.$},ref=$\text{A2}^\prime$}
\crefalias{assumprimei}{assumptionlisti}
\newlist{assumpcons}{enumerate}{1}
\setlist[assumpcons]{label=\textbf{C.},ref=C}
\crefalias{assumpconsi}{assumptionlisti}

\theoremstyle{plain}

\renewcommand{\thetheorem}{\arabic{theorem}}

\newaliascnt{proposition}{theorem}
\newtheorem{proposition}[proposition]{Proposition}
\aliascntresetthe{proposition}

\newaliascnt{corollary}{theorem}

\aliascntresetthe{corollary}

\newaliascnt{lemma}{theorem}

\aliascntresetthe{lemma}

\theoremstyle{remark}

\crefname{theorem}{Theorem}{Theorems}
\Crefname{theorem}{Theorem}{Theorems}

\crefname{proposition}{Proposition}{Propositions}
\Crefname{proposition}{Proposition}{Propositions}

\crefname{corollary}{Corollary}{Corollaries}
\Crefname{corollary}{Corollary}{Corollaries}

\crefname{lemma}{Lemma}{Lemmas}
\Crefname{lemma}{Lemma}{Lemmas}

\crefname{remark}{remark}{remarks}
\Crefname{remark}{Remark}{Remarks}

\newaliascnt{appsec}{section}
\crefname{appsec}{Appendix}{Appendices}
\Crefname{appsec}{Appendix}{Appendices}
\makeatletter
\pretocmd{\appendix}{%
  \renewcommand{\thesection}{\Alph{section}}%
  \renewcommand{\thesubsection}{\Alph{section}.\arabic{subsection}}%
  \crefalias{section}{appsec}%
  \crefalias{subsection}{appsec}%
  \titleformat{\section}
    {\normalfont\Large\bfseries}
    {Appendix~\thesection}{1em}{}%
}{}{%
  \PackageWarning{appendix}{patching \appendix failed (unexpected)}%
}
\makeatother

\DeclareCaptionLabelFormat{suppl}{#1~#2}
\newcommand{\beginsupplement}{%
  \setcounter{table}{0}%
  \setcounter{figure}{0}%
  \renewcommand{\thetable}{S\arabic{table}}%
  \renewcommand{\thefigure}{S\arabic{figure}}%
  \captionsetup[table]{name=Supplementary Table, labelformat=suppl}%
  \captionsetup[figure]{name=Supplementary Figure, labelformat=suppl}%
  \crefname{table}{Table}{Tables}%
  \Crefname{table}{Table}{Tables}%
  \crefname{figure}{Figure}{Figures}%
  \Crefname{figure}{Figure}{Figures}%
}
\newcommand{\HeadingFont}{\normalfont\bfseries}

\titleformat{\section}      {\HeadingFont\Large}{\thesection}{1em}{}
\titleformat{\subsection}   {\HeadingFont\large}{\thesubsection}{1em}{}
\titleformat{\subsubsection}{\HeadingFont\normalsize}{\thesubsubsection}{1em}{}

\begin{document}

\begin{center}
{\LARGE
A Direct Variance Estimation (DiVE)\\
for Meta-Analysis of Median Differences
}
\end{center}
\vspace{10mm}

\begin{center}
{\large
Tadahisa Okuda$^{1,2}$, Masataka Taguri$^{1}$, Kenichi Hayashi$^{3}$\\
\vspace{1em}
$^{1}$Department of Health Data Science, Tokyo Medical University\\
$^{2}$Human Health Sciences, Kyoto University Graduate School of Medicine\\
$^{3}$Department of Mathematics, Keio University
}
\end{center}
\vspace{1em}

{
\noindent\textbf{Abstract}\\
Meta-analyses of two-group studies that report median differences typically rely on methods that require, in addition to the median difference and sample size, summary measures of dispersion such as quartiles or ranges.
Studies that do not report such statistics are often excluded from the meta-analysis.
Existing two-stage approaches first estimate the asymptotic variance of the median difference within each study under parametric assumptions, and then combine these study-specific estimates to obtain the pooled median difference and its variance.
We propose Direct Variance Estimation (DiVE), a method that directly estimates the variance of the pooled difference using only study-level median differences and their sample sizes.
A comprehensive simulation study across a wide range of distributional scenarios shows that DiVE performs comparably to or better than conventional two-stage methods. 
A re-analysis of published meta-analyses demonstrates that DiVE enables the inclusion of studies lacking dispersion statistics, leading to a more comprehensive and potentially less biased synthesis of evidence.
}

{
\leftskip=15mm \rightskip=15mm
\vskip 2mm
\noindent\small\textbf{KEYWORDS}\\
Meta-analysis, Median difference, Distribution-free method, Incomplete reporting, Skewed data, Variance estimation
}

\vspace{1em}
\section{Introduction}\label{sec:intro}
Evidence-based medicine (EBM) relies on meta-analysis to synthesize results from multiple independent studies and to provide more precise and generalizable estimates of treatment effects \supercite{Sutton2008-rz,Page2021-wh}.
When individual trials report apparently conflicting results, meta-analysis offers an overarching quantitative judgment and often underpins the formulation of clinical guidelines \supercite{Higgins2019-mn,Higgins2024-eb}.
The credibility of any meta-analysis hinges on (i) the comprehensive inclusion of all relevant studies and (ii) the accurate extraction and appropriate synthesis of study-level data.

Many meta-analyses in medicine pool differences in continuous outcomes between two groups.
When outcomes are approximately symmetric, investigators routinely report the sample mean and its standard deviation (SD).
However, outcomes such as clinical laboratory values, assessment scales, healthcare costs, or length of stay are frequently skewed.
The sample mean is then sensitive to extreme values and may not adequately represent the central tendency.
For such skewed data, the median is widely recommended because of its robustness to outliers, and an increasing number of primary studies now report the median difference as the treatment effect\supercite{Frokjaer2006-hh,Nagasako2009-hb,Calatayud2010-cv,Haycock2010-dk,White2011-tz,Shore2016-la,Desender2017-bp,Dmytriw2021-zz,Katzenschlager2021-ol}.

For mean differences, the conventional meta-analytic framework is inverse-variance weighting (IVW), where the pooled estimate is a weighted average of study-specific differences, with weights proportional to the inverses of their within-study variances \supercite{Sanchez-Meca1998-ss,Marin-Martinez2010-sd}.
In contrast, variances (or standard errors; SEs) for median differences are rarely reported, so the IVW machinery cannot be applied directly.
Two broad strategies have been proposed to address the absence of variance estimates.
First, transformation-based methods convert reported quantiles into surrogate means and SDs under an approximate normality assumption (possibly on a transformed scale), after which a conventional IVW meta-analysis is performed \supercite{Hozo2005-ik,Bland2014-qb,Wan2014-yu,Kwon2015-se,Luo2018-ig}.
These methods can work well when outcomes are roughly symmetric, but they rely on strong distributional assumptions and require several summary statistics beyond the median (e.g., quartiles or ranges).
Second, the quantile-estimation (QE) method relaxes the normality assumption by fitting several candidate parametric distributions (e.g., Weibull, log-normal) to the reported quantiles and then using the fitted density at the median to estimate the asymptotic variance of each study's sample median before applying IVW \supercite{McGrath2020-gt}.
QE thus weakens the parametric restriction of transformation-based approaches, but it still relies on explicit distributional modeling and requires the availability and accurate reporting of multiple quantiles per study.
Additional details on QE appear in~\Cref{sec:meth_qe}.

Although these methods are widely used, they face two critical limitations.
First, they require a correct specification of the outcome distribution; misspecification can bias variance estimates and lead to poor interval coverage.
Second, they require summary statistics beyond the median and sample size, which often leads to the exclusion of otherwise eligible studies.
This pattern of incomplete reporting has been widely documented \supercite{Wiebe2006-if,Chowdhry2016-pq}, underscoring the need for approaches that rely only on medians and sample sizes.
Specifically, Weir et al.\supercite{Weir2018-vv} showed that in a systematic review of early supported discharge (ESD) interventions, interquartile ranges (IQRs) were reported in only one-quarter of studies that reported medians.
Consistent with this evidence, our empirical application in \Cref{sec:exp} revisits the ESD dataset synthesized in the Cochrane review by Langhorne et al.\supercite{Langhorne2017-tw}, as summarized by Weir et al.\supercite{Weir2018-vv}.
In this dataset, many trials report only medians and sample sizes, omitting IQRs.
This pattern of incomplete reporting makes transformation-based analyses or QE analyses infeasible without excluding a substantial subset of otherwise eligible trials, thereby reducing precision and potentially introducing bias in the synthesis \supercite{Saldanha2020-va}.

To overcome these limitations, we introduce Direct Variance Estimation (DiVE), a distribution-free framework for meta-analysis when only study-level median differences and sample sizes are available.
The key idea is to separate the estimation of the pooled effect from the estimation of its variance.
We define the pooled point estimator as a weighted average of study-level summaries using any pre-specified weights; DiVE provides a closed-form estimator of the variance of the pooled estimator, based only on the chosen weights and the empirical dispersion of study-level effects around the pooled estimate.
It does not require estimating within-study variances or a between-study variance component, nor does it require specifying a parametric outcome family.
Under mild regularity conditions (formalized in~\Cref{sec:meth_dive}), the variance estimator is asymptotically unbiased for the variance of the pooled estimator.
Because the weights are treated as fixed, any defensible weighting scheme can be used; in this study, we consider using weights proportional to the study sample sizes.
By avoiding distributional assumptions and within-study variance estimation, DiVE allows inclusion of studies with incomplete reporting, increases robustness to model misspecification, and yields simple, closed-form inference.

The remainder of this paper is structured as follows.
\Cref{sec:meth} presents the statistical methods, including the standard meta-analytic framework, the existing QE method, and our proposed DiVE estimator.
\Cref{sec:sim} describes the design of the Monte Carlo simulation study and presents the results comparing DiVE with existing approaches.
\Cref{sec:exp} provides an application to a real-world dataset.
Finally, \Cref{sec:disc} concludes with a discussion of the findings and their broader implications.

\vspace{1em}
\section{Methods}\label{sec:meth}
\subsection{Meta-analytic Framework}\label{sec:meth_base}
We consider $N$ studies indexed by $i=1,\ldots,N$.
For study $i$, let $Y_i$ denote the study-level effect on the analysis scale (e.g., a between-group difference).
Throughout this paper, we assume that each study-level effect has a finite second moment, so that $\operatorname{Var}(Y_i) < \infty$ for $i=1,\ldots,N$.
There are two traditional statistical models for meta-analysis of continuous outcomes: the fixed-effect (FE) model and the random-effects (RE) model.
The standard meta-analytic framework relies on the following assumptions that underpin these two models \supercite{Hedges1985-os,Borenstein2009-hf,Nakagawa2023-fm}:
\begin{assumptionlist}
  \item\label{mainass:A1}
  $Y_1,\ldots,Y_N$ are mutually independent.
  \item\label{mainass:A2}
  $E[Y_i]=\mu$ for $i=1,\ldots,N$.
  \item\label{mainass:A3}
  $\operatorname{Var}(Y_1),\ldots,\operatorname{Var}(Y_N)$ are known or can be consistently estimated under the assumed model.
\end{assumptionlist}
Here, $\mu$ denotes the overall effect on the analysis scale.
We next describe the FE and RE models in detail \supercite{Hedges1998-oq,Borenstein2010-im}.

The FE model assumes a common true effect $\mu$ across studies:
\begin{equation}
  Y_i = \mu + \varepsilon_i, \quad {\rm where} \quad E[\varepsilon_i]=0, \quad \operatorname{Var}(\varepsilon_i)=\sigma_i^2,
  \label{eq:fe_model}
\end{equation}
with the within-study errors $\varepsilon_1,\ldots,\varepsilon_N$ assumed mutually independent.
The IVW estimator and its variance are
\begin{equation}
  \hat{\mu}_{\text{FE}} = \frac{\ \sum_{i=1}^{N} w_i Y_i\ }{\sum_{i=1}^{N} w_i}, \quad \widehat{\operatorname{Var}}(\hat\mu_{\text{FE}}) = \frac{1}{\ \sum_{i=1}^{N} w_i\ }, \quad w_i = \frac{1}{\ \sigma_i^2\ }.
  \label{eq:fe_est}
\end{equation}
Under \Cref{mainass:A1,mainass:A2,mainass:A3}, if one additionally assumes normal errors ($\varepsilon_i\sim\mathcal{N}(0,\sigma_i^2)$) with known $\sigma_i^2$, then $\hat{\mu}_{\text{FE}}$ is the uniformly minimum variance unbiased estimator (UMVUE) of $\mu$ \supercite{Hedges1985-os}.

In contrast, the RE model allows the true effect to vary across studies.
\begin{equation}
  Y_i=\mu+\delta_i+\varepsilon_i, \quad {\rm where} \quad E[\delta_i]=0, \quad \operatorname{Var}(\delta_i)=\tau^2,
  \label{eq:re_model}
\end{equation}
with $\varepsilon_1,\ldots,\varepsilon_N,\delta_1,\ldots,\delta_N$ mutually independent.
If $\tau^2$ were known, the IVW estimator and its variance are
\begin{equation}
  \hat{\mu}_{\text{RE}} = \frac{\ \sum_{i=1}^{N} w_i^{\ast} Y_i\ }{\sum_{i=1}^{N} w_i^{\ast}}, \quad \widehat{\operatorname{Var}}(\hat{\mu}_{\text{RE}}) = \frac{1}{\ \sum_{i=1}^{N} w_i^{\ast}\ }, \quad w_i^{\ast} = \frac{1}{\ \sigma_i^2 + \tau^2\ }.
  \label{eq:re_est}
\end{equation}
In practice, since $\tau^2$ is unknown, it is often estimated by the method-of-moments estimator of DerSimonian and Laird\supercite{DerSimonian1986-gc}:
\begin{equation}
  \hat{\tau}^2 = \max \left\{ 0,\ \frac{Q_{w} - (N - 1)}{\ \sum_{i=1}^{N} w_i - \frac{\sum_{i=1}^{N} {w_i}^2}{\sum_{i=1}^{N} w_i}\ } \right\},
  \quad  Q_{w} = \sum_{i=1}^{N} w_i \left( Y_i - \frac{\ \sum_{\ell=1}^{N} w_{\ell} Y_{\ell}\ }{\sum_{\ell=1}^{N} w_{\ell}} \right)^2.
  \label{eq:tau}
\end{equation}
Under \Cref{mainass:A1,mainass:A2,mainass:A3}, and assuming normality ($\varepsilon_i\sim\mathcal{N}(0,\sigma_i^2)$, $\delta_i\sim\mathcal{N}(0,\tau^2)$) with known $\sigma_i^2$ (for $i=1,\ldots,N$) and $\tau^2$, $\hat{\mu}_{\text{RE}}$ is the UMVUE of $\mu$ \supercite{Viechtbauer2005-kl}.

The remainder of this section focuses on pooling sample median differences across independent studies.

\subsection{Existing Approach: The Quantile-Estimation (QE) Method}\label{sec:meth_qe}
As a primary existing approach for meta-analysis of median differences, we describe the QE method proposed by McGrath et al.\supercite{McGrath2020-gt}.
Let $M_{ij}$ denote the sample median for each group $j\ (j = 1, 2)$ in study $i\ (i = 1,\ldots,N)$ and let $Y_i = M_{i1} - M_{i2}$ be the sample median difference.
The total sample size for study $i$ is $n_i = n_{i1} + n_{i2}$, with group-specific sample sizes $n_{i1}$ and $n_{i2}$.
QE requires several statistics---typically the first quartile $q_{ij}^{(1)}$, the third quartile $q_{ij}^{(3)}$, and the sample size---that are often reported with the median and can be used to estimate $\sigma_i^2$ (for $i=1,\ldots,N$).
The original QE algorithm also accommodates settings where the minimum and maximum values are reported in place of, or in addition to, the quartiles.
The core of the method involves fitting a parametric distribution $P_{ij}$ (e.g., normal, Weibull, log-normal, or gamma) to these empirical quantiles by estimating parameters $\theta_{ij}$ to minimize the sum of squared differences between theoretical and empirical quantiles.
The estimator $\hat{\theta}_{ij}$ for the parameter vector $\theta_{ij}$ is thus obtained by minimizing the loss function:
\begin{equation}
  S_{P_{ij}} (\theta_{ij}) = \left\{ F^{-1}(0.25|\theta_{ij}) - q_{ij}^{(1)} \right\}^2 + \left\{ F^{-1}(0.50|\theta_{ij}) - M_{ij} \right\}^2 + \left\{ F^{-1}(0.75|\theta_{ij}) - q_{ij}^{(3)} \right\}^2,
  \label{eq:qe_loss}
\end{equation}
where $F^{-1}(p|\theta_{ij})$ denotes the quantile function (the inverse of the cumulative distribution function) of the distribution $P_{ij}$, evaluated at probability $p$ for the given parameters $\theta_{ij}$.
In practice, this fitting is carried out for each candidate family of distribution, and the family achieving the smallest value of \Cref{eq:qe_loss} is selected separately for each group within each study.
The fitted density at the median of the selected fitted distribution is then used to obtain the within-study variance.
Specifically, the within-study variance of the sample median difference is estimated by
\begin{equation}
  \sum_{j=1}^{2}\frac{1}{4n_{ij}\hat{f}_{ij}\!\left\{F^{-1}(0.50|\hat{\theta}_{ij})\right\}^{2}},
  \label{eq:qe_var}
\end{equation}
where $\hat{f}_{ij}(\cdot)$ denotes the density of the selected fitted distribution $P_{ij}$ under the estimated parameters $\hat{\theta}_{ij}$.
IVW is then applied using the within-study variances from \Cref{eq:qe_var} (and, if modeling heterogeneity, an estimate of $\tau^2$), to obtain a pooled difference and its variance.
We next introduce our distribution-free alternative.

\subsection{The Proposed Direct Variance Estimation (DiVE) Method}\label{sec:meth_dive}
While QE enables variance estimation from sample quantiles, its applicability is limited in practice.
First, as implemented, QE restricts fitting to two-parameter families (e.g., normal, log-normal, Weibull, gamma).
When only quartiles are available, more flexible parametric families (e.g., three-parameter models) can achieve perfect fits to the observed quantiles, making model selection infeasible.
Second, accuracy depends on the correct specification of the parametric family; misspecification can bias the estimated density at the median and consequently bias $\widehat{\operatorname{Var}}(Y_i)$.
Third, at least three quantiles per group (e.g., median with quartiles or with min-max values) are needed; studies reporting only medians cannot be analyzed and are excluded.

To overcome the distributional dependence and multi-quantile requirements of existing approaches, we propose DiVE, a distribution-free meta-analytic method.
DiVE pools the sample median differences using weights proportional to study sample sizes.
It then estimates the variance of the pooled estimator from their weighted dispersion, without requiring \Cref{mainass:A3}.
The method does not require the estimates of  $\sigma_i^2$ (for $i=1,\ldots,N$); it requires only the reported pair $(Y_i,n_i)$ from each study and thus includes median-only studies.
Although we use sample size as a transparent proxy for precision, DiVE is not tied to a specific weighting scheme.
When more informative measures of precision are available (e.g., reliable within-study variances), they can be incorporated within the same framework.

Here we use weights proportional to study sample sizes, as the variance of a sample median decreases with sample size.
We define the pooled estimator using normalized sample-size weights $\tilde w_i^{\rm s}$ (for $i=1,\ldots,N$) as
\begin{equation}
  \hat{\mu}^{\rm s} = \sum_{i=1}^{N} \tilde w_i^{\rm s}\, Y_i,
  \qquad \tilde w_i^{\rm s} \;=\; \frac{n_i}{\sum_{\ell=1}^{N} n_\ell}.
  \label{eq:prep_estmu}
\end{equation}
Under independence across studies (\Cref{mainass:A1}), the variance of the pooled estimator satisfies
\begin{equation}
  \operatorname{Var}(\hat{\mu}^{\rm s}) = \sum_{i=1}^{N} (\tilde w_i^{\rm s})^2 \operatorname{Var}(Y_i).
  \label{eq:var_exact}
\end{equation}
To obtain a variance estimator, we consider the quantity $T = \sum_{i=1}^{N} k_i\,(Y_i-\hat{\mu}^{\rm s})^2$ and determine the constants $k_i$ (for $i=1,\ldots,N$) so that $E[T]$ matches the right-hand side of \Cref{eq:var_exact}; see \Cref{app:dive_derivation} for the algebra.
We then propose the direct estimator of $\operatorname{Var}(\hat{\mu}^{\rm s})$,
\begin{equation}
  \widehat{\operatorname{Var}}(\hat{\mu}^{\rm s}) \;=\; \sum_{i=1}^{N} \frac{h_i}{\,1+\sum_{\ell=1}^{N} h_\ell\,}\,(Y_i - \hat{\mu}^{\rm s})^2,
  \qquad h_i=\frac{(\tilde w_i^{\rm s})^2}{\,1-2\tilde w_i^{\rm s}\,},
  \label{eq:estvar_dive}
\end{equation}
which is well-defined when $\max_i \tilde w_i^{\rm s} < \tfrac{1}{2}$.

Under several assumptions, the proposed variance estimator has the following unbiasedness property.
\begin{proposition}\label{prop:var_exact_unbiased}
  Under \Cref{mainass:A1,mainass:A2} and the condition $\max_i \tilde w_i^{\rm s} < \tfrac{1}{2}$, we have
  \begin{equation*}
    E\!\left[\widehat{\operatorname{Var}}(\hat\mu^{\rm s})\right]
    \;=\;
    \operatorname{Var}(\hat\mu^{\rm s}).
  \end{equation*}
\end{proposition}
\noindent\textit{Proof.} See \Cref{app:dive_derivation}.\medskip

\noindent
For sample median differences, however, \Cref{mainass:A2} need not hold exactly.
To establish asymptotic unbiasedness of the DiVE variance estimator in this setting, write $n_{\min, i}:=\min(n_{i1},n_{i2})$ and $n_{\min}:=\min_i n_{\min, i}$, and replace \Cref{mainass:A2} by:
\begin{assumprime}
  \item\label{mainass:A2prime}
  For $i=1,\ldots,N$, $E[Y_i]\to\mu$ as $n_{\min, i}\to\infty$.
\end{assumprime}
We additionally assume the following regularity condition on the weight limits:
\begin{assumptionlist}[resume]
  \item\label{mainass:A4}
  For $i=1,\ldots,N$, the weights $\tilde w_i^{\rm s}=\tilde w_i^{\rm s}(n_{\min})$ converge to $\tilde w_i^{\rm s\ast}$ as $n_{\min}\to\infty$, where $\tilde w_i^{\rm s\ast}\ge 0$, $\sum_{i=1}^N \tilde w_i^{\rm s\ast}=1$, and $\max_i \tilde w_i^{\rm s\ast}<\tfrac{1}{2}$.
\end{assumptionlist}
For sample median differences, under standard regularity conditions, $E[Y_i] = \mu + O(n_{\min,i}^{-1})$ as $n_{\min, i}\to\infty$\supercite{Bahadur1966-fy,Lehmann1998-or}.
All asymptotics are taken with $N$ fixed and $n_{\min}\to\infty$.
Moreover, if each study-level median difference is consistent for the common effect (i.e., $Y_i\xrightarrow{p}\mu$ as $n_{\min}\to\infty$), then the pooled estimator is also consistent, $\hat\mu^{\rm s}\xrightarrow{p}\mu$, because it is a convex combination of $Y_1,\ldots,Y_N$.
Under these conditions, the proposed variance estimator has the following asymptotic unbiasedness property.

\begin{proposition}\label{prop:var_unbiased}
  Under \Cref{mainass:A1,mainass:A2prime,mainass:A4}, we have
  \begin{equation*}
    E\!\left[\widehat{\operatorname{Var}}(\hat\mu^{\rm s})\right]\;\longrightarrow\;\operatorname{Var}(\hat\mu^{\rm s}) \quad\text{as } n_{\min}\to\infty.
  \end{equation*}
\end{proposition}
\noindent\textit{Proof.} See \Cref{app:dive_proof}.\medskip

For interval estimation, as within-study sample sizes grow ($n_{\min}\to\infty$) each study-level median difference $Y_i$ is asymptotically normal; hence any fixed-weight linear combination $\sum_{i=1}^N \tilde w_i^{\rm s} Y_i$ is asymptotically normal with mean $\mu$ and variance $\operatorname{Var}(\hat\mu^{\rm s})$.
Accordingly, the Wald-type interval
\begin{equation}
  \hat\mu^{\rm s}\ \pm\ z_{1 - \alpha/2}\ \sqrt{\widehat{\operatorname{Var}}(\hat{\mu}^{\rm s})}
  \label{eq:prep_ci}
\end{equation}
is the natural choice; in small $N$ (e.g., $N<30$), replacing $z_{1-\alpha/2}$ with the $t$-quantile with $N\!-\!1$ degrees of freedom often improves finite-sample coverage \supercite{DAgostino1988-cf}.

When all studies have the same total sample size, as a special case, the weights satisfy $\tilde w_i^{\rm s} = 1/N$ and the pooled estimator equals the arithmetic mean of the study-level median differences,
\begin{equation}
  \hat{\mu}^{\rm s} = \frac{1}{N}\sum_{i=1}^N Y_i,
  \label{eq:mean_est}
\end{equation}
while \Cref{eq:estvar_dive} reduces to the familiar unbiased estimator of the variance of the sample mean,
\begin{equation}
  \widehat{\operatorname{Var}}(\hat{\mu}^{\rm s}) = \frac{1}{N(N\!-\!1)} \sum_{i=1}^{N} \left( Y_i - \hat{\mu}^{\rm s} \right)^2.
  \label{eq:unbias_estvar}
\end{equation}
This special case confirms the internal coherence of the proposed estimator.

We note that the proposed framework accommodates any pre-specified nonnegative study weights $\tilde w_i$ (for $i=1,\ldots,N$) provided they satisfy the no-dominance condition $\max_i \tilde w_i < \tfrac{1}{2}$.
In highly unbalanced situations where a single study receives a weight close to $\tfrac{1}{2}$, the factor $h_i=\tilde w_i^{2}/(1-2\tilde w_i)$ can become large, so the corresponding coefficient in the variance estimator can also become large, making the estimator increasingly sensitive to that study.

We emphasize that the validity of the estimator in \Cref{eq:estvar_dive} requires the study weights to be treated as fixed, that is, not depending on the observed summaries $Y_1,\ldots,Y_N$.
Under this requirement, alternative weight choices can also be considered.
For example, IVWs under an FE or RE model could be used if reliable study-specific variance estimates are available.
Because such variance estimates are often unavailable or difficult to justify in median-based meta-analysis settings, we adopt normalized sample-size weights in this paper.

\vspace{1em}
\section{Simulation Study}\label{sec:sim}
We conducted a fully factorial Monte Carlo study to evaluate the finite-sample performance of DiVE in terms of bias, mean squared error (MSE), and confidence interval (CI) coverage across a range of distributional shapes, sample-size configurations, and heterogeneity levels, compared with the QE method.
All simulations were performed in R 4.6.1 with $R=1,000$ replicates for each simulation setting.

\subsection{Simulation design and data generation}\label{sec:sim_design}
We considered a fully factorial design crossing the number of studies $N\in\{10,30\}$, the per-group sample-size pattern (fixed versus varying), the design-average per-group sample size $n\in\{100,300\}$, the outcome distribution (normal, skew-normal, log-normal), and between-study heterogeneity $I^2\in\{0,25,50,75\}\%$ (\Cref{tab:sim_design}).
Here $I^2$ denotes the proportion (0--1 scale) of the total variance in the study-level effects attributable to between-study heterogeneity; for each target $I^2$, we set $\tau^2$ using the unequal-size correction\supercite{Higgins2002-xr,Higgins2003-dz}, as detailed in the note to \Cref{tab:sim_design} and \Cref{app:metrics}.

For fixed-size scenarios, we set $n_{i1}=n_{i2}=n$ for all studies $i=1,\ldots,N$.
For varying-size scenarios, we targeted a total of $U=Nn$ observations in group~1 at the design-average per-group sample size $n$.
Each study initially received a baseline of 50 participants in group~1 to avoid extreme size imbalance, so that $50N$ observations were allocated deterministically.
The remaining $U-50N$ observations were then distributed across the $N$ studies according to a Dirichlet--multinomial distribution $\mathcal{DM}\!\left(U-50N;\,\boldsymbol{\alpha}\right)$ with $\boldsymbol{\alpha}=(1,\ldots,1)^\top\in\mathbb{R}^N$, yielding sample sizes $n_{i1}\ge 50$ and $\sum_{i=1}^N n_{i1}=U$.
A 1{:}1 allocation was enforced by setting $n_{i2}=n_{i1}$ for $i = 1,\ldots,N$.
Allocations were redrawn until no study accounted for 50\% or more of the group-specific total sample size, ensuring $\max_i \tilde w_i^{\rm s} < \tfrac{1}{2}$.

Outcome generation proceeded as follows.
Outcomes were generated in a two-step process to incorporate heterogeneity.
First, for each study $i$, we sampled a random effect $\delta_i \sim \mathcal{N}(0, \tau^2)$ (cf.\ \Cref{eq:re_model}).
Second, we drew observations from one of the distributional families described below, adding the realization $\delta_i$ to the group~1 observations to induce the random shift in the study-level effect.
Finally, summary statistics, including the sample median difference $Y_i$, were then computed from these generated samples, following the notation in \Cref{sec:meth}.
These $Y_i$ values served as the inputs to the meta-analytic estimators evaluated in \Cref{sec:sim_estimator}.

To reflect outcome shapes commonly encountered in practice, we generated individual-level outcomes for each study and group from one of three distributional families: normal, skew-normal, and log-normal.
Representative probability density functions for group~1 (solid) and group~2 (dashed) under each family are shown in \Cref{fig:out.d}.
In the normal outcome family, following the location-shift setup of McGrath et al.\supercite{McGrath2020-gt}, we shifted the observations of group~1 by a constant $c$.
The shift was calibrated separately for each design-average per-group sample size $n$ using a two-sided known-variance normal-mean power calculation with a significance level of 0.05 and nominal power of 0.60.
The exact parameterizations for all families and design settings are reported in Supplementary Table~S1.

\begin{figure}[htp]
  \centering
  \begin{subfigure}[t]{0.32\linewidth}
    \centering
    \includegraphics[width=\linewidth]{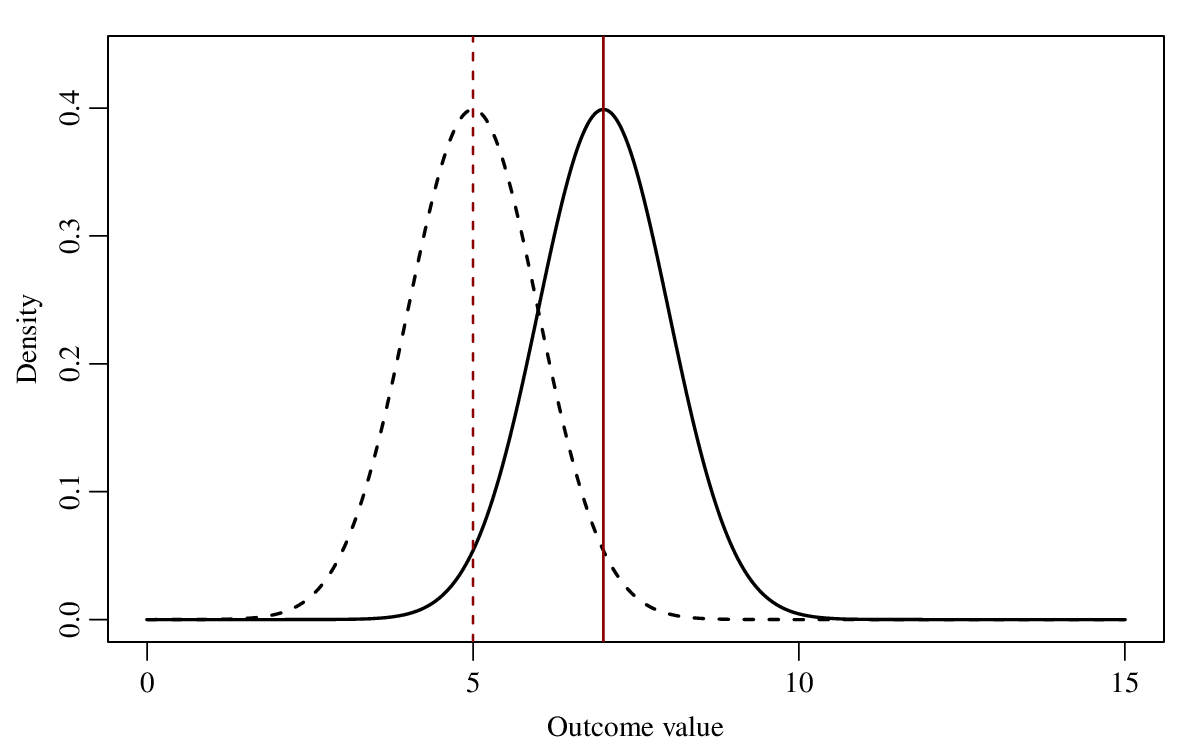}
    \caption{normal}
    \label{fig:out.d_N}
  \end{subfigure}
  \hfill
  \begin{subfigure}[t]{0.32\linewidth}
    \centering
    \includegraphics[width=\linewidth]{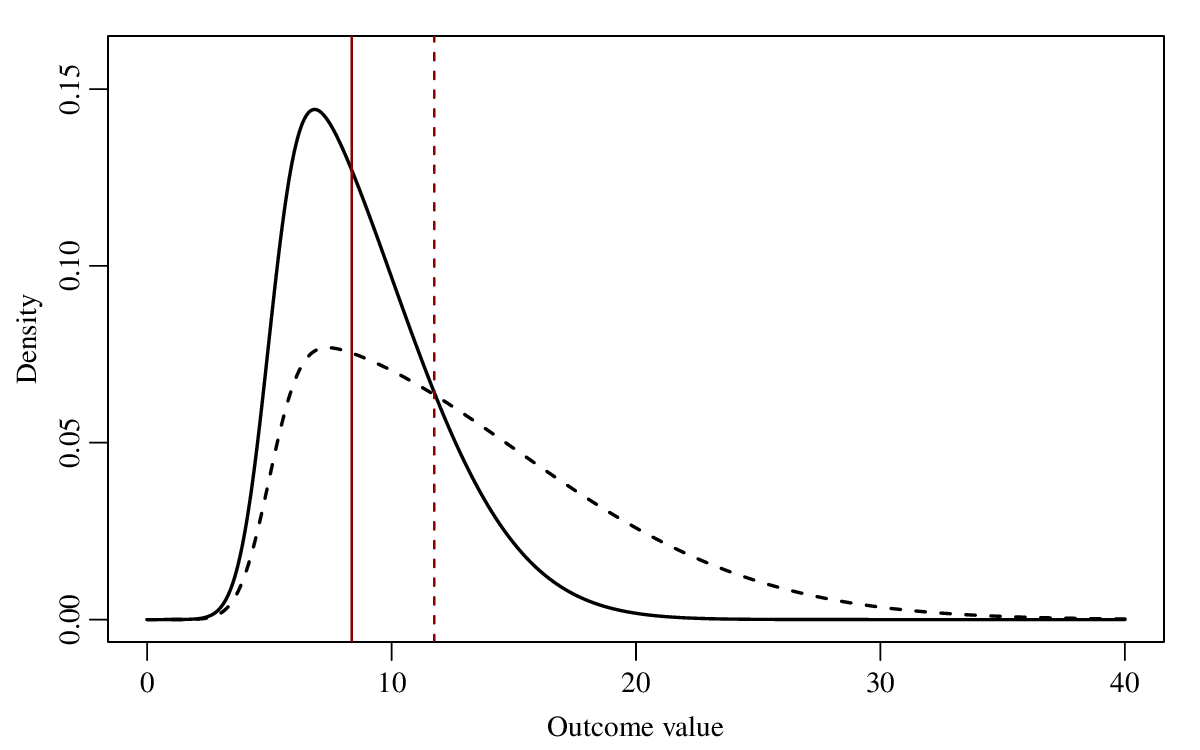}
    \caption{skew-normal}
    \label{fig:out.d_SN}
  \end{subfigure}
  \hfill
  \begin{subfigure}[t]{0.32\linewidth}
    \centering
    \includegraphics[width=\linewidth]{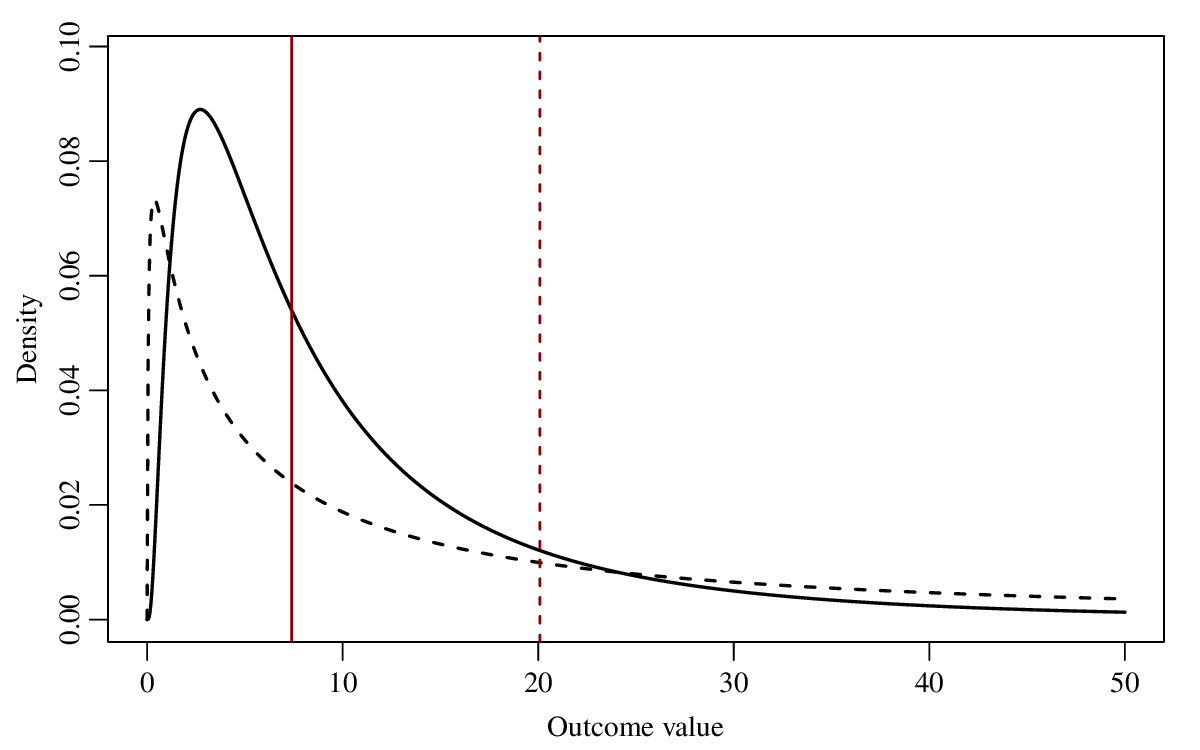}
    \caption{log-normal}
    \label{fig:out.d_LN}
  \end{subfigure}
  \caption{Probability density functions (PDFs) of the outcome distributions for group~1 (solid) and group~2 (dashed).\\[0.5em]
    {\footnotesize
      \textit{Note.} Vertical lines indicate the true median in each distribution.
      In panel~(a), $c=2$ is used only for visualization and does not represent the simulation-specific shift.
    }
  }
  \label{fig:out.d}
\end{figure}

To strictly evaluate the variance estimators, we constructed analytic benchmarks based on the known data-generating process.
The theoretical variance $\sigma_i^2$ of the study-level median difference $Y_i=M_{i1}-M_{i2}$ was computed from the underlying population densities at the group-specific medians and the realized sample sizes, using the standard large-sample approximation $\sum_{j=1}^{2} \{4n_{ij} f_{ij}(m_{ij})^2\}^{-1}$, where $m_{ij}$ denotes the true population median of group~$j$.
Combining these true $\sigma_i^2$ with the prescribed $\tau^2$ (as defined above), we derived the true variance of the pooled estimator corresponding to each weighting scheme.
These analytic targets served as the benchmarks for assessing the bias of the variance estimators in \Cref{sec:sim_metric}, with formal expressions given in \Cref{app:metrics}.

\subsection{Estimators}\label{sec:sim_estimator}
Each simulated dataset was analyzed using three meta-analytic estimators: DiVE, the QE method under a RE model (QE--RE), and the QE method under an FE model (QE--FE), as follows.
DiVE consists of (i) the sample-size-weighted pooled estimator and (ii) its direct variance estimator as defined in \Cref{sec:meth_dive}.
For the QE method, we used the ``median + IQR'' configuration with the standard four-family candidate set (normal, log-normal, Weibull, gamma), as described in \Cref{sec:meth_qe}.
QE analyses were implemented using the metamedian package (version 1.2.2).
Within each study and group, we selected the family that minimized the quantile-matching loss and used the fitted density to compute the within-study variance via \Cref{eq:qe_var}.
For QE--RE, the between-study variance was estimated within each replicate by the method of moments of DerSimonian and Laird\supercite{DerSimonian1986-gc}, and IVW was applied using the resulting weights.
For QE--FE, we set $\tau^2=0$ and applied IVW using the within-study variances from \Cref{eq:qe_var}.

\subsection{Performance Metrics}\label{sec:sim_metric}
We evaluated performance in terms of point estimation, variance estimation, and interval estimation.
For point estimation, we report relative bias (\%Bias) and relative MSE (\%MSE) of the pooled estimator $\hat\mu$ with respect to the true pooled difference in each design setting.
For variance estimation, we report \%Bias and \%MSE of $\widehat{\operatorname{Var}}(\hat{\mu})$ with respect to a method-specific analytic benchmark derived from the data-generating model.
For interval estimation, we report the empirical 95\% coverage probability (CP) and the average width (AW) of the interval.
CIs are constructed using $t$-quantiles with $N\!-\!1$ degrees of freedom unless otherwise specified.
We also report $z$-based intervals for comparison.
Replicate-level relative errors for both $\hat\mu$ and $\widehat{\operatorname{Var}}(\hat{\mu})$ are summarized to visualize distributional features beyond means and variances.

Analytic benchmarks are derived from the data-generating model (\Cref{app:metrics}).

\subsection{Results}\label{sec:sim_result}
In the main text, we report one representative skewed outcome scenario (log-normal; $N=30$, varying allocation, average $n=100$).
Full results for all design settings are reported in Supplementary Figures S1--S12 and Supplementary Tables S2--S24.
\Cref{fig:res_all_logN30n100} visualizes replicate-wise relative errors for the pooled estimate (panel~a) and for its variance (panel~b), and \Cref{tab:res_logN30n100} summarizes performance across $I^2\in\{0,25,50,75\}\%$.
For point estimation, DiVE showed stable relative bias across heterogeneity levels, whereas QE--RE's bias moved toward zero as heterogeneity increased.
Relative MSEs were comparable, ranging from $0.80$--$2.74\%$ for DiVE and $1.05$--$2.17\%$ for QE--RE.
For variance estimation, both DiVE and QE--RE exhibited small bias.
Using $t$-based 95\% intervals ($df=N\!-\!1$), empirical coverage was $0.922$, $0.931$, $0.934$, and $0.951$ for DiVE, with average widths $(\mathrm{AW})$ of $3.919$, $4.640$, $5.785$, and $8.286$ as $I^2$ increased from $0$ to $75\%$.
QE--RE's $t$-based coverage was $0.817$, $0.886$, $0.928$, and $0.957$, with AW of $3.898$, $4.451$, $5.487$, and $7.795$.
In contrast, QE--FE exhibited substantial undercoverage when heterogeneity was present.

\begin{figure}[htp]
  \centering
  \begin{subfigure}{\linewidth}
    \centering
    \includegraphics[width=\linewidth]{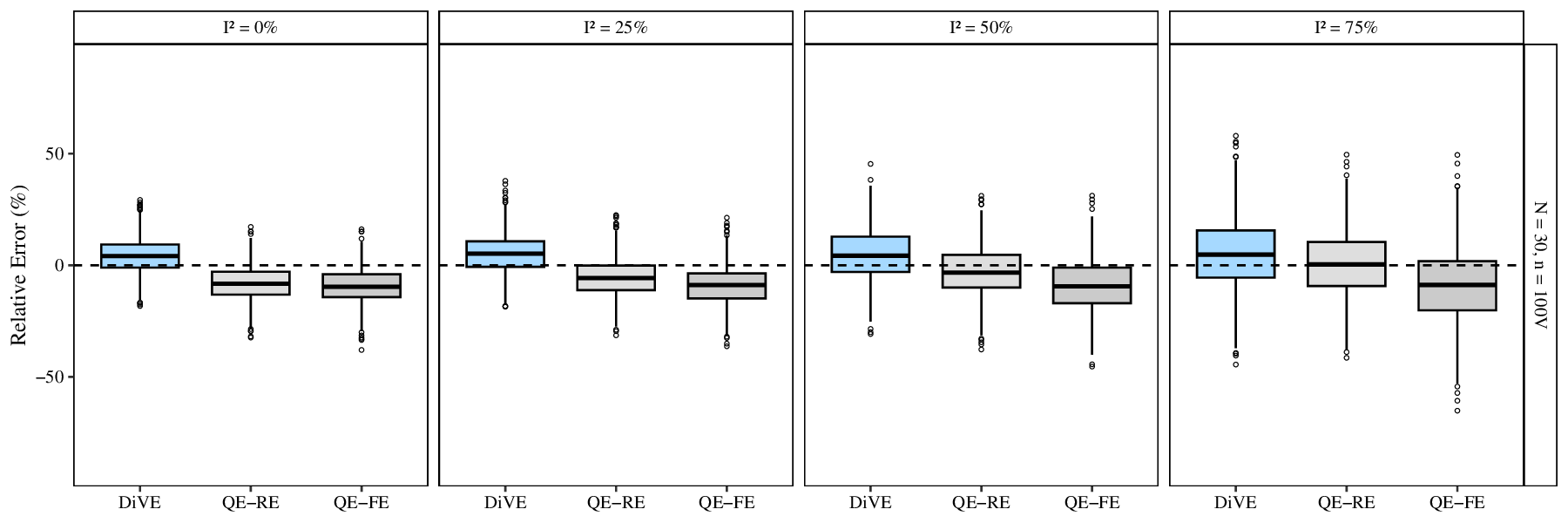}
    \caption{Point estimate.}\label{fig:res_box.p_logN30n100}
  \end{subfigure}
  \vspace{0.8em}
  \begin{subfigure}{\linewidth}
    \centering
    \includegraphics[width=\linewidth]{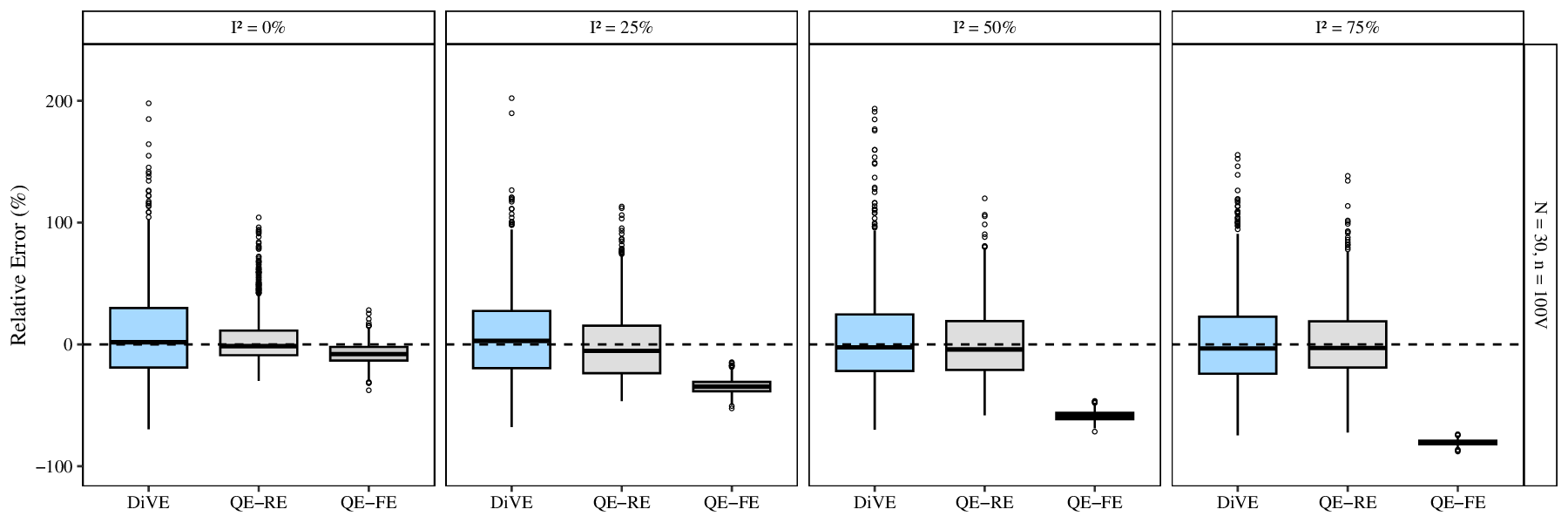}
    \caption{Variance estimate.}\label{fig:res_box.v_logN30n100}
  \end{subfigure}
  \caption[Point and variance errors (log-normal; $N{=}30$, $n{=}100$)]%
  {Distribution of relative errors under the challenging skewed outcome scenario (log-normal) with number of studies ($N=30$) and varying per-group sample sizes (average $n=100$).\\[0.5em]
    {\footnotesize
      \textit{Note.} Panels show (a) relative errors for the point estimator, and (b) relative errors for the variance estimator; methods compared are DiVE, the QE method under a RE model (QE--RE), and the QE method under an FE model (QE--FE).
      The four columns of panels correspond to $I^2\in\{0,25,50,75\}\%$. The central line in each box denotes the median error; boxes span the interquartile range; the dashed reference line at zero indicates no error.
    }
  }
  \label{fig:res_all_logN30n100}
\end{figure}

Under this skewed and size-imbalanced setting, both procedures produced pooled estimates close to the truth across all heterogeneity levels (\Cref{fig:res_box.p_logN30n100}; \Cref{tab:res_logN30n100}).
DiVE's small, positive relative bias remained remarkably stable as $I^2$ increased, whereas QE--RE's relative bias attenuated toward zero with greater heterogeneity.
This pattern is consistent with their constructions: DiVE's sample-size weighting does not hinge on a fitted distribution and therefore changes little as between-study dispersion grows, whereas QE--RE's RE weighting interacts with the parametric fit to the reported medians and IQRs, yielding bias that is more sensitive to the extent of heterogeneity.

For variance estimation, DiVE showed a modest upward bias that decreased as $I^2$ increased, whereas QE--RE's bias remained near zero on average but changed sign over $I^2$ (\Cref{fig:res_box.v_logN30n100}; \Cref{tab:res_logN30n100}).
With $t$-based critical values, DiVE maintained coverage close to nominal across the $I^2$ grid, with only moderate changes in average width.
In contrast, QE--RE was under-covered in the no-heterogeneity case but improved steadily as $I^2$ increased, eventually matching or slightly exceeding nominal coverage with somewhat shorter intervals at high heterogeneity.
Taken together, these results suggest that when only medians and sample sizes are available, and skewness is substantial, DiVE provides robust inference without distributional modeling, with competitive precision; when heterogeneity is moderate to large, QE--RE may offer modest gains in average width without compromising coverage.
Once heterogeneity departs from zero, QE--FE materially underestimates variance and its intervals under-cover, reinforcing that a RE formulation, or a distribution-free alternative such as DiVE, is preferable in practice (\Cref{tab:res_logN30n100}).

We emphasize that, throughout the simulation study, the QE method was evaluated under the assumption that quartiles were available for all studies.
In practice, however, QE also requires that Q1 and Q3 be reported alongside the median, and studies lacking these statistics cannot be included in the analysis as we describe in the next section.

\vspace{1em}
\section{Application to Real Data}\label{sec:exp}
As our empirical application, we analyzed a subset of randomized controlled trials (hereafter, ``studies'') from the Cochrane Review (CD000443) of ESD services for acute stroke\supercite{Langhorne2017-tw}.
In line with that review, we focused on the length of initial hospital stay (days) as the primary resource outcome and restricted attention to two-arm comparisons of ESD versus conventional care.
For the present analysis, we included studies that reported group-specific medians and sample sizes for this endpoint; eight studies met these criteria and were included in the analysis.

\subsection{Data extraction and preprocessing}\label{sec:exp_data}
From the published reports and Cochrane data tables, we extracted, for each eligible study, the group-specific sample sizes and any available measures of the median and dispersion of the length of initial hospital stay.
Studies were retained for DiVE if both group medians and group sizes were reported.
To apply QE--RE, we constructed a method-specific subset of studies reporting group-specific medians and IQRs in both arms.
This yielded eight studies for DiVE and two for QE--RE.

We oriented all effects so that negative values favor ESD (shorter initial hospital stay).
Specifically, study-level effects were defined as the ESD median minus the control median (for DiVE) and analogously for QE--RE.
Study identifiers and arm labels were aligned exactly with the Cochrane review; outcome definitions were not modified.

\subsection{Analysis plan}\label{sec:exp_plan}
In keeping with the paper's general framework, DiVE treated the pooling weights as fixed and estimated the variance of the resulting weighted average of study-level medians.
In this application, we used sample-size weights, which require only the reported group sizes.
Inference for the pooled median difference used a large-sample approximation with 95\% CIs computed using $t$-quantiles with $N\!-\!1$ degrees of freedom (where $N$ is the number of studies included by the method; $N=8$ for DiVE), and, as a sensitivity analysis, also reported the corresponding $z$-based intervals.

We applied QE--RE only to the subset of studies that reported median and IQRs in both groups, which was a strict subset of the DiVE dataset because QE--RE requires quartiles.
Following the QE approach, within-study variances of the sample medians were obtained by fitting the standard four-family candidate set (normal, log-normal, Weibull, gamma) to the reported quantiles, then applying IVW where the between-study variance was estimated by the DerSimonian--Laird method of moments \supercite{DerSimonian1986-gc}.
QE--FE was not considered in this application.

Because QE--RE required additional summary statistics beyond medians and sample sizes, the set of analyzable studies differed by method.
All choices above were prespecified and aligned with the Cochrane review's focus on length of initial hospital stay as a measure of resource use for ESD versus conventional care.
For transparency, we reported $\tilde w_i^{\rm s}$ (for $i=1,\ldots,N$) used in this application (see \Cref{tab:esd_inputs}); the maximum weight was 0.378 (contributed by London~1997) with $N=8$ studies.

\subsection{Results}\label{sec:exp_result}
Complete study-level inputs, including group medians, sample sizes, and quartiles where available, are summarized in \Cref{tab:esd_inputs}, and the corresponding DiVE and QE--RE results are reported in \Cref{tab:example_esd}.
Among the eight studies reporting group medians and sample sizes, the pooled median difference estimated by DiVE was $-5.69$ days (ESD minus conventional care); the 95\% CI based on $z$-quantiles was [$-8.18$,\,$-3.20$], and the corresponding $t$-based interval with $N\!-\!1=7$ degrees of freedom was [$-8.69$,\,$-2.69$].
By contrast, QE--RE could be applied to only two studies because it requires medians and IQRs in both arms.
On this two-study subset, the pooled median difference was $-5.92$ days, but uncertainty was substantial: the $z$-based 95\% interval remained wide, [$-22.54$,\,$10.70$], while the corresponding $t$-based interval with $N\!-\!1=1$ degree of freedom spanned [$-113.67$,\,$101.83$].
The large standard error (8.48 days) and the extremely wide $t$-interval reflect the combination of a very small number of studies ($N=2$) and the heavy finite-sample penalty from using $t$-critical values with one degree of freedom.

Taken together, these findings show that when only medians and minimal design information are available, DiVE yields a clinically interpretable reduction of about 5--6 days with well-calibrated uncertainty.
In contrast, QE--RE is severely limited by reporting requirements in this dataset, resulting in substantial imprecision despite similar point estimates.
Direct comparisons of interval width across methods should therefore be read with caution because the analyzable study sets differ by construction; nevertheless, DiVE's usable precision using eight studies suggests that inference without distributional modeling is feasible without discarding a large fraction of the evidence.

For the clinical context, the Cochrane review reported a similar reduction of 5--6 days using pragmatic approximations (medians treated as means and IQRs mapped to standard errors via normal assumptions).
This approach is distinct from rigorously derived transformation-based procedures\supercite{Wan2014-yu,Luo2018-ig}, so we cite it only as a contextual reference rather than a benchmark.

\vspace{1em}
\section{Discussion}\label{sec:disc}
This work introduces DiVE, a distribution-agnostic procedure that provides a closed-form, distribution-free estimator of the variance of a pooled estimator (e.g., pooled median differences) under pre-specified study weights.
DiVE requires only study-level central tendencies (typically medians) and group sizes; it does not reconstruct an outcome distribution or require estimation of within- or between-study variances.
Instead, it directly targets the variance of the chosen weighted average.
When the study-level effects are exactly unbiased for the common effect, the variance estimator is also unbiased for the variance of the chosen weighted average; for sample median differences, it is asymptotically unbiased under mild regularity conditions.
Our simulations demonstrate well-calibrated $t$-based CIs across a wide range of skewed distributions.
The empirical application further demonstrates that DiVE can synthesize evidence even when quartiles are incompletely reported. In such settings, parametric alternatives cannot be applied without discarding studies.

QE--RE  fits candidate parametric families to estimate within-study variances of the sample medians, and then applies IVW.
When the fitted family is well aligned with the data-generating distribution and the required quantiles are available, QE--RE can be efficient.
However, its applicability depends on the availability of at least three quantiles per group, and its variance estimates may be affected by model misspecification.

The DiVE variance estimator assumes that no single study dominates the weights ($\max_i \tilde w_i^{\rm s} < \tfrac{1}{2}$).
When total sample sizes are highly imbalanced, this condition may be nearly violated, leading to unstable variance estimates and inference being overly influenced by a single study.
However, such dominance would be problematic for any meta-analysis, not only DiVE, because it undermines the fundamental idea of synthesizing evidence across studies.

When applying DiVE, we recommend reporting the following items for reproducibility and transparency as reported in \Cref{tab:esd_inputs,tab:example_esd}: (i) the study-level medians and group-specific sample sizes, (ii) the study-specific weights and their maximum, and (iii) the choice of critical values (normal vs.\ $t$ with $df=N\!-\!1$).

As future work, while this study focuses on median differences for continuous outcomes, median survival time is commonly reported in oncology trials.
Because the proposed method requires only study-level sample sizes and medians, it is, in principle, applicable in such settings.
Further investigation is needed to assess its performance in the presence of censoring.

\vspace{1em}
\section*{Data availability statement.}
The simulation code, analysis code, curated summary-level dataset, and study materials supporting this study are available through the project repository and are archived at Zenodo: \url{https://doi.org/10.5281/zenodo.22920134}.
The real-data application uses study-level summary values transcribed from published reports and the Cochrane review.

\section*{Funding statement.}
Tadahisa Okuda received no specific funding for this work.
Masataka Taguri was supported by JSPS KAKENHI Grant Number JP24K14862.
Kenichi Hayashi was supported by JSPS KAKENHI Grant Number JP23K11013.

\section*{Competing interest statement.}
The authors declare that no competing interests exist.

\section*{Author contributions}
T.O.: Methodology, Formal analysis, Software, Writing -- original draft;
M.T.: Conceptualization, Methodology, Supervision, Writing -- review \& editing;
K.H.: Methodology, Supervision, Writing -- review \& editing.

\section*{Acknowledgments}
The authors are grateful to Professor Shuichi Kawano of Kyushu University for his helpful comments and suggestions on the manuscript.

\section*{Use of generative AI tools}
OpenAI ChatGPT (specific model versions were not recorded), accessed through the ChatGPT web interface, was used from September 2025 through April 2026 to assist with translation between Japanese and English and with language editing and proofreading.
It was used again in September 2026 to assist with identifying and reviewing the target journal's submission requirements and to provide suggestions concerning the preparation of journal-required presentation and submission materials based on those requirements.
All AI-assisted content incorporated into the manuscript or submission materials was reviewed by the authors, who take full responsibility for the final content of the manuscript and submission materials.
No generative AI tool was used to generate or alter study data, perform statistical analyses, or interpret the study findings.

\FloatBarrier
\newpage
\printbibliography[title={References}]

@ARTICLE{Bahadur1966-fy,
  title        = {A note on quantiles in large samples},
  author       = {Bahadur, R. Raj},
  journaltitle = {Ann. Math. Stat.},
  volume       = {37},
  number        = {3},
  pages        = {577--580},
  year         = {1966}
}

@ARTICLE{Bland2014-qb,
  title        = {Estimating mean and standard deviation from the sample size, three quartiles, minimum, and maximum},
  author       = {Bland, Martin},
  journaltitle = {Int. J. Stat. Med. Res.},
  volume       = {4},
  number        = {1},
  pages        = {57--64},
  year         = {2014}
}

@BOOK{Borenstein2009-hf,
  title     = {Introduction to meta-analysis},
  author    = {Borenstein, Michael and Hedges, Larry V. and Higgins, Julian P. T. and Rothstein, Hannah R.},
  publisher = {Wiley-Blackwell},
  location  = {Hoboken, NJ},
  year      = {2009},
  pagetotal = {452}
}

@ARTICLE{Borenstein2010-im,
  title        = {A basic introduction to fixed-effect and random-effects models for meta-analysis},
  author       = {Borenstein, Michael and Hedges, Larry V. and Higgins, Julian P. T. and Rothstein, Hannah R.},
  journaltitle = {Res. Synth. Methods},
  volume       = {1},
  number        = {2},
  pages        = {97--111},
  year         = {2010}
}

@ARTICLE{Calatayud2010-cv,
  title        = {Warm-up in a virtual reality environment improves performance in the operating room},
  author       = {Calatayud, Dan and Arora, Sonal and Aggarwal, Rajesh and Kruglikova, Irina and Schulze, Svend and Funch-Jensen, Peter and Grantcharov, Teodor},
  journaltitle = {Ann. Surg.},
  volume       = {251},
  number        = {6},
  pages        = {1181--1185},
  year         = {2010}
}

@ARTICLE{Chowdhry2016-pq,
  title        = {Meta-analysis with missing study-level sample variance data},
  author       = {Chowdhry, Amit K. and Dworkin, Robert H. and McDermott, Michael P.},
  journaltitle = {Stat. Med.},
  volume       = {35},
  number        = {17},
  pages        = {3021--3032},
  year         = {2016}
}

@ARTICLE{DAgostino1988-cf,
  title        = {The appropriateness of some common procedures for testing the equality of two independent binomial populations},
  author       = {D'Agostino, Ralph B. and Chase, Warren and Belanger, Albert},
  journaltitle = {Am. Stat.},
  volume       = {42},
  number        = {3},
  pages        = {198},
  year         = {1988}
}

@ARTICLE{DerSimonian1986-gc,
  title        = {Meta-analysis in clinical trials},
  author       = {DerSimonian, Rebecca and Laird, Nan},
  journaltitle = {Control. Clin. Trials},
  volume       = {7},
  number        = {3},
  pages        = {177--188},
  year         = {1986}
}

@ARTICLE{Desender2017-bp,
  title        = {A multicentre trial of patient specific rehearsal prior to {EVAR}: impact on procedural planning and team performance},
  author       = {Desender, L. and Van Herzeele, I. and Lachat, M. and Duchateau, J. and Bicknell, C. and Teijink, J. and Heyligers, J. and Vermassen, F. and {PAVLOV Study Group}},
  journaltitle = {Eur. J. Vasc. Endovasc. Surg.},
  volume       = {53},
  number        = {3},
  pages        = {354--361},
  year         = {2017}
}

@ARTICLE{Dmytriw2021-zz,
  title        = {Outcomes of acute respiratory distress syndrome in {COVID}-19 patients compared to the general population: a systematic review and meta-analysis},
  author       = {Dmytriw, Adam A. and Chibbar, Richa and Chen, Petty Pin Yu and Traynor, Michael D. and Kim, Dong Wook and Bruno, Fernando P. and Cheung, Christopher C. and Pareek, Anuj and Chou, Andrew Chia Chen and Graham, Jeffrey and Dibas, Mahmoud and Paranjape, Geeta and Reierson, Natalie L. and Kamrowski, Shelby and Rozowsky, Jacob and Barrett, Averi and Schmidt, Megan and Shahani, Disha and Cowie, Kathryn and Davis, Amber R. and Abdelmegeed, Mohamed and Touchette, Jillienne C. and Kallmes, Kevin M. and Pederson, John M. and Keesari, Praneeth Reddy},
  journaltitle = {Expert Rev. Respir. Med.},
  volume       = {15},
  number        = {10},
  pages        = {1347--1354},
  year         = {2021}
}

@ARTICLE{Frokjaer2006-hh,
  title        = {Autonomic dysfunction and impaired cerebral autoregulation in cirrhosis},
  author       = {Frøkjaer, Vibe G. and Strauss, Gitte I. and Mehlsen, Jesper and Knudsen, Gitte M. and Rasmussen, Verner and Larsen, Fin S.},
  journaltitle = {Clin. Auton. Res.},
  volume       = {16},
  number        = {3},
  pages        = {208--216},
  year         = {2006}
}

@ARTICLE{Haycock2010-dk,
  title        = {Training and transfer of colonoscopy skills: a multinational, randomized, blinded, controlled trial of simulator versus bedside training},
  author       = {Haycock, Adam and Koch, Arjun D. and Familiari, Pietro and {van Delft}, Foke and Dekker, Evelien and Petruzziello, Lucio and Haringsma, Jelle and Thomas-Gibson, Siwan},
  journaltitle = {Gastrointest. Endosc.},
  volume       = {71},
  number        = {2},
  pages        = {298--307},
  year         = {2010}
}

@BOOK{Hedges1985-os,
  title     = {Statistical methods for meta-analysis},
  author    = {Hedges, Larry V. and Olkin, Ingram},
  publisher = {Academic Press},
  location  = {San Diego, CA},
  year      = {1985},
  pagetotal = {369}
}

@ARTICLE{Hedges1998-oq,
  title        = {Fixed- and random-effects models in meta-analysis},
  author       = {Hedges, Larry V. and Vevea, Jack L.},
  journaltitle = {Psychol. Methods},
  volume       = {3},
  number        = {4},
  pages        = {486--504},
  year         = {1998}
}

@ARTICLE{Higgins2002-xr,
  title        = {Quantifying heterogeneity in a meta-analysis},
  author       = {Higgins, Julian P. T. and Thompson, Simon G.},
  journaltitle = {Stat. Med.},
  volume       = {21},
  number        = {11},
  pages        = {1539--1558},
  year         = {2002}
}

@ARTICLE{Higgins2003-dz,
  title        = {Measuring inconsistency in meta-analyses},
  author       = {Higgins, Julian P. T. and Thompson, Simon G. and Deeks, Jonathan J. and Altman, Douglas G.},
  journaltitle = {BMJ},
  volume       = {327},
  number        = {7414},
  pages        = {557--560},
  year         = {2003}
}

@BOOK{Higgins2019-mn,
  title     = {Cochrane handbook for systematic reviews of interventions},
  editor    = {Higgins, Julian P. T. and Thomas, James and Chandler, Jacqueline and Cumpston, Miranda and Li, Tianjing and Page, Matthew J. and Welch, Vivian A.},
  publisher = {Wiley-Blackwell},
  location  = {Hoboken, NJ},
  edition   = {2},
  year      = {2019}
}

@ONLINE{Higgins2024-eb,
  editor = {Higgins, J. P. T. and Thomas, J. and Chandler, J. and Cumpston, M. and Li, T. and Page, M. J. and Welch, V. A.},
  title  = {Cochrane handbook for systematic reviews of interventions},
  note   = {Version 6.5, updated August 2024. Accessed March 15, 2026},
  year   = {2024},
  url    = {https://www.cochrane.org/handbook}
}

@ARTICLE{Hozo2005-ik,
  title        = {Estimating the mean and variance from the median, range, and the size of a sample},
  author       = {Hozo, Stela Pudar and Djulbegovic, Benjamin and Hozo, Iztok},
  journaltitle = {BMC Med. Res. Methodol.},
  volume       = {5},
  pages        = {13},
  year         = {2005}
}

@ARTICLE{Katzenschlager2021-ol,
  title        = {Can we predict the severe course of {COVID}-19 - a systematic review and meta-analysis of indicators of clinical outcome?},
  author       = {Katzenschlager, Stephan and Zimmer, Alexandra J. and Gottschalk, Claudius and Grafeneder, Jürgen and Schmitz, Stephani and Kraker, Sara and Ganslmeier, Marlene and Muth, Amelie and Seitel, Alexander and Maier-Hein, Lena and Benedetti, Andrea and Larmann, Jan and Weigand, Markus A. and McGrath, Sean and Denkinger, Claudia M.},
  journaltitle = {PLoS One},
  volume       = {16},
  number        = {7},
  pages        = {e0255154},
  year         = {2021}
}

@ARTICLE{Kwon2015-se,
  title        = {Simulation-based estimation of mean and standard deviation for meta-analysis via Approximate Bayesian Computation ({ABC})},
  author       = {Kwon, Deukwoo and Reis, Isildinha M.},
  journaltitle = {BMC Med. Res. Methodol.},
  volume       = {15},
  number        = {1},
  pages        = {61},
  year         = {2015}
}

@ARTICLE{Langhorne2017-tw,
  title        = {Early supported discharge services for people with acute stroke},
  author       = {Langhorne, Peter and Baylan, Satu and {Early Supported Discharge Trialists}},
  journaltitle = {Cochrane Database Syst. Rev.},
  volume       = {7},
  number        = {7},
  pages        = {CD000443},
  year         = {2017}
}

@BOOK{Lehmann1998-or,
  title     = {Theory of point estimation},
  author    = {Lehmann, Erich L. and Casella, George},
  publisher = {Springer},
  location  = {New York, NY},
  edition   = {2},
  year      = {1998}
  % , series    = {Springer Texts in Statistics}
}

@ARTICLE{Luo2018-ig,
  title        = {Optimally estimating the sample mean from the sample size, median, mid-range, and/or mid-quartile range},
  author       = {Luo, Dehui and Wan, Xiang and Liu, Jiming and Tong, Tiejun},
  journaltitle = {Stat. Methods Med. Res.},
  volume       = {27},
  number        = {6},
  pages        = {1785--1805},
  year         = {2018}
}

@ARTICLE{Marin-Martinez2010-sd,
  title        = {Weighting by inverse variance or by sample size in random-effects meta-analysis},
  author       = {Marín-Martínez, Fulgencio and Sánchez-Meca, Julio},
  journaltitle = {Educ. Psychol. Meas.},
  volume       = {70},
  number        = {1},
  pages        = {56--73},
  year         = {2010}
}

@ARTICLE{McGrath2020-gt,
  title        = {Meta-analysis of the difference of medians},
  author       = {McGrath, Sean and Sohn, Hojoon and Steele, Russell and Benedetti, Andrea},
  journaltitle = {Biom. J.},
  volume       = {62},
  number        = {1},
  pages        = {69--98},
  year         = {2020}
}

@ARTICLE{Nagasako2009-hb,
  title        = {Investigation of autonomic function and orocecal transit time in patients with nonalcoholic cirrhosis and the potential influence of these factors on disease outcome},
  author       = {Nagasako, Cristiane Kibune and {de Oliveira Figueiredo}, M. J. and {de Souza Almeida}, J. R. and Lorena, S. and Akasaka, Helena Midori Imamura and Pavan, C. R. and Yamanaka, Ademar and Pereira, T. and Soares, E. and Mesquita, M.},
  journaltitle = {J. Clin. Gastroenterol.},
  volume       = {43},
  pages        = {884--889},
  year         = {2009}
}

@ARTICLE{Nakagawa2023-fm,
  title        = {Quantitative evidence synthesis: a practical guide on meta-analysis, meta-regression, and publication bias tests for environmental sciences},
  author       = {Nakagawa, Shinichi and Yang, Yefeng and Macartney, Erin L. and Spake, Rebecca and Lagisz, Malgorzata},
  journaltitle = {Environ. Evid.},
  volume       = {12},
  number        = {1},
  pages        = {8},
  year         = {2023}
}

@ARTICLE{Page2021-wh,
  title        = {The {PRISMA} 2020 statement: an updated guideline for reporting systematic reviews},
  author       = {Page, Matthew J. and McKenzie, Joanne E. and Bossuyt, Patrick M. and Boutron, Isabelle and Hoffmann, Tammy C. and Mulrow, Cynthia D. and Shamseer, Larissa and Tetzlaff, Jennifer M. and Akl, Elie A. and Brennan, Sue E. and Chou, Roger and Glanville, Julie and Grimshaw, Jeremy M. and Hróbjartsson, Asbjørn and Lalu, Manoj M. and Li, Tianjing and Loder, Elizabeth W. and Mayo-Wilson, Evan and McDonald, Steve and McGuinness, Luke A. and Stewart, Lesley A. and Thomas, James and Tricco, Andrea C. and Welch, Vivian A. and Whiting, Penny and Moher, David},
  journaltitle = {Syst. Rev.},
  volume       = {10},
  number        = {1},
  pages        = {89},
  year         = {2021}
}

@ARTICLE{Saldanha2020-va,
  title        = {Outcome choice and definition in systematic reviews leads to few eligible studies included in meta-analyses: a case study},
  author       = {Saldanha, Ian J. and Lindsley, Kristina B. and Money, Sarah and Kimmel, Hannah J. and Smith, Bryant T. and Dickersin, Kay},
  journaltitle = {BMC Med. Res. Methodol.},
  volume       = {20},
  number        = {1},
  pages        = {30},
  year         = {2020}
}

@ARTICLE{Sanchez-Meca1998-ss,
  title        = {Weighting by inverse variance or by sample size in meta-analysis: a simulation study},
  author       = {Sánchez-Meca, Julio and Marín-Martínez, Fulgencio},
  journaltitle = {Educ. Psychol. Meas.},
  volume       = {58},
  number        = {2},
  pages        = {211--220},
  year         = {1998}
}

@ARTICLE{Shore2016-la,
  title        = {Validating a standardized laparoscopy curriculum for gynecology residents: a randomized controlled trial},
  author       = {Shore, Eliane M. and Grantcharov, Teodor P. and Husslein, Heinrich and Shirreff, Lindsay and Dedy, Nicolas J. and McDermott, Colleen D. and Lefebvre, Guylaine G.},
  journaltitle = {Am. J. Obstet. Gynecol.},
  volume       = {215},
  number        = {2},
  pages        = {204.e1--204.e11},
  year         = {2016}
}

@ARTICLE{Sutton2008-rz,
  title        = {Recent developments in meta-analysis},
  author       = {Sutton, Alexander J. and Higgins, Julian P. T.},
  journaltitle = {Stat. Med.},
  volume       = {27},
  number        = {5},
  pages        = {625--650},
  year         = {2008}
}

@ARTICLE{Viechtbauer2005-kl,
  title        = {Bias and efficiency of meta-analytic variance estimators in the random-effects model},
  author       = {Viechtbauer, Wolfgang},
  journaltitle = {J. Educ. Behav. Stat.},
  volume       = {30},
  number        = {3},
  pages        = {261--293},
  year         = {2005}
}

@ARTICLE{Wan2014-yu,
  title        = {Estimating the sample mean and standard deviation from the sample size, median, range and/or interquartile range},
  author       = {Wan, Xiang and Wang, Wenqian and Liu, Jiming and Tong, Tiejun},
  journaltitle = {BMC Med. Res. Methodol.},
  volume       = {14},
  number        = {1},
  pages        = {135},
  year         = {2014}
}

@ARTICLE{Weir2018-vv,
  title        = {Dealing with missing standard deviation and mean values in meta-analysis of continuous outcomes: a systematic review},
  author       = {Weir, Christopher J. and Butcher, Isabella and Assi, Valentina and Lewis, Stephanie C. and Murray, Gordon D. and Langhorne, Peter and Brady, Marian C.},
  journaltitle = {BMC Med. Res. Methodol.},
  volume       = {18},
  number        = {1},
  pages        = {25},
  year         = {2018}
}

@ARTICLE{White2011-tz,
  title        = {Costs and cost-effectiveness of malaria control interventions: a systematic review},
  author       = {White, Michael T. and Conteh, Lesong and Cibulskis, Richard and Ghani, Azra C.},
  journaltitle = {Malar. J.},
  volume       = {10},
  number        = {1},
  pages        = {337},
  year         = {2011}
}

@ARTICLE{Wiebe2006-if,
  title        = {A systematic review identifies a lack of standardization in methods for handling missing variance data},
  author       = {Wiebe, Natasha and Vandermeer, Ben and Platt, Robert W. and Klassen, Terry P. and Moher, David and Barrowman, Nicholas J.},
  journaltitle = {J. Clin. Epidemiol.},
  volume       = {59},
  number        = {4},
  pages        = {342--353},
  year         = {2006}
}

\FloatBarrier
\newpage
\begin{table}[t]
  \centering
  \caption{Data generation parameters.}\label{tab:sim_design}
  \renewcommand{\arraystretch}{1.15}
  \begin{threeparttable}
      \begin{tabular}{@{}ll@{}}
        \toprule
        \textbf{Parameter}              & \textbf{Values} \\
        \midrule
        Number of studies ($N$)      & 10, 30 \\
        Per-group sample size ($n$)  & Fixed: 100, 300;\quad Varying: average per-group size of 100, 300 \\
        Outcome distribution         & normal, skew-normal, log-normal \\
        Heterogeneity ($I^{2}$)      & 0\% (None, $\tau^2=0$), 25\% (Low), 50\% (Medium), 75\% (High) \\
        \bottomrule
      \end{tabular}
        \begin{tablenotes}[flushleft]\footnotesize
          \item \textit{Note.} Across all outcome families, the data-generating parameters encode a nonzero difference in group medians; see Supplementary Table~S1 for the exact parameterizations.
          $I^2$ denotes the proportion (0--1 scale) of total variance due to between-study heterogeneity. Using the unequal-size correction\supercite{Higgins2002-xr,Higgins2003-dz}, we set
          $\tau^2=\left\{I^2/(1 - I^2)\right\}\,s^2_{\text{typical}}$
          with
          $s^2_{\text{typical}} = (N\!-\!1)\sum_{i=1}^N w_i / \{(\sum_{i=1}^N w_i)^2-\sum_{i=1}^N w_i^2\}$,
          where $w_i = 1/\sigma_i^2$.
        \end{tablenotes}
  \end{threeparttable}
\end{table}

\clearpage
\begin{table}[t]
  \centering
  \caption{Performance of point and variance estimators under a log-normal scenario ($N=30$, varying sample sizes with average $n=100$).}
  \label{tab:res_logN30n100}
  \renewcommand{\arraystretch}{1.15}
  \begin{threeparttable}
    \begin{tabular}{@{}cc*{8}{S}@{}}
      \toprule
      &
      & \multicolumn{2}{c}{\textbf{Point}}
      & \multicolumn{2}{c}{\textbf{Variance}}
      & \multicolumn{2}{c}{\textbf{$z$-based}}
      & \multicolumn{2}{c}{\textbf{$t$-based}} \\
      \cmidrule(lr){3-4}\cmidrule(lr){5-6}\cmidrule(lr){7-8}\cmidrule(lr){9-10}
      \textbf{$I^2$}
      & \textbf{Method}
      & \textbf{\%Bias} & \textbf{\%MSE}
      & \textbf{\%Bias} & \textbf{\%MSE}
      & \textbf{CP}    & \textbf{AW}
      & \textbf{CP}    & \textbf{AW} \\
      \midrule
      $0\%$   & DiVE    &   4.440 & 0.800 &   8.133 & 14.803 & 0.914 & 3.755 & 0.922 & 3.919 \\
              & QE--RE  &  -7.963 & 1.196 &   4.937 &  4.889 & 0.797 & 3.735 & 0.817 & 3.898 \\
              & QE--FE  & -9.263 & 1.439 & -7.663 &  1.315 & 0.712 & 3.517 & 0.743 & 3.670 \\
      \midrule
      $25\%$  & DiVE    &   5.029 & 1.037 &  6.967 & 13.367 & 0.918 & 4.447 & 0.931 & 4.640 \\
              & QE--RE  &  -5.579 & 1.049 &  -0.521 &  8.264 & 0.875 & 4.265 & 0.886 & 4.451 \\
              & QE--FE  & -9.088 & 1.616 & -34.450 & 12.217 & 0.695 & 3.525 & 0.718 & 3.679 \\
      \midrule
      $50\%$  & DiVE    &   4.757 & 1.553 &   5.186 & 15.615 & 0.921 & 5.543 & 0.934 & 5.785 \\
              & QE--RE  &  -2.810 & 1.290 &  -0.184 &  8.531 & 0.920 & 5.258 & 0.928 & 5.487 \\
              & QE--FE  & -9.065 & 2.213 & -58.622 & 34.529 & 0.629 & 3.527 & 0.653 & 3.681 \\
      \midrule
      $75\%$  & DiVE    &   4.750 & 2.741 &   2.134 & 13.829 & 0.936 & 7.941 & 0.951 & 8.286 \\
              & QE--RE  &  0.590 & 2.165 &   1.743 &  8.810 & 0.947 & 7.470 & 0.957 & 7.795 \\
              & QE--FE  & -9.020 & 3.526 & -80.442 & 64.752 & 0.542 & 3.524 & 0.559 & 3.677 \\
      \bottomrule
    \end{tabular}
    \begin{tablenotes}[flushleft]\footnotesize
        \item \textit{Note.} Results are based on 1,000 simulation replicates per heterogeneity level.
        \%Bias, relative bias; \%MSE, relative mean-squared error; CP, empirical coverage probability of 95\% confidence intervals; AW, average width.
        $z$-based intervals use the standard normal critical value, while $t$-based intervals use quantiles from a $t$-distribution with $N\!-\!1$ degrees of freedom.
        The true pooled difference is the benchmark for point-estimate metrics; variance metrics are evaluated against method-specific analytic targets defined in \Cref{app:metrics}.
    \end{tablenotes}
  \end{threeparttable}
\end{table}

\clearpage
\begin{table}[t]
  \centering
  \small
  \begin{threeparttable}
    \caption[ESD trials for acute stroke: hospital-stay inputs]{Study-level inputs for the real-data application: randomized trials of early supported discharge (ESD) for acute stroke, with group medians and sample sizes for length of initial hospital stay and quartiles where available.}
    \label{tab:esd_inputs}
    \sisetup{
      table-number-alignment = center,
      input-open-uncertainty = ,
      input-close-uncertainty = ,
      detect-weight = true,
      detect-family = true
    }
    \begin{tabular}{@{}l
      S[table-format=3.0]
      S[table-format=3.1]
      S[table-format=3.1]
      S[table-format=3.1]
      S[table-format=3.0]
      S[table-format=3.1]
      S[table-format=3.1]
      S[table-format=3.1]
      S[table-format=1.3]
      S[table-format=1.3]
      @{}}
      \toprule
      & \multicolumn{4}{c}{\textbf{Group 1 (ESD)}} & \multicolumn{4}{c}{\textbf{Group 2 (conventional)}} & \\
      \cmidrule(lr){2-5}\cmidrule(lr){6-9}
      \textbf{Study ID}
      & \textbf{$n_{i1}$} & \textbf{Median} & \textbf{Q1} & \textbf{Q3}
      & \textbf{$n_{i2}$} & \textbf{Median} & \textbf{Q1} & \textbf{Q3}
      & {$\tilde w_i^{\rm s}$} & {$\tilde w_i^{\ast}$} \\
      \midrule
      {\textsuperscript{\dag}Adelaide~2000}
        &  42 & 15.0 &  8.0 & 22.0
        &  44 & 30.0 & 17.3 & 48.5
        & 0.098 & 0.466 \\
      Adelaide~2016
        &  31 & 16.0 & \multicolumn{1}{c}{\textit{NA}} & \multicolumn{1}{c}{\textit{NA}}
        &  32 & 20.0 & \multicolumn{1}{c}{\textit{NA}} & \multicolumn{1}{c}{\textit{NA}}
        & 0.072 & \multicolumn{1}{c}{\textit{NA}} \\
      Belfast~2004
        &  59 & 31.0 & \multicolumn{1}{c}{\textit{NA}} & \multicolumn{1}{c}{\textit{NA}}
        &  54 & 32.0 & \multicolumn{1}{c}{\textit{NA}} & \multicolumn{1}{c}{\textit{NA}}
        & 0.129 & \multicolumn{1}{c}{\textit{NA}} \\
      {\textsuperscript{\dag}Copenhagen~2009}
        &  31 & 18.0 & 16.0 & 21.0
        &  30 & 16.0 & 12.0 & 21.0
        & 0.070 & 0.534 \\
      London~1997
        & 167 &  6.0 & \multicolumn{1}{c}{\textit{NA}} & \multicolumn{1}{c}{\textit{NA}}
        & 164 & 12.0 & \multicolumn{1}{c}{\textit{NA}} & \multicolumn{1}{c}{\textit{NA}}
        & 0.378 & \multicolumn{1}{c}{\textit{NA}} \\
      Oslo~2000
        &  42 & 22.0 & \multicolumn{1}{c}{\textit{NA}} & \multicolumn{1}{c}{\textit{NA}}
        &  40 & 31.0 & \multicolumn{1}{c}{\textit{NA}} & \multicolumn{1}{c}{\textit{NA}}
        & 0.094 & \multicolumn{1}{c}{\textit{NA}} \\
      Stockholm~1998
        &  39 & 12.0 & \multicolumn{1}{c}{\textit{NA}} & \multicolumn{1}{c}{\textit{NA}}
        &  38 & 23.0 & \multicolumn{1}{c}{\textit{NA}} & \multicolumn{1}{c}{\textit{NA}}
        & 0.088 & \multicolumn{1}{c}{\textit{NA}} \\
      Trondheim~2004
        &  31 & 12.0 & \multicolumn{1}{c}{\textit{NA}} & \multicolumn{1}{c}{\textit{NA}}
        &  31 & 10.0 & \multicolumn{1}{c}{\textit{NA}} & \multicolumn{1}{c}{\textit{NA}}
        & 0.071 & \multicolumn{1}{c}{\textit{NA}} \\
      \bottomrule
    \end{tabular}
    \begin{tablenotes}[flushleft]\footnotesize
      \item \textit{Source.} Data were extracted from the 2017 update of the Cochrane Review of early supported discharge services for people with acute stroke (CD000443)\supercite{Langhorne2017-tw}.
      \item[\dag] Rows marked with \dag{} report medians and IQRs in both groups and were eligible for QE--RE; all rows were eligible for DiVE because only medians and group sizes are required.
      \item \textit{Notes.} Values are on the native ``days'' scale as in the original sources and the Cochrane review; ``\textit{NA}'' denotes not reported. Quartiles are first (Q1) and third (Q3). The study-level effect in the main text is the ESD median minus the control median (negative favors ESD). DiVE uses the sample-size weights $\tilde w_i^{\rm s}$ defined in \Cref{eq:prep_estmu}; the maximum weight was 0.378 (London~1997). QE--RE weights $\tilde w_i^{\ast}$ are reported only for studies with medians and IQRs in both arms, where $\tilde w_i^{\ast} := w_i^{\ast}/{\sum_{\ell=1}^N w_\ell^{\ast}}$ and $w_i^{\ast}$ follows the IVW form in \Cref{eq:re_est} with $\tau^2$ replaced by its DerSimonian--Laird estimate; remaining entries are \textit{NA}.
    \end{tablenotes}
  \end{threeparttable}
\end{table}
\clearpage
\begin{table}[t]
  \centering
  \caption{Real-data application: meta-analysis of initial hospital stay (days) using DiVE and QE--RE (ESD vs.\ conventional care).}
  \label{tab:example_esd}
  \renewcommand{\arraystretch}{1.15}
  \begin{threeparttable}
    \begin{tabular}{@{}l
                    S[table-format=1.0]
                    S[table-format=3.0]
                    S[table-format=-2.2]
                    c
                    c
                    c
                    c@{}}
      \toprule
      & & & & \multicolumn{2}{c}{\textbf{$z$-based}} & \multicolumn{2}{c}{\textbf{$t$-based}} \\
      \cmidrule(lr){5-6}\cmidrule(lr){7-8}
      \textbf{Method}
        & $N$
        & $n_{\text{total}}$
        & \textbf{\shortstack{Estimate\\(days)}}
        & \textbf{95\% CI}
        & \textbf{$p$-value}
        & \textbf{95\% CI}
        & \textbf{$p$-value} \\
      \midrule
      DiVE   & 8 & 875 & -5.69 & [$-8.18$,\,$-3.20$]   & $< 0.001$ & [$-8.69$,\,$-2.69$]    & 0.003 \\
      QE--RE & 2 & 147 & -5.92 & [$-22.54$,\,$10.70$]  & 0.485     & [$-113.67$,\,$101.83$] & 0.612 \\
      \bottomrule
    \end{tabular}
    \begin{tablenotes}[flushleft]\footnotesize
        \item \textit{Note.} Estimate is the pooled median difference (ESD minus conventional care); negative values favor ESD (shorter stay).
        $n_{\text{total}}$ denotes the total sample size across all studies included by each method.
        Reported $p$-values are two-sided for the null hypothesis of no pooled median difference.
        $z$-based intervals and $p$-values use the standard normal reference distribution; $t$-based intervals and $p$-values use $N\!-\!1$ degrees of freedom.
        The analyzable sets differ by method: DiVE includes studies reporting group medians and sample sizes in both arms ($N=8$), whereas QE--RE is limited to studies reporting medians and interquartile ranges in both arms ($N=2$).
        DiVE uses sample-size weights, and QE--RE uses IVW with a DerSimonian--Laird RE component.
        QE--FE was not included in this application.
    \end{tablenotes}
  \end{threeparttable}
\end{table}

\clearpage

\FloatBarrier
\newpage
\appendix
\counterwithin{theorem}{section} 
\renewcommand{\thetheorem}{\thesection.\arabic{theorem}} 

\section{Direct estimator for $\operatorname{Var}(\hat{\mu}^{\rm s})$: derivation and proof}\label{app:dive}
\subsection{Notation and assumptions}\label{app:dive_notation}
Recall $\hat{\mu}^{\rm s} = \sum_{i=1}^N \tilde w_i^{\rm s} Y_i$ with $\tilde w_i^{\rm s} = n_i / \sum_{\ell=1}^N n_\ell$.
For \Cref{app:dive_derivation}, we work under \Cref{mainass:A1,mainass:A2}.
For \Cref{app:dive_proof}, we work under \Cref{mainass:A1,mainass:A2prime,mainass:A4}.
In particular, \Cref{mainass:A4} implies $\max_i \tilde w_i^{\rm s} < \tfrac{1}{2}$ for sufficiently large $n_{\min}$, so that $h_i=(\tilde w_i^{\rm s})^2/(1-2\tilde w_i^{\rm s})$ are well-defined.
Recall that $n_{\min,i} = \min(n_{i1},n_{i2})$ and $n_{\min} = \min_i n_{\min,i}$.
For sample median differences, under standard regularity conditions, $E[Y_i]-\mu = O(n_{\min,i}^{-1})$ as $n_{\min,i}\to\infty$\supercite{Bahadur1966-fy,Lehmann1998-or}.

\subsection{Proof of Proposition~\ref{prop:var_exact_unbiased}}\label{app:dive_derivation}
\begin{proof}
To prove \Cref{prop:var_exact_unbiased}, we consider the quantity
\begin{equation*}
  T \;=\; \sum_{i=1}^N k_i \left( Y_i - \hat{\mu}^{\rm s} \right)^2,
\end{equation*}
with constants $k_i$ (for $i=1,\ldots,N$) to be determined.
We then choose these constants so that $T$ is unbiased for $\operatorname{Var}(\hat{\mu}^{\rm s})$.

Under \Cref{mainass:A1,mainass:A2}, we have
\begin{equation}
  E\!\left[ \left( Y_i-\hat{\mu}^{\rm s} \right)^2 \right]
  = E\!\left[ (Y_i-\mu)^2 \right]
     - 2E\!\left[ (Y_i-\mu)(\hat{\mu}^{\rm s}-\mu) \right]
     + E\!\left[ (\hat{\mu}^{\rm s}-\mu)^2 \right].
  \label{eq:dive_expand}
\end{equation}
We evaluate the three terms on the right-hand side of \Cref{eq:dive_expand} in turn.

First, since \Cref{mainass:A2} gives $E[Y_i]=\mu$, we see that
\begin{equation*}
  E\!\left[ (Y_i-\mu)^2 \right] = \operatorname{Var}(Y_i).
\end{equation*}

Second, we obtain
\begin{align*}
  E\!\left[ (Y_i-\mu)(\hat{\mu}^{\rm s}-\mu) \right]
  &= E\!\left[ (Y_i-\mu)\sum_{\ell=1}^{N} \tilde w_\ell^{\rm s}(Y_\ell-\mu) \right] \\
  &= \tilde w_i^{\rm s}E\!\left[ (Y_i-\mu)^2 \right]
     + \sum_{\ell\ne i} \tilde w_\ell^{\rm s}E\!\left[ (Y_i-\mu)(Y_\ell-\mu) \right] \\
  &= \tilde w_i^{\rm s}\operatorname{Var}(Y_i),
\end{align*}
where the last equality follows from \Cref{mainass:A1}.

Finally, we have
\begin{equation*}
  E\!\left[ (\hat{\mu}^{\rm s}-\mu)^2 \right] = \operatorname{Var}(\hat{\mu}^{\rm s}) = \sum_{i=1}^{N} (\tilde w_i^{\rm s})^2 \operatorname{Var}(Y_i).
\end{equation*}

Hence, using the above results, we obtain
\begin{align*}
  E[T]
  &= \sum_{i=1}^N k_i\,E\!\left[ \left( Y_i-\hat{\mu}^{\rm s} \right)^2 \right] \\
  &= \sum_{i=1}^N \left\{ k_i(1-2\tilde w_i^{\rm s}) + K(\tilde w_i^{\rm s})^2 \right\}\operatorname{Var}(Y_i),
\end{align*}
where $K=\sum_{\ell=1}^N k_\ell$.
To make $E[T]$ equal the variance of the pooled estimator in \Cref{eq:var_exact}, we therefore choose $k_i$ so that
\begin{equation*}
  k_i(1-2\tilde w_i^{\rm s}) + K(\tilde w_i^{\rm s})^2
  = (\tilde w_i^{\rm s})^2,
  \qquad i=1,\ldots,N.
\end{equation*}
Equivalently,
\begin{equation*}
  k_i = (1-K)\frac{(\tilde w_i^{\rm s})^2}{1-2\tilde w_i^{\rm s}}.
\end{equation*}
Define
\begin{equation*}
  h_i = \frac{(\tilde w_i^{\rm s})^2}{1-2\tilde w_i^{\rm s}},
  \qquad
  H = \sum_{\ell=1}^N h_\ell,
\end{equation*}
which are well-defined when $\max_i \tilde w_i^{\rm s}<\tfrac{1}{2}$.
Then
\begin{equation*}
  k_i = (1-K)h_i.
\end{equation*}
Summing the above equation over $i$ gives $K = (1-K)H$, and hence $K = H/(1+H)$.
Therefore, $k_i = h_i/(1+H)$.
Thus, we obtain
\begin{equation*}
  \widehat{\operatorname{Var}}(\hat{\mu}^{\rm s})
  = T
  = \sum_{i=1}^N \frac{h_i}{1+\sum_{\ell=1}^N h_\ell}\left( Y_i-\hat{\mu}^{\rm s} \right)^2,
\end{equation*}
which is exactly \Cref{eq:estvar_dive}.
By construction,
\begin{equation*}
  E\!\left[\widehat{\operatorname{Var}}(\hat{\mu}^{\rm s})\right]
  = \operatorname{Var}(\hat{\mu}^{\rm s}),
\end{equation*}
proving \Cref{prop:var_exact_unbiased}.
\end{proof}

\subsection{Proof of \Cref{prop:var_unbiased}}\label{app:dive_proof}
\begin{proof}
Under \Cref{mainass:A1,mainass:A2prime,mainass:A4}, let $b_i := E[Y_i] - \mu$, and define $\hat{\mu}^{\rm s\ast} = \sum_{i=1}^N \tilde w_i^{\rm s\ast}Y_i$.
Also define
\begin{equation*}
  h_i^\ast = \frac{(\tilde w_i^{\rm s\ast})^2}{1-2\tilde w_i^{\rm s\ast}},
\end{equation*}
which is well-defined and finite.

By expanding the square and taking expectations,
\begin{align*}
  E[\widehat{\operatorname{Var}}(\hat{\mu}^{\rm s})]
  &= \sum_{i=1}^{N} \frac{h_i}{\,1 + \sum_{\ell=1}^{N} h_\ell\,}\,
     E\!\left[ \left( Y_i-\hat{\mu}^{\rm s} \right)^2 \right] \\
  &= \sum_{i=1}^{N} \frac{h_i}{\,1 + \sum_{\ell=1}^{N} h_\ell\,}\,
     \left\{ E\!\left[ (Y_i - \mu)^2 \right]
     - 2E\!\left[ (Y_i - \mu)(\hat{\mu}^{\rm s} - \mu) \right]
     + E\!\left[ (\hat{\mu}^{\rm s} - \mu)^2 \right] \right\}.
\end{align*}

Independence and the vanishing bias $b_i \to 0$ give
\begin{equation*}
  E\!\left[ (Y_i-\mu)^2 \right]
  = \operatorname{Var}(Y_i) + \left(E[Y_i-\mu]\right)^2
  = \operatorname{Var}(Y_i) + b_i^2
  \;\longrightarrow\;\operatorname{Var}(Y_i),
\end{equation*}
\begin{align*}
  E\!\left[ (Y_i-\mu)(\hat{\mu}^{\rm s}-\mu) \right]
  &= E\!\left[ (Y_i-\mu)\sum_{\ell=1}^{N} \tilde w_\ell^{\rm s}(Y_\ell-\mu) \right] \\
  &= \tilde w_i^{\rm s}E\!\left[ (Y_i-\mu)^2 \right]
     + \sum_{\ell \ne i} \tilde w_\ell^{\rm s}E\!\left[ (Y_i-\mu)(Y_\ell-\mu) \right] \\
  &= \tilde w_i^{\rm s}\left\{ \operatorname{Var}(Y_i) + b_i^2 \right\}
     + \sum_{\ell \ne i} \tilde w_\ell^{\rm s} b_i b_\ell
  \;\longrightarrow\;\tilde w_i^{\rm s\ast}\operatorname{Var}(Y_i),
\end{align*}
and
\begin{equation*}
  E\!\left[ (\hat{\mu}^{\rm s}-\mu)^2 \right]
  = \operatorname{Var}(\hat{\mu}^{\rm s})
    + \left( \sum_{i=1}^N \tilde w_i^{\rm s} b_i \right)^2
  \;\longrightarrow\;\operatorname{Var}(\hat{\mu}^{\rm s\ast}).
\end{equation*}

Substituting these limits and collecting terms,
\begin{equation*}
  E[\widehat{\operatorname{Var}}(\hat{\mu}^{\rm s})]
  \;\longrightarrow\;
  \sum_{i=1}^{N} \frac{h_i^\ast}{\,1 + \sum_{\ell=1}^{N} h_\ell^\ast\,}
  \left\{ (1 - 2\tilde w_i^{\rm s\ast})\operatorname{Var}(Y_i)
  + \operatorname{Var}(\hat{\mu}^{\rm s\ast}) \right\}.
\end{equation*}
Since $h_i^\ast(1 - 2\tilde w_i^{\rm s\ast}) = (\tilde w_i^{\rm s\ast})^2$ by definition,
\begin{align*}
  E[\widehat{\operatorname{Var}}(\hat{\mu}^{\rm s})]
  \;\longrightarrow&\;
  \frac{1}{\,1 + \sum_{\ell=1}^{N} h_\ell^\ast\,}
  \left\{
    \sum_{i=1}^{N} (\tilde w_i^{\rm s\ast})^2\operatorname{Var}(Y_i)
    + \sum_{i=1}^{N} h_i^\ast \operatorname{Var}(\hat{\mu}^{\rm s\ast})
  \right\} \\
  =&\;
  \frac{1}{\,1 + \sum_{\ell=1}^{N} h_\ell^\ast\,}
  \left( 1 + \sum_{i=1}^{N} h_i^\ast \right)
  \operatorname{Var}(\hat{\mu}^{\rm s\ast}) \\
  =&\;
  \operatorname{Var}(\hat{\mu}^{\rm s\ast}).
\end{align*}

Finally, since $\tilde w_i^{\rm s}\to \tilde w_i^{\rm s\ast}$ and $N$ is fixed, we have
\begin{equation*}
  \operatorname{Var}(\hat{\mu}^{\rm s})-\operatorname{Var}(\hat{\mu}^{\rm s\ast}) \to 0
\end{equation*}
under the standing finite-second-moment condition.
Hence
\begin{equation*}
  E[\widehat{\operatorname{Var}}(\hat{\mu}^{\rm s})]\;\longrightarrow\;\operatorname{Var}(\hat{\mu}^{\rm s})
\end{equation*}
as $n_{\min}\to\infty$.
\end{proof}

\FloatBarrier
\vspace{3em}
\section{Performance Metrics}\label{app:metrics}
We evaluate point estimators against the common truth $\mu$, while variance estimators are judged against method-specific analytic targets defined below.
Because $\mu\neq 0$ by design in all settings, all relative metrics are well-defined; no zero-denominator handling is required in the reported results.
This appendix formalizes the metrics used in the simulation study in \Cref{sec:sim}.
Let $R$ denote the number of replicates per design setting.

Within a replicate and study $i$, let $\sigma_i^2$ be the theoretical variance of the study-level median difference, computed from the population densities at the group-specific medians and the realized sample sizes (standard large-sample formula; see McGrath et al.\supercite{McGrath2020-gt}).
Given a target heterogeneity $I^2$, $\tau^2$ is set by the unequal-size correction\supercite{Higgins2002-xr,Higgins2003-dz}.

Let $w_i^{\rm s}$ denote the raw DiVE study weights, and $\tilde w_i^{\rm s}=w_i^{\rm s}/\sum_{\ell=1}^N w_{\ell}^{\rm s}$ their normalized form (they sum to one).
Let $w_i = 1/\sigma_i^2$ denote the FE raw inverse-variance weights and $\tilde w_i = w_i / \sum_{\ell=1}^N w_\ell$ their normalized form.
The method-specific analytic variance targets used for evaluating variance estimators are
\begin{align*}
  \operatorname{Var}(\hat\mu^{\rm s}) &= \sum_{i=1}^N (\tilde w_i^{\rm s})^2\,(\sigma_i^2 + \tau^2), \\
  \operatorname{Var}(\hat\mu_{\text{RE}}) &= \left( \sum_{i=1}^N \frac{1}{\sigma_i^2 + \tau^2} \right)^{\!-1}, \\
  \operatorname{Var}(\hat\mu_{\text{FE}}) &= \sum_{i=1}^N \tilde w_i^{2}\,(\sigma_i^2 + \tau^2).
\end{align*}
For the $r$th replicate, let $V_{\mathrm{target},r}$ denote the analytic variance target corresponding to the method under evaluation.
These are the benchmarks against which $\widehat{\operatorname{Var}}(\hat\mu)$ is judged for each method.

\subsection*{Point Estimator Performance}\label{app:met_point}
For the pooled point estimator $\hat\mu$,
\begin{align*}
  \%\mathrm{Bias}(\hat\mu) &= 100 \cdot \frac{1}{R} \sum_{r=1}^R \frac{\hat\mu_{r} - \mu}{\mu}, \\
  \%\mathrm{MSE}(\hat\mu) &= 100 \cdot \frac{1}{R} \sum_{r=1}^R \left( \frac{\hat\mu_{r} - \mu}{\mu} \right)^{\!2} .
\end{align*}
Here, \%Bias denotes the relative bias and \%MSE the relative mean-squared error.

\subsection*{Variance Estimator Performance}\label{app:met_var}
For the variance estimator $\widehat{\operatorname{Var}}(\hat\mu)$,
\begin{align*}
  \%\mathrm{Bias}\left( \widehat{\operatorname{Var}}(\hat\mu) \right)
  &= 100 \cdot \frac{1}{R} \sum_{r=1}^R
  \frac{\widehat{\operatorname{Var}}(\hat\mu)_{r} - V_{\mathrm{target},r}}
       {V_{\mathrm{target},r}}, \\
  \%\mathrm{MSE}\left( \widehat{\operatorname{Var}}(\hat\mu) \right)
  &= 100 \cdot \frac{1}{R} \sum_{r=1}^R
  \left(
  \frac{\widehat{\operatorname{Var}}(\hat\mu)_{r} - V_{\mathrm{target},r}}
       {V_{\mathrm{target},r}}
  \right)^{\!2}.
\end{align*}
Both metrics are expressed as percentages relative to the replicate-specific analytic target $V_{\mathrm{target},r}$.

\subsection*{Confidence Interval (CI) Performance}\label{app:met_ci}
For 95\% confidence intervals (CIs) constructed for $\hat\mu$,
\begin{align*}
  \mathrm{CP} &= \frac{1}{R} \sum_{r=1}^{R} \mathbf{1}\!\left\{ \mathrm{LL}_{r} \le \mu \le \mathrm{UL}_{r} \right\},\\
  \mathrm{AW} &= \frac{1}{R} \sum_{r=1}^{R} \left( \mathrm{UL}_{r} - \mathrm{LL}_{r} \right).
\end{align*}
Here, CP denotes the coverage probability and AW the average width.
Unless otherwise specified, CIs use $t$-quantiles with $N\!-\!1$ degrees of freedom:
\begin{equation*}
  \mathrm{LL}_{r} = \hat\mu_{r} - t_{0.975,\,N\!-\!1}\sqrt{\widehat{\operatorname{Var}}(\hat\mu)_{r}},\quad
  \mathrm{UL}_{r} = \hat\mu_{r} + t_{0.975,\,N\!-\!1}\sqrt{\widehat{\operatorname{Var}}(\hat\mu)_{r}},
\end{equation*}
and $z$-based counterparts replace $t_{0.975,\,N\!-\!1}$ by $z_{0.975}$.

\subsection*{Replication-Level Metrics for Graphical Display}\label{app:met_error}
For replicate-wise summaries (e.g., boxplots), we report relative errors:
\begin{align*}
  \%\mathrm{Error}\left( \hat\mu \right)_r &= 100 \times \frac{\hat\mu_{r} - \mu}{\mu},\\
  \%\mathrm{Error}\left( \widehat{\operatorname{Var}}(\hat\mu) \right)_r &= 100 \times \frac{\widehat{\operatorname{Var}}(\hat\mu)_{r} - V_{\mathrm{target},r}}{V_{\mathrm{target},r}}.
\end{align*}

\subsubsection*{Notation.}
$\hat\mu_r$ and $\widehat{\operatorname{Var}}(\hat\mu)_r$ denote, respectively, the estimate of $\mu$ and the estimated variance of $\hat\mu$ in the $r$th replicate; $\mathrm{LL}_r,\mathrm{UL}_r$ are CI limits; $\mathbf{1}\{\cdot\}$ is the indicator function.

\newpage
\section*{Supplementary Material}
\beginsupplement

\begin{table}[h!]
  \centering
  \caption{Parameters for each data-generating distribution (two-group design).}
  \label{suptab:sim_design}
  \renewcommand{\arraystretch}{1.15}
  \begin{threeparttable}
  \begin{tabular}{@{}lccccc@{}}
    \toprule
    \ \textbf{Distribution} & \textbf{Group} & \textbf{Parameter 1} & \textbf{Parameter 2} & \textbf{Parameter 3} & \textbf{True Median}\ \\
    \midrule
    \ normal   
                          & 1 & $\text{Mean} = 5 + c$     & $\text{SD} = 1$     & --                   & 5.00 + $c$\ \\
                          & 2 & $\text{Mean} = 5$     & $\text{SD} = 1$     & --                   & 5.00\ \\
    \ skew-normal
                          & 1 & $\text{Location} = 5$ & $\text{Scale} = 5$  & $\text{Shape} = 5$   & 8.37\ \\
                          & 2 & $\text{Location} = 5$ & $\text{Scale} = 10$ & $\text{Shape} = 10$  & 11.74\ \\
    \ log-normal
                          & 1 & $\text{Meanlog} = 2$  & $\text{SDlog} = 1$  & --                   & 7.39\ \\
                          & 2 & $\text{Meanlog} = 3$  & $\text{SDlog} = 2$  & --                   & 20.09\ \\
    \bottomrule
  \end{tabular}
    \begin{tablenotes}[flushleft]\footnotesize
        \item \textit{Note.} For the normal distribution scenarios, the true median difference $\mu$ was introduced by shifting the mean of Group~1 by $c$, following the location-shift setup of McGrath et al.\supercite{McGrath2020-gt}.
        The shift was calibrated as described in Section~3.1; specifically, $c=0.2213$ and $c=0.1278$ for design-average per-group sample sizes $n=100$ and $n=300$, respectively.
        Group~2 receives no shift, so the true median difference equals $c$.
        Skew-normal parameters follow the direct parameterization (location, scale, shape).
        Log-normal parameters are on the natural-log scale.
    \end{tablenotes}
  \end{threeparttable}
\end{table}

\begin{figure}[t]
  \centering
  \includegraphics[width=\linewidth]{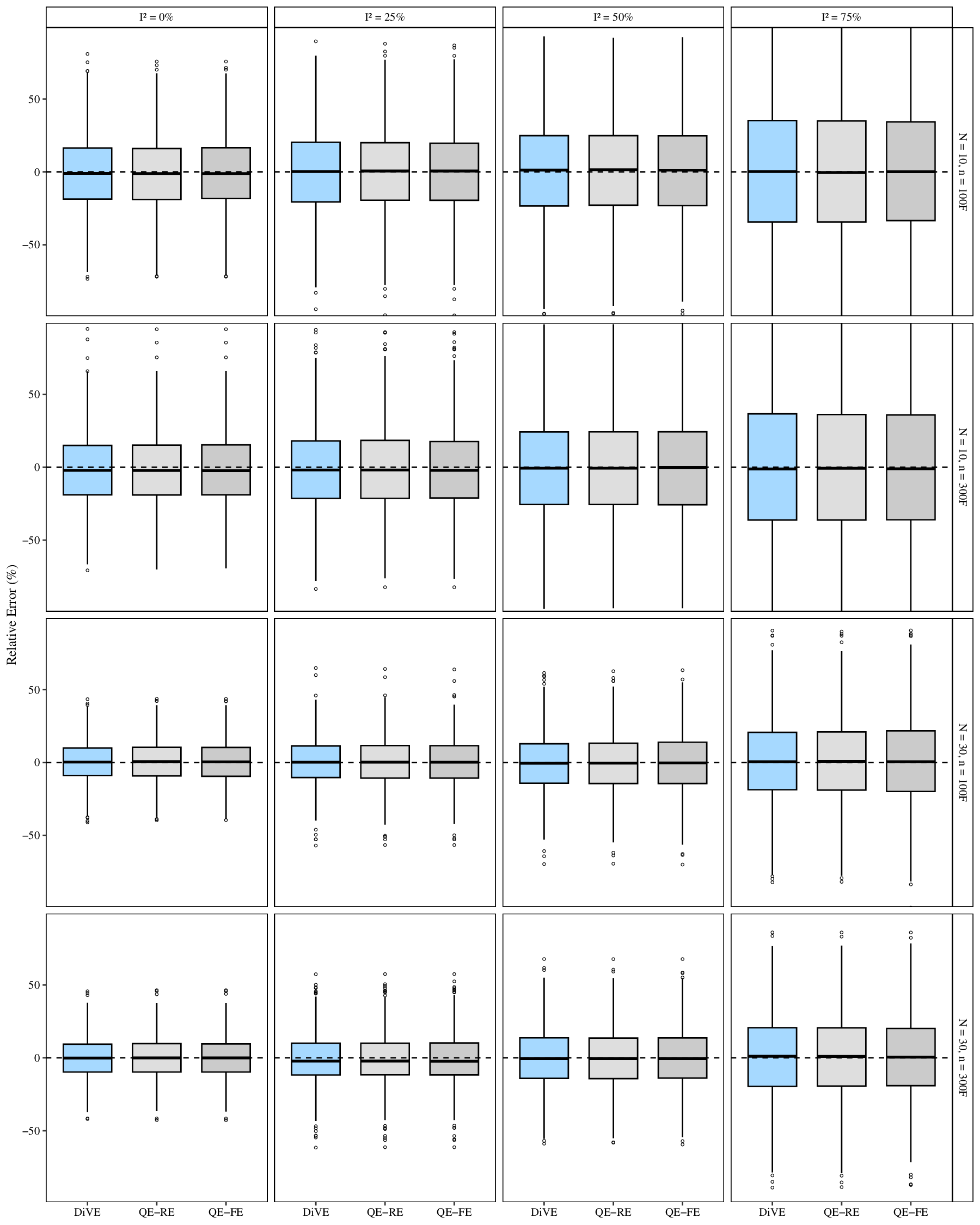}
  \caption{Distribution of relative errors for point estimate (normal scenario with fixed sample sizes).\\[0.5em]
    {\footnotesize
      \textit{Note.} Panels show relative errors for the point estimator; methods compared are DiVE, the QE method under a RE model (QE--RE), and the QE method under an FE model (QE--FE).
      The four columns of panels correspond to $I^2\in\{0,25,50,75\}\%$.
      The central line in each box denotes the median error; boxes span the interquartile range; the dashed reference line at zero indicates no error.
    }
  }
  \label{fig:res_box.p_norm_F}
\end{figure}

\begin{figure}[t]
  \centering
  \includegraphics[width=\linewidth]{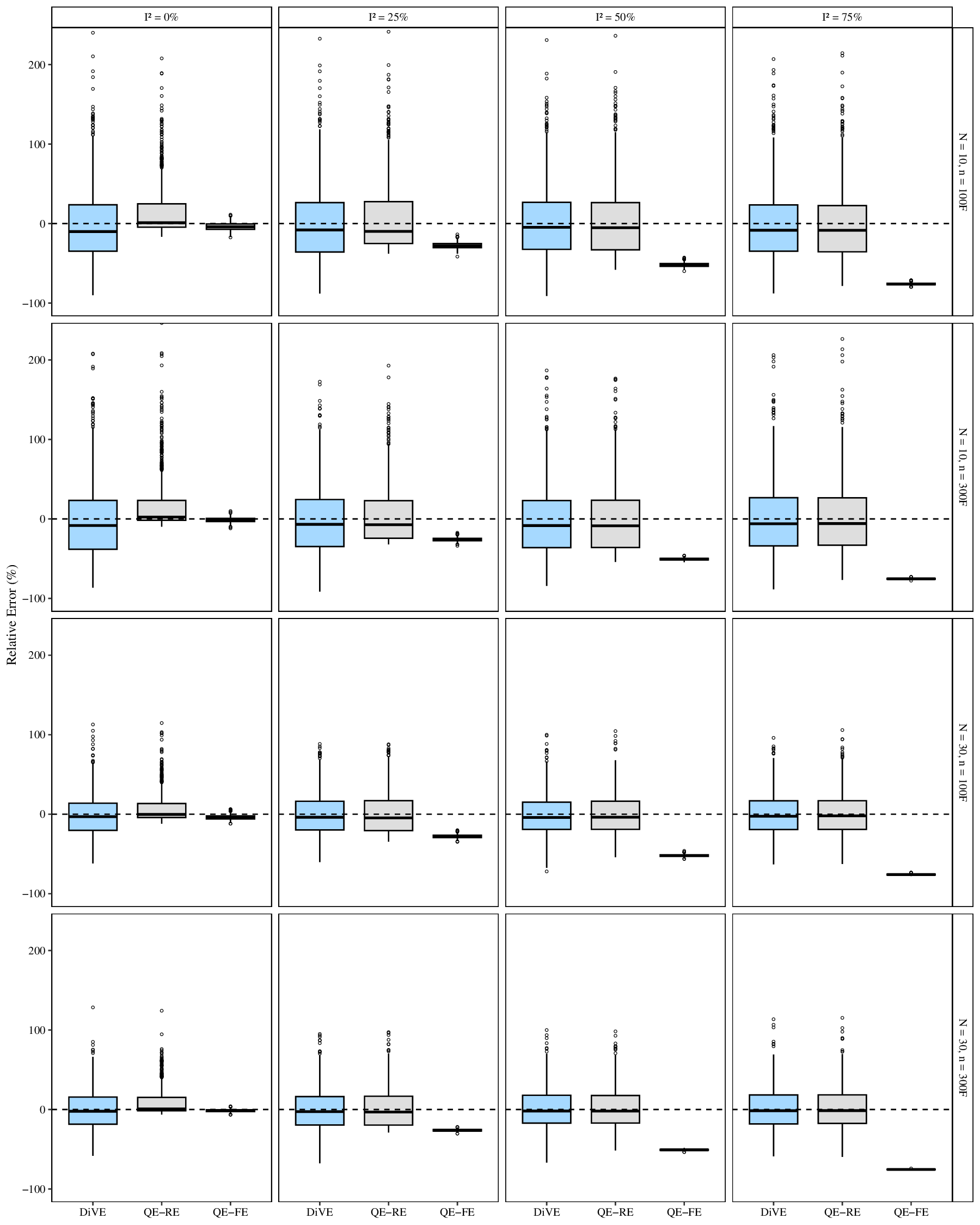}
  \caption{Distribution of relative errors for variance estimate (normal scenario with fixed sample sizes).\\[0.5em]
    {\footnotesize
      \textit{Note.} Panels show relative errors for the variance estimator; methods compared are DiVE, the QE method under a RE model (QE--RE), and the QE method under an FE model (QE--FE).
      The four columns of panels correspond to $I^2\in\{0,25,50,75\}\%$.
      The central line in each box denotes the median error; boxes span the interquartile range; the dashed reference line at zero indicates no error.
    }
  }
  \label{fig:res_box.v_norm_F}
\end{figure}

\FloatBarrier
\begin{table}[t]
  \centering
  \caption{Performance of point and variance estimators under a normal scenario ($N=10$, fixed sample sizes $n=100$).}
  \label{tab:res_F_normN10n100}
  \renewcommand{\arraystretch}{1.15}
  \begin{threeparttable}
    \begin{tabular}{@{}cc*{8}{S}@{}}
      \toprule
      &
      & \multicolumn{2}{c}{\textbf{Point}}
      & \multicolumn{2}{c}{\textbf{Variance}}
      & \multicolumn{2}{c}{\textbf{$z$-based}}
      & \multicolumn{2}{c}{\textbf{$t$-based}} \\
      \cmidrule(lr){3-4}\cmidrule(lr){5-6}\cmidrule(lr){7-8}\cmidrule(lr){9-10}
      \textbf{$I^2$}
      & \textbf{Method}
      & \textbf{\%Bias} & \textbf{\%MSE}
      & \textbf{\%Bias} & \textbf{\%MSE}
      & \textbf{CP}    & \textbf{AW}
      & \textbf{CP}    & \textbf{AW} \\
      \midrule
      $0\%$   & DiVE    &  -0.978 &  6.501 &  -1.218 & 23.547 & 0.904 & 0.212 & 0.937 & 0.245 \\
              & QE--RE  &  -1.074 &  6.571 &  15.906 & 14.183 & 0.951 & 0.234 & 0.980 & 0.271 \\
              & QE--FE  &  -1.075 &  6.616 &  -3.837 &  0.375 & 0.937 & 0.215 & 0.974 & 0.249 \\
      \midrule
      $25\%$  & DiVE    & -0.167 &  8.989 &   0.080 & 23.919 & 0.903 & 0.246 & 0.937 & 0.284 \\
              & QE--RE  & -0.098 &  8.982 &   7.020 & 19.166 & 0.931 & 0.258 & 0.965 & 0.298 \\
              & QE--FE  & 0.005 &  9.062 & -27.837 &  7.891 & 0.892 & 0.215 & 0.941 & 0.249 \\
      \midrule
      $50\%$  & DiVE    & 0.408 & 12.765 &   1.067 & 21.350 & 0.918 & 0.304 & 0.951 & 0.351 \\
              & QE--RE  & 0.287 & 12.742 &   2.474 & 20.801 & 0.925 & 0.307 & 0.959 & 0.355 \\
              & QE--FE  &  0.205 & 12.866 & -51.957 & 27.060 & 0.829 & 0.215 & 0.876 & 0.248 \\
      \midrule
      $75\%$  & DiVE    &  1.680 & 25.468 &  -1.572 & 22.432 & 0.923 & 0.424 & 0.956 & 0.489 \\
              & QE--RE  &  1.612 & 25.522 &   -1.304 & 22.920 & 0.923 & 0.424 & 0.956 & 0.490 \\
              & QE--FE  &  1.669 & 26.251 & -75.982 & 57.750 & 0.664 & 0.215 & 0.721 & 0.248 \\
      \bottomrule
    \end{tabular}
    \begin{tablenotes}[flushleft]\footnotesize
        \item \textit{Note.} Results are based on 1,000 simulation replicates per heterogeneity level.
        \%Bias, relative bias; \%MSE, relative mean-squared error; CP, empirical coverage probability of 95\% confidence intervals; AW, average width.
        $z$-based intervals use the standard normal critical value, while $t$-based intervals use quantiles from a $t$-distribution with $N\!-\!1$ degrees of freedom.
        The true pooled difference is the benchmark for point-estimate metrics; variance metrics are evaluated against method-specific analytic targets defined in Appendix~B.
    \end{tablenotes}
  \end{threeparttable}
\end{table}

\begin{table}[t]
  \centering
  \caption{Performance of point and variance estimators under a normal scenario ($N=10$, fixed sample sizes $n=300$).}
  \label{tab:res_F_normN10n300}
  \renewcommand{\arraystretch}{1.15}
  \begin{threeparttable}
    \begin{tabular}{@{}cc*{8}{S}@{}}
      \toprule
      &
      & \multicolumn{2}{c}{\textbf{Point}}
      & \multicolumn{2}{c}{\textbf{Variance}}
      & \multicolumn{2}{c}{\textbf{$z$-based}}
      & \multicolumn{2}{c}{\textbf{$t$-based}} \\
      \cmidrule(lr){3-4}\cmidrule(lr){5-6}\cmidrule(lr){7-8}\cmidrule(lr){9-10}
      \textbf{$I^2$}
      & \textbf{Method}
      & \textbf{\%Bias} & \textbf{\%MSE}
      & \textbf{\%Bias} & \textbf{\%MSE}
      & \textbf{CP}    & \textbf{AW}
      & \textbf{CP}    & \textbf{AW} \\
      \midrule
      $0\%$   & DiVE    &  -1.171 &  6.046 &  -1.934 & 23.865 & 0.927 & 0.122 & 0.957 & 0.141 \\
              & QE--RE  &  -1.302 &  6.078 &  17.458 & 14.323 & 0.964 & 0.136 & 0.986 & 0.157 \\
              & QE--FE  &  -1.306 &  6.083 &  -1.303 &  0.103 & 0.954 & 0.126 & 0.979 & 0.145 \\
      \midrule
      $25\%$  & DiVE    &  -1.447 &  8.495 &  -1.710 & 19.631 & 0.922 & 0.141 & 0.949 & 0.163 \\
              & QE--RE  &  -1.522 &  8.525 &   5.386 & 14.026 & 0.945 & 0.148 & 0.972 & 0.171 \\
              & QE--FE  &  -1.576 &  8.540 & -25.926 &  6.772 & 0.921 & 0.126 & 0.951 & 0.145 \\
      \midrule
      $50\%$  & DiVE    &  -0.346 & 13.145 &  -2.178 & 19.813 & 0.922 & 0.173 & 0.951 & 0.200 \\
              & QE--RE  &  -0.267 & 13.215 &   -0.738 & 18.188 & 0.920 & 0.175 & 0.957 & 0.202 \\
              & QE--FE  & -0.209 & 13.354 & -50.660 & 25.685 & 0.823 & 0.126 & 0.888 & 0.145 \\
      \midrule
      $75\%$  & DiVE    & -0.314 & 26.936 &   0.341 & 22.949 & 0.912 & 0.247 & 0.950 & 0.285 \\
              & QE--RE  & -0.289 & 26.960 &   0.433 & 22.934 & 0.914 & 0.247 & 0.950 & 0.285 \\
              & QE--FE  & -0.336 & 27.168 & -75.365 & 56.804 & 0.650 & 0.126 & 0.723 & 0.145 \\
      \bottomrule
    \end{tabular}
    \begin{tablenotes}[flushleft]\footnotesize
        \item \textit{Note.} Results are based on 1,000 simulation replicates per heterogeneity level.
        \%Bias, relative bias; \%MSE, relative mean-squared error; CP, empirical coverage probability of 95\% confidence intervals; AW, average width.
        $z$-based intervals use the standard normal critical value, while $t$-based intervals use quantiles from a $t$-distribution with $N\!-\!1$ degrees of freedom.
        The true pooled difference is the benchmark for point-estimate metrics; variance metrics are evaluated against method-specific analytic targets defined in Appendix~B.
    \end{tablenotes}
  \end{threeparttable}
\end{table}

\begin{table}[t]
  \centering
  \caption{Performance of point and variance estimators under a normal scenario ($N=30$, fixed sample sizes $n=100$).}
  \label{tab:res_F_normN30n100}
  \renewcommand{\arraystretch}{1.15}
  \begin{threeparttable}
    \begin{tabular}{@{}cc*{8}{S}@{}}
      \toprule
      &
      & \multicolumn{2}{c}{\textbf{Point}}
      & \multicolumn{2}{c}{\textbf{Variance}}
      & \multicolumn{2}{c}{\textbf{$z$-based}}
      & \multicolumn{2}{c}{\textbf{$t$-based}} \\
      \cmidrule(lr){3-4}\cmidrule(lr){5-6}\cmidrule(lr){7-8}\cmidrule(lr){9-10}
      \textbf{$I^2$}
      & \textbf{Method}
      & \textbf{\%Bias} & \textbf{\%MSE}
      & \textbf{\%Bias} & \textbf{\%MSE}
      & \textbf{CP}    & \textbf{AW}
      & \textbf{CP}    & \textbf{AW} \\
      \midrule
      $0\%$   & DiVE    & 0.626 &  2.148 &  -1.670 &  6.613 & 0.934 & 0.125 & 0.943 & 0.130 \\
              & QE--RE  & 0.512 &  2.166 &   7.193 &  3.687 & 0.952 & 0.131 & 0.959 & 0.137 \\
              & QE--FE  & 0.480 &  2.175 &  -3.963 &  0.245 & 0.939 & 0.124 & 0.950 & 0.130 \\
      \midrule
      $25\%$  & DiVE    &  0.592 &  2.723 &  -0.870 &  6.965 & 0.931 & 0.145 & 0.946 & 0.151 \\
              & QE--RE  &  0.486 &  2.775 &   0.549 &  6.155 & 0.936 & 0.146 & 0.952 & 0.152 \\
              & QE--FE  &  0.425 &  2.812 & -27.932 &  7.852 & 0.899 & 0.124 & 0.904 & 0.130 \\
      \midrule
      $50\%$  & DiVE    & -0.410 &  4.366 &   -1.865 &  6.681 & 0.929 & 0.176 & 0.941 & 0.184 \\
              & QE--RE  & -0.359 &  4.350 &   -1.293 &  6.788 & 0.926 & 0.177 & 0.944 & 0.184 \\
              & QE--FE  & -0.316 &  4.396 & -52.021 & 27.083 & 0.810 & 0.124 & 0.828 & 0.130 \\
      \midrule
      $75\%$  & DiVE    &  0.653 &  8.551 &  -0.415 &  7.022 & 0.940 & 0.251 & 0.949 & 0.262 \\
              & QE--RE  & 0.695 &  8.599 &   0.022 &  7.193 & 0.941 & 0.251 & 0.951 & 0.262 \\
              & QE--FE  & 0.782 &  8.917 & -75.973 & 57.725 & 0.654 & 0.124 & 0.671 & 0.130 \\
      \bottomrule
    \end{tabular}
    \begin{tablenotes}[flushleft]\footnotesize
        \item \textit{Note.} Results are based on 1,000 simulation replicates per heterogeneity level.
        \%Bias, relative bias; \%MSE, relative mean-squared error; CP, empirical coverage probability of 95\% confidence intervals; AW, average width.
        $z$-based intervals use the standard normal critical value, while $t$-based intervals use quantiles from a $t$-distribution with $N\!-\!1$ degrees of freedom.
        The true pooled difference is the benchmark for point-estimate metrics; variance metrics are evaluated against method-specific analytic targets defined in Appendix~B.
    \end{tablenotes}
  \end{threeparttable}
\end{table}

\begin{table}[t]
  \centering
  \caption{Performance of point and variance estimators under a normal scenario ($N=30$, fixed sample sizes $n=300$).}
  \label{tab:res_F_normN30n300}
  \renewcommand{\arraystretch}{1.15}
  \begin{threeparttable}
    \begin{tabular}{@{}cc*{8}{S}@{}}
      \toprule
      &
      & \multicolumn{2}{c}{\textbf{Point}}
      & \multicolumn{2}{c}{\textbf{Variance}}
      & \multicolumn{2}{c}{\textbf{$z$-based}}
      & \multicolumn{2}{c}{\textbf{$t$-based}} \\
      \cmidrule(lr){3-4}\cmidrule(lr){5-6}\cmidrule(lr){7-8}\cmidrule(lr){9-10}
      \textbf{$I^2$}
      & \textbf{Method}
      & \textbf{\%Bias} & \textbf{\%MSE}
      & \textbf{\%Bias} & \textbf{\%MSE}
      & \textbf{CP}    & \textbf{AW}
      & \textbf{CP}    & \textbf{AW} \\
      \midrule
      $0\%$   & DiVE    &  0.080 &  2.121 &   -0.251 &  6.718 & 0.933 & 0.073 & 0.943 & 0.076 \\
              & QE--RE  &  0.105 &  2.126 &  9.441 &  3.833 & 0.955 & 0.076 & 0.961 & 0.080 \\
              & QE--FE  &  0.104 &  2.130 &  -1.422 &  0.049 & 0.945 & 0.073 & 0.954 & 0.076 \\
      \midrule
      $25\%$  & DiVE    & -1.256 &  3.023 &  -0.733 &  6.932 & 0.931 & 0.084 & 0.940 & 0.087 \\
              & QE--RE  &  -1.294 &  3.046 &   1.082 &  5.646 & 0.934 & 0.084 & 0.947 & 0.088 \\
              & QE--FE  &  -1.327 &  3.055 & -26.091 &  6.824 & 0.900 & 0.073 & 0.913 & 0.076 \\
      \midrule
      $50\%$  & DiVE    &  -0.295 &  4.190 &  0.757 &  7.027 & 0.943 & 0.103 & 0.951 & 0.108 \\
              & QE--RE  &  -0.289 &  4.196 &  0.840 &  6.994 & 0.945 & 0.103 & 0.954 & 0.108 \\
              & QE--FE  &  -0.261 &  4.216 & -50.664 & 25.676 & 0.837 & 0.073 & 0.851 & 0.076 \\
      \midrule
      $75\%$  & DiVE    &  0.530 &  8.484 &   1.448 &  6.573 & 0.956 & 0.146 & 0.961 & 0.153 \\
              & QE--RE  &  0.537 &  8.474 &   1.620 &  6.705 & 0.954 & 0.146 & 0.960 & 0.153 \\
              & QE--FE  &  0.558 &  8.491 & -75.370 & 56.809 & 0.673 & 0.073 & 0.692 & 0.076 \\
      \bottomrule
    \end{tabular}
    \begin{tablenotes}[flushleft]\footnotesize
        \item \textit{Note.} Results are based on 1,000 simulation replicates per heterogeneity level.
        \%Bias, relative bias; \%MSE, relative mean-squared error; CP, empirical coverage probability of 95\% confidence intervals; AW, average width.
        $z$-based intervals use the standard normal critical value, while $t$-based intervals use quantiles from a $t$-distribution with $N\!-\!1$ degrees of freedom.
        The true pooled difference is the benchmark for point-estimate metrics; variance metrics are evaluated against method-specific analytic targets defined in Appendix~B.
    \end{tablenotes}
  \end{threeparttable}
\end{table}

\begin{figure}[t]
  \centering
  \includegraphics[width=\linewidth]{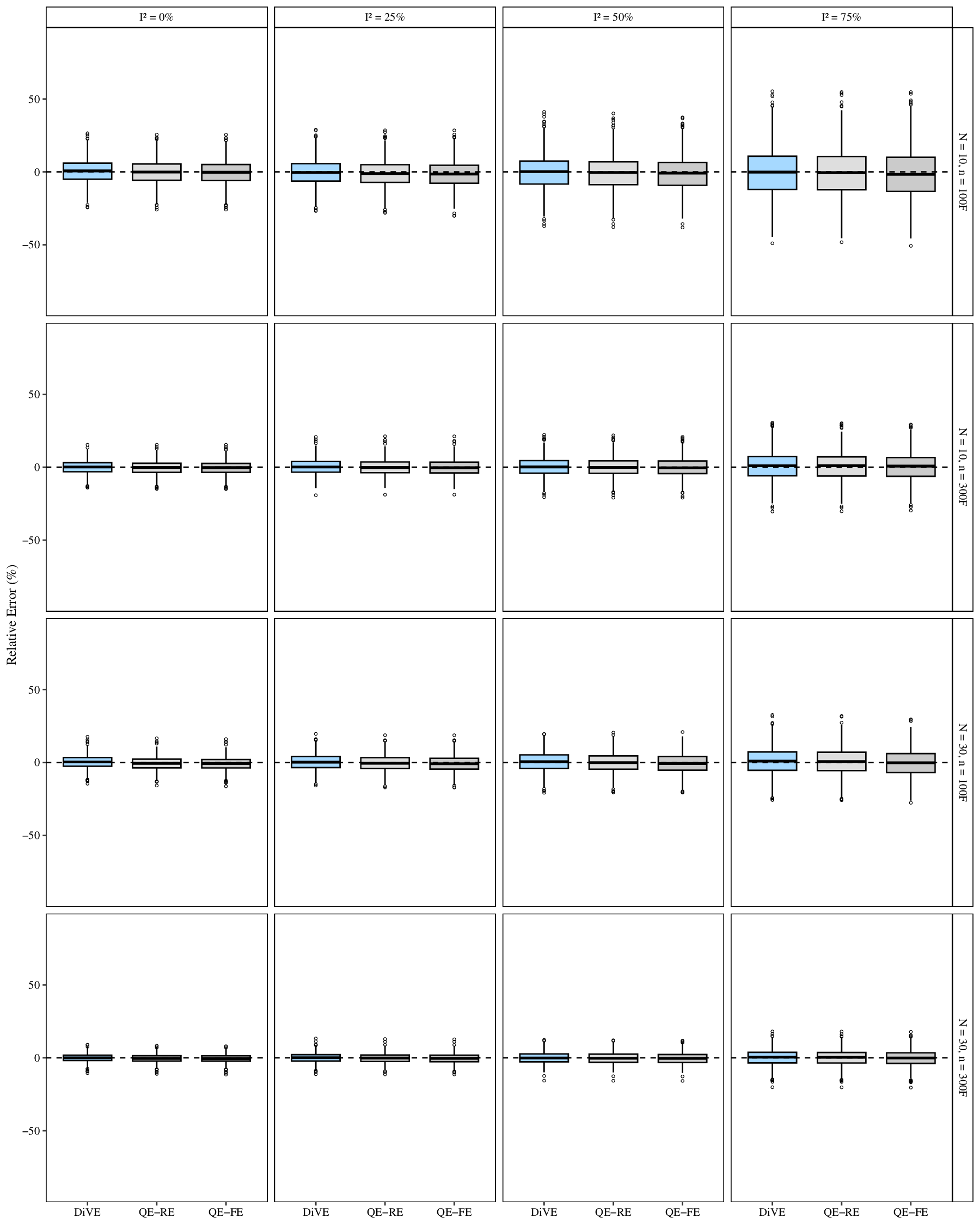}
  \caption{Distribution of relative errors for point estimate (skew-normal scenario with fixed sample sizes).\\[0.5em]
    {\footnotesize
      \textit{Note.} Panels show relative errors for the point estimator; methods compared are DiVE, the QE method under a RE model (QE--RE), and the QE method under an FE model (QE--FE).
      The four columns of panels correspond to $I^2\in\{0,25,50,75\}\%$.
      The central line in each box denotes the median error; boxes span the interquartile range; the dashed reference line at zero indicates no error.
    }
  }
  \label{fig:res_box.p_skew_F}
\end{figure}

\begin{figure}[t]
  \centering
  \includegraphics[width=\linewidth]{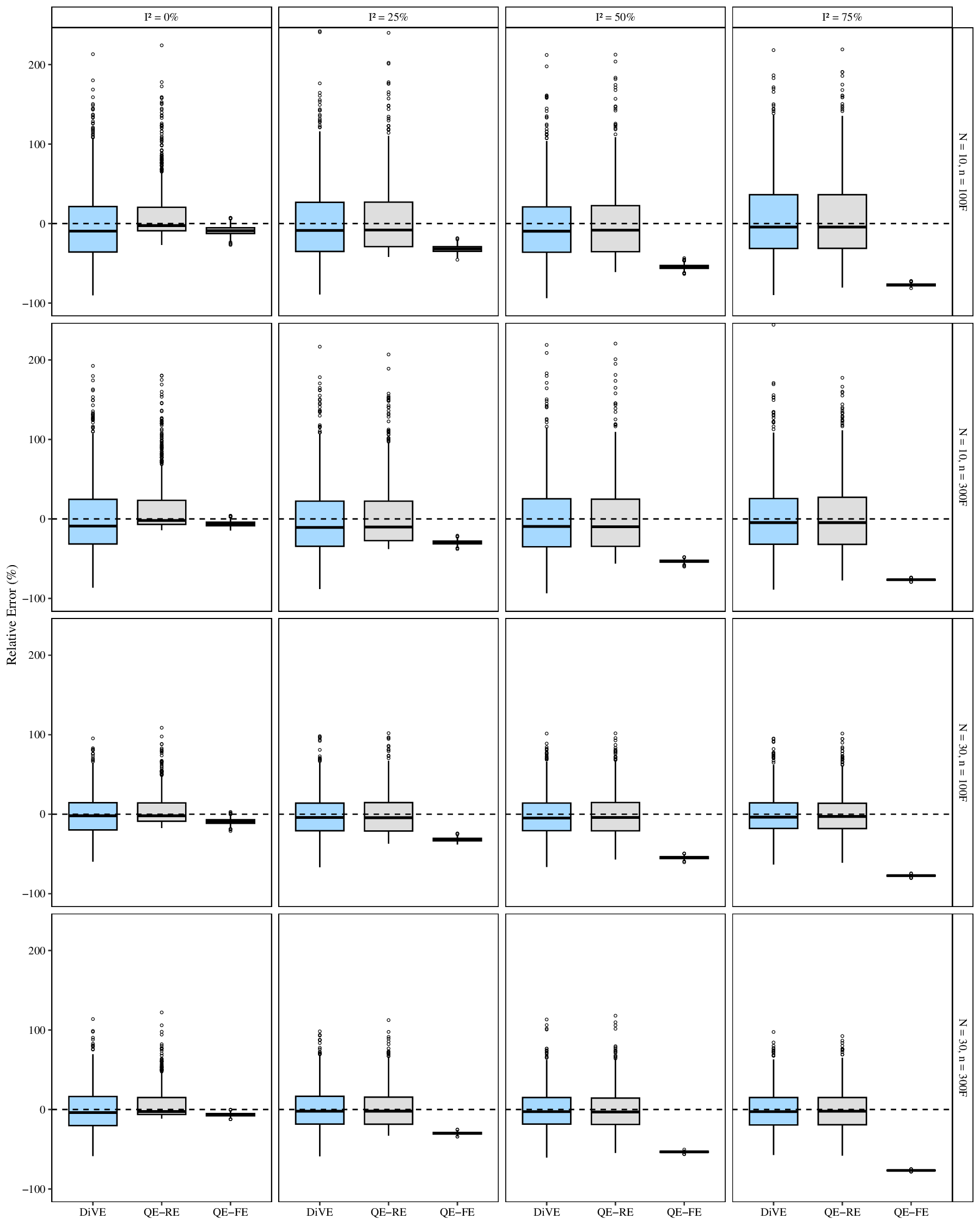}
  \caption{Distribution of relative errors for variance estimate (skew-normal scenario with fixed sample sizes).\\[0.5em]
    {\footnotesize
      \textit{Note.} Panels show relative errors for the variance estimator; methods compared are DiVE, the QE method under a RE model (QE--RE), and the QE method under an FE model (QE--FE).
      The four columns of panels correspond to $I^2\in\{0,25,50,75\}\%$.
      The central line in each box denotes the median error; boxes span the interquartile range; the dashed reference line at zero indicates no error.
    }
  }
  \label{fig:res_box.v_skew_F}
\end{figure}

\FloatBarrier
\begin{table}[t]
  \centering
  \caption{Performance of point and variance estimators under a skew-normal scenario ($N=10$, fixed sample sizes $n=100$).}
  \label{tab:res_F_skewN10n100}
  \renewcommand{\arraystretch}{1.15}
  \begin{threeparttable}
    \begin{tabular}{@{}cc*{8}{S}@{}}
      \toprule
      &
      & \multicolumn{2}{c}{\textbf{Point}}
      & \multicolumn{2}{c}{\textbf{Variance}}
      & \multicolumn{2}{c}{\textbf{$z$-based}}
      & \multicolumn{2}{c}{\textbf{$t$-based}} \\
      \cmidrule(lr){3-4}\cmidrule(lr){5-6}\cmidrule(lr){7-8}\cmidrule(lr){9-10}
      \textbf{$I^2$}
      & \textbf{Method}
      & \textbf{\%Bias} & \textbf{\%MSE}
      & \textbf{\%Bias} & \textbf{\%MSE}
      & \textbf{CP}    & \textbf{AW}
      & \textbf{CP}    & \textbf{AW} \\
      \midrule
      $0\%$   & DiVE    &  0.625 & 0.676 &  -2.265 & 20.552 & 0.913 & 1.050 & 0.943 & 1.211 \\
              & QE--RE  & -0.148 & 0.683 &  11.517 & 12.697 & 0.955 & 1.141 & 0.975 & 1.317 \\
              & QE--FE  & -0.356 & 0.689 & -8.830 &  1.095 & 0.939 & 1.041 & 0.965 & 1.201 \\
      \midrule
      $25\%$  & DiVE    &  -0.204 & 0.862 &  -0.721 & 22.763 & 0.929 & 1.219 & 0.960 & 1.407 \\
              & QE--RE  & -0.909 & 0.873 &   4.550 & 18.309 & 0.951 & 1.265 & 0.967 & 1.460 \\
              & QE--FE  & -1.278 & 0.899 & -31.869 & 10.338 & 0.895 & 1.039 & 0.931 & 1.199 \\
      \midrule
      $50\%$  & DiVE    &  0.087 & 1.378 &   -2.753 & 20.344 & 0.919 & 1.480 & 0.943 & 1.709 \\
              & QE--RE  & -0.435 & 1.377 &   -1.391 & 20.028 & 0.919 & 1.495 & 0.950 & 1.726 \\
              & QE--FE  & -1.048 & 1.417 & -54.475 & 29.760 & 0.807 & 1.040 & 0.860 & 1.200 \\
      \midrule
      $75\%$  & DiVE    & -0.303 & 2.742 &  4.346 & 23.468 & 0.921 & 2.168 & 0.953 & 2.502 \\
              & QE--RE  & -0.539 & 2.741 &  4.746 & 23.954 & 0.923 & 2.172 & 0.954 & 2.507 \\
              & QE--FE  & -1.369 & 2.835 & -77.203 & 59.624 & 0.637 & 1.041 & 0.720 & 1.201 \\
      \bottomrule
    \end{tabular}
    \begin{tablenotes}[flushleft]\footnotesize
        \item \textit{Note.} Results are based on 1,000 simulation replicates per heterogeneity level.
        \%Bias, relative bias; \%MSE, relative mean-squared error; CP, empirical coverage probability of 95\% confidence intervals; AW, average width.
        $z$-based intervals use the standard normal critical value, while $t$-based intervals use quantiles from a $t$-distribution with $N\!-\!1$ degrees of freedom.
        The true pooled difference is the benchmark for point-estimate metrics; variance metrics are evaluated against method-specific analytic targets defined in Appendix~B.
    \end{tablenotes}
  \end{threeparttable}
\end{table}

\begin{table}[t]
  \centering
  \caption{Performance of point and variance estimators under a skew-normal scenario ($N=10$, fixed sample sizes $n=300$).}
  \label{tab:res_F_skewN10n300}
  \renewcommand{\arraystretch}{1.15}
  \begin{threeparttable}
    \begin{tabular}{@{}cc*{8}{S}@{}}
      \toprule
      &
      & \multicolumn{2}{c}{\textbf{Point}}
      & \multicolumn{2}{c}{\textbf{Variance}}
      & \multicolumn{2}{c}{\textbf{$z$-based}}
      & \multicolumn{2}{c}{\textbf{$t$-based}} \\
      \cmidrule(lr){3-4}\cmidrule(lr){5-6}\cmidrule(lr){7-8}\cmidrule(lr){9-10}
      \textbf{$I^2$}
      & \textbf{Method}
      & \textbf{\%Bias} & \textbf{\%MSE}
      & \textbf{\%Bias} & \textbf{\%MSE}
      & \textbf{CP}    & \textbf{AW}
      & \textbf{CP}    & \textbf{AW} \\
      \midrule
      $0\%$   & DiVE    &  0.031 & 0.229 &   -0.397 & 21.583 & 0.910 & 0.612 & 0.949 & 0.706 \\
              & QE--RE  & -0.303 & 0.229 &  13.852 & 13.121 & 0.947 & 0.666 & 0.976 & 0.768 \\
              & QE--FE  & -0.389 & 0.230 &  -6.385 &  0.516 & 0.934 & 0.609 & 0.967 & 0.703 \\
      \midrule
      $25\%$  & DiVE    &  0.217 & 0.303 &   -1.890 & 21.040 & 0.922 & 0.701 & 0.944 & 0.809 \\
              & QE--RE  &  -0.036 & 0.301 &   3.574 & 15.949 & 0.935 & 0.728 & 0.960 & 0.841 \\
              & QE--FE  & -0.171 & 0.303 & -29.655 &  8.859 & 0.895 & 0.610 & 0.936 & 0.704 \\
      \midrule
      $50\%$  & DiVE    &  0.233 & 0.432 &   -1.812 & 21.338 & 0.910 & 0.858 & 0.943 & 0.990 \\
              & QE--RE  & 0.043 & 0.436 &   -0.572 & 19.654 & 0.917 & 0.867 & 0.953 & 1.001 \\
              & QE--FE  & -0.155 & 0.443 & -53.199 & 28.330 & 0.821 & 0.609 & 0.878 & 0.703 \\
      \midrule
      $75\%$  & DiVE    &  0.730 & 0.969 &  0.440 & 21.411 & 0.912 & 1.229 & 0.958 & 1.418 \\
              & QE--RE  &  0.641 & 0.967 &  0.818 & 21.761 & 0.917 & 1.231 & 0.956 & 1.421 \\
              & QE--FE  & 0.361 & 0.972 & -76.561 & 58.623 & 0.643 & 0.609 & 0.700 & 0.703 \\
      \bottomrule
    \end{tabular}
    \begin{tablenotes}[flushleft]\footnotesize
        \item \textit{Note.} Results are based on 1,000 simulation replicates per heterogeneity level.
        \%Bias, relative bias; \%MSE, relative mean-squared error; CP, empirical coverage probability of 95\% confidence intervals; AW, average width.
        $z$-based intervals use the standard normal critical value, while $t$-based intervals use quantiles from a $t$-distribution with $N\!-\!1$ degrees of freedom.
        The true pooled difference is the benchmark for point-estimate metrics; variance metrics are evaluated against method-specific analytic targets defined in Appendix~B.
    \end{tablenotes}
  \end{threeparttable}
\end{table}

\begin{table}[t]
  \centering
  \caption{Performance of point and variance estimators under a skew-normal scenario ($N=30$, fixed sample sizes $n=100$).}
  \label{tab:res_F_skewN30n100}
  \renewcommand{\arraystretch}{1.15}
  \begin{threeparttable}
    \begin{tabular}{@{}cc*{8}{S}@{}}
      \toprule
      &
      & \multicolumn{2}{c}{\textbf{Point}}
      & \multicolumn{2}{c}{\textbf{Variance}}
      & \multicolumn{2}{c}{\textbf{$z$-based}}
      & \multicolumn{2}{c}{\textbf{$t$-based}} \\
      \cmidrule(lr){3-4}\cmidrule(lr){5-6}\cmidrule(lr){7-8}\cmidrule(lr){9-10}
      \textbf{$I^2$}
      & \textbf{Method}
      & \textbf{\%Bias} & \textbf{\%MSE}
      & \textbf{\%Bias} & \textbf{\%MSE}
      & \textbf{CP}    & \textbf{AW}
      & \textbf{CP}    & \textbf{AW} \\
      \midrule
      $0\%$   & DiVE    &  0.310 & 0.206 &  -1.203 &  6.501 & 0.954 & 0.621 & 0.962 & 0.648 \\
              & QE--RE  & -0.668 & 0.209 &   5.008 &  3.936 & 0.959 & 0.643 & 0.972 & 0.671 \\
              & QE--FE  & -0.830 & 0.211 & -9.203 &  0.963 & 0.942 & 0.600 & 0.953 & 0.626 \\
      \midrule
      $25\%$  & DiVE    &  0.274 & 0.306 &  -1.712 &  6.516 & 0.942 & 0.715 & 0.953 & 0.746 \\
              & QE--RE  & -0.474 & 0.311 &  -1.062 &  6.088 & 0.938 & 0.718 & 0.942 & 0.749 \\
              & QE--FE  & -0.825 & 0.321 & -31.949 & 10.267 & 0.880 & 0.600 & 0.893 & 0.626 \\
      \midrule
      $50\%$  & DiVE    &  0.487 & 0.455 &  -1.656 &  7.167 & 0.947 & 0.875 & 0.953 & 0.913 \\
              & QE--RE  & -0.031 & 0.451 &   -1.360 &  7.375 & 0.947 & 0.876 & 0.953 & 0.914 \\
              & QE--FE  & -0.627 & 0.465 & -54.588 & 29.825 & 0.807 & 0.600 & 0.827 & 0.626 \\
      \midrule
      $75\%$  & DiVE    &  0.951 & 0.922 &  -0.497 &  6.351 & 0.931 & 1.246 & 0.942 & 1.300 \\
              & QE--RE  & 0.701 & 0.922 &  -0.123 &  6.455 & 0.933 & 1.248 & 0.940 & 1.303 \\
              & QE--FE  & -0.080 & 0.951 & -77.295 & 59.752 & 0.631 & 0.600 & 0.648 & 0.626 \\
      \bottomrule
    \end{tabular}
    \begin{tablenotes}[flushleft]\footnotesize
        \item \textit{Note.} Results are based on 1,000 simulation replicates per heterogeneity level.
        \%Bias, relative bias; \%MSE, relative mean-squared error; CP, empirical coverage probability of 95\% confidence intervals; AW, average width.
        $z$-based intervals use the standard normal critical value, while $t$-based intervals use quantiles from a $t$-distribution with $N\!-\!1$ degrees of freedom.
        The true pooled difference is the benchmark for point-estimate metrics; variance metrics are evaluated against method-specific analytic targets defined in Appendix~B.
    \end{tablenotes}
  \end{threeparttable}
\end{table}

\begin{table}[t]
  \centering
  \caption{Performance of point and variance estimators under a skew-normal scenario ($N=30$, fixed sample sizes $n=300$).}
  \label{tab:res_F_skewN30n300}
  \renewcommand{\arraystretch}{1.15}
  \begin{threeparttable}
    \begin{tabular}{@{}cc*{8}{S}@{}}
      \toprule
      &
      & \multicolumn{2}{c}{\textbf{Point}}
      & \multicolumn{2}{c}{\textbf{Variance}}
      & \multicolumn{2}{c}{\textbf{$z$-based}}
      & \multicolumn{2}{c}{\textbf{$t$-based}} \\
      \cmidrule(lr){3-4}\cmidrule(lr){5-6}\cmidrule(lr){7-8}\cmidrule(lr){9-10}
      \textbf{$I^2$}
      & \textbf{Method}
      & \textbf{\%Bias} & \textbf{\%MSE}
      & \textbf{\%Bias} & \textbf{\%MSE}
      & \textbf{CP}    & \textbf{AW}
      & \textbf{CP}    & \textbf{AW} \\
      \midrule
      $0\%$   & DiVE    &  -0.014 & 0.079 &   -0.519 &  7.275 & 0.931 & 0.359 & 0.942 & 0.375 \\
              & QE--RE  & -0.391 & 0.081 &   6.975 &  4.214 & 0.945 & 0.375 & 0.956 & 0.391 \\
              & QE--FE  & -0.451 & 0.081 &  -6.387 &  0.447 & 0.932 & 0.352 & 0.941 & 0.367 \\
      \midrule
      $25\%$  & DiVE    &  0.106 & 0.101 &  0.440 &  7.162 & 0.941 & 0.417 & 0.949 & 0.435 \\
              & QE--RE  & -0.208 & 0.101 &   1.203 &  6.289 & 0.945 & 0.419 & 0.955 & 0.437 \\
              & QE--FE  & -0.349 & 0.103 & -29.838 &  8.925 & 0.900 & 0.352 & 0.913 & 0.367 \\
      \midrule
      $50\%$  & DiVE    &  0.014 & 0.159 &   -0.543 &  6.523 & 0.933 & 0.508 & 0.946 & 0.531 \\
              & QE--RE  &  -0.178 & 0.159 &   -0.542 &  6.654 & 0.929 & 0.508 & 0.941 & 0.530 \\
              & QE--FE  & -0.389 & 0.160 & -53.221 & 28.334 & 0.800 & 0.352 & 0.819 & 0.367 \\
      \midrule
      $75\%$  & DiVE    &  0.216 & 0.294 &  -0.946 &  6.426 & 0.948 & 0.718 & 0.957 & 0.749 \\
              & QE--RE  &  0.121 & 0.294 &  -0.810 &  6.495 & 0.950 & 0.718 & 0.956 & 0.749 \\
              & QE--FE  & -0.183 & 0.298 & -76.592 & 58.666 & 0.676 & 0.352 & 0.702 & 0.367 \\
      \bottomrule
    \end{tabular}
    \begin{tablenotes}[flushleft]\footnotesize
        \item \textit{Note.} Results are based on 1,000 simulation replicates per heterogeneity level.
        \%Bias, relative bias; \%MSE, relative mean-squared error; CP, empirical coverage probability of 95\% confidence intervals; AW, average width.
        $z$-based intervals use the standard normal critical value, while $t$-based intervals use quantiles from a $t$-distribution with $N\!-\!1$ degrees of freedom.
        The true pooled difference is the benchmark for point-estimate metrics; variance metrics are evaluated against method-specific analytic targets defined in Appendix~B.
    \end{tablenotes}
  \end{threeparttable}
\end{table}

\FloatBarrier
\begin{figure}[t]
  \centering
  \includegraphics[width=\linewidth]{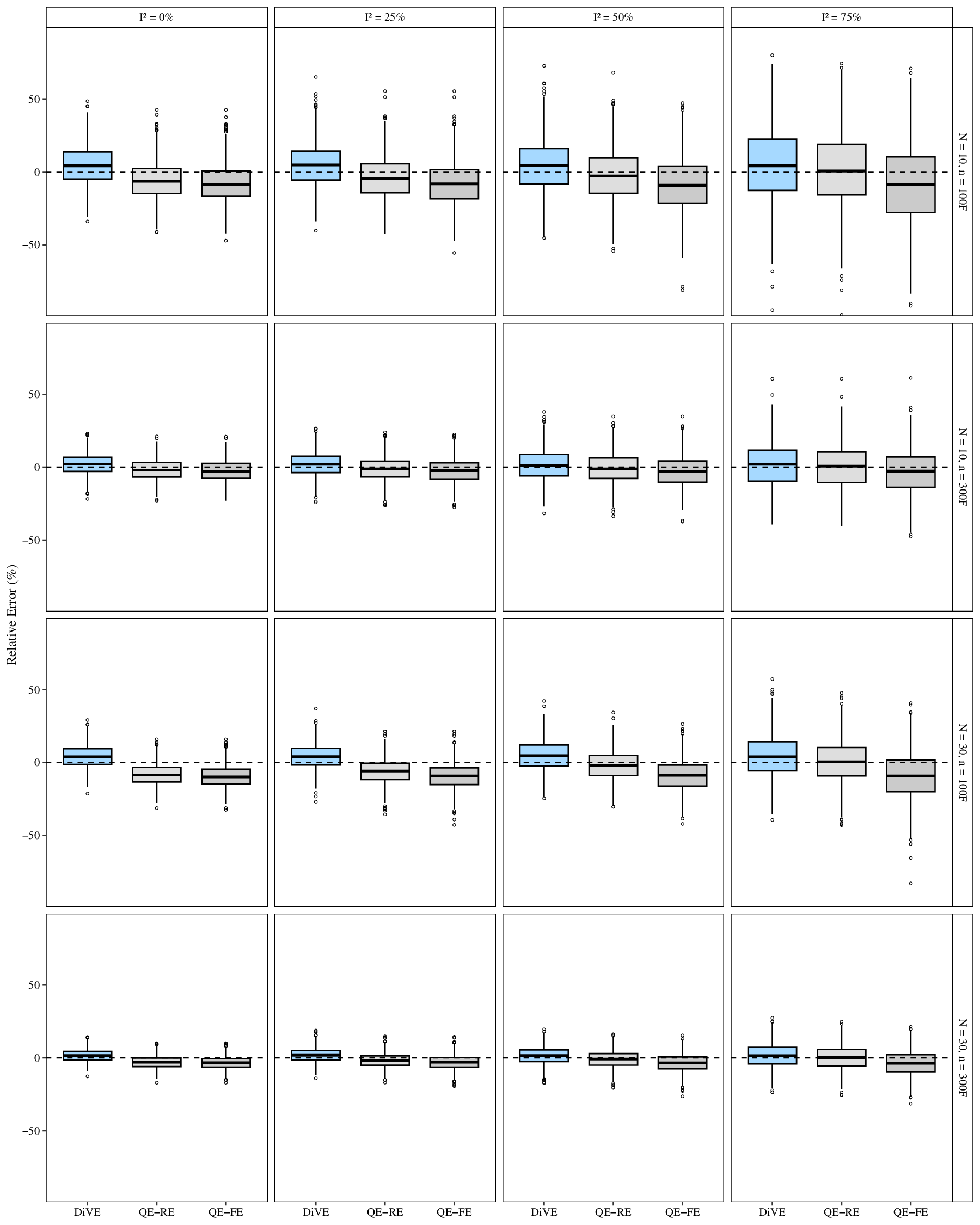}
  \caption{Distribution of relative errors for point estimate (log-normal scenario with fixed sample sizes).\\[0.5em]
    {\footnotesize
      \textit{Note.} Panels show relative errors for the point estimator; methods compared are DiVE, the QE method under a RE model (QE--RE), and the QE method under an FE model (QE--FE).
      The four columns of panels correspond to $I^2\in\{0,25,50,75\}\%$.
      The central line in each box denotes the median error; boxes span the interquartile range; the dashed reference line at zero indicates no error.
    }
  }
  \label{fig:res_box.p_log_F}
\end{figure}

\begin{figure}[t]
  \centering
  \includegraphics[width=\linewidth]{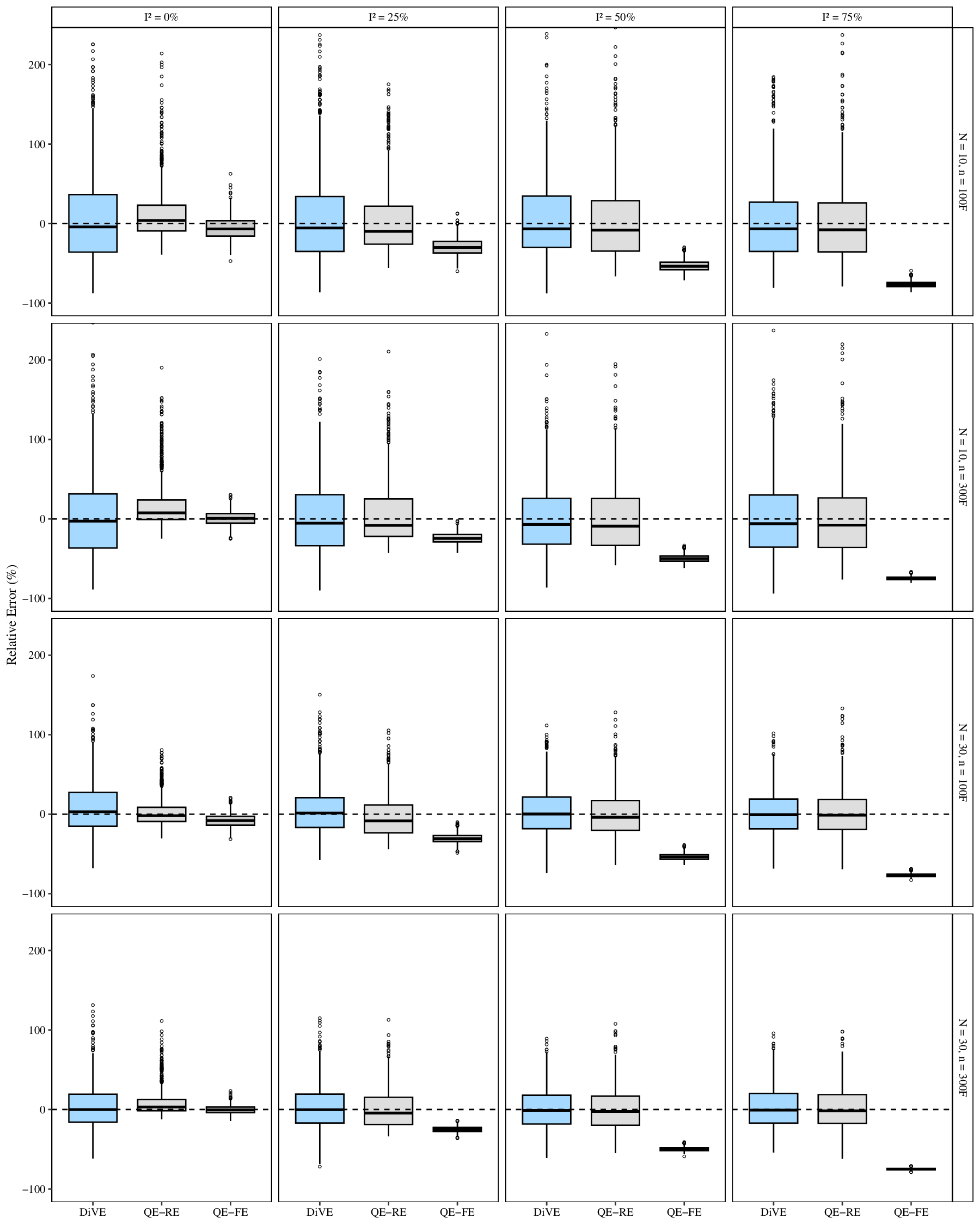}
  \caption{Distribution of relative errors for variance estimate (log-normal scenario with fixed sample sizes).\\[0.5em]
    {\footnotesize
      \textit{Note.} Panels show relative errors for the variance estimator; methods compared are DiVE, the QE method under a RE model (QE--RE), and the QE method under an FE model (QE--FE).
      The four columns of panels correspond to $I^2\in\{0,25,50,75\}\%$.
      The central line in each box denotes the median error; boxes span the interquartile range; the dashed reference line at zero indicates no error.
    }
  }
  \label{fig:res_box.v_log_F}
\end{figure}

\FloatBarrier
\begin{table}[t]
  \centering
  \caption{Performance of point and variance estimators under a log-normal scenario ($N=10$, fixed sample sizes $n=100$).}
  \label{tab:res_F_logN10n100}
  \renewcommand{\arraystretch}{1.15}
  \begin{threeparttable}
    \begin{tabular}{@{}cc*{8}{S}@{}}
      \toprule
      &
      & \multicolumn{2}{c}{\textbf{Point}}
      & \multicolumn{2}{c}{\textbf{Variance}}
      & \multicolumn{2}{c}{\textbf{$z$-based}}
      & \multicolumn{2}{c}{\textbf{$t$-based}} \\
      \cmidrule(lr){3-4}\cmidrule(lr){5-6}\cmidrule(lr){7-8}\cmidrule(lr){9-10}
      \textbf{$I^2$}
      & \textbf{Method}
      & \textbf{\%Bias} & \textbf{\%MSE}
      & \textbf{\%Bias} & \textbf{\%MSE}
      & \textbf{CP}    & \textbf{AW}
      & \textbf{CP}    & \textbf{AW} \\
      \midrule
      $0\%$   & DiVE    &  4.844 & 2.014 &   8.260 & 38.496 & 0.912 &  6.365 & 0.945 &  7.347 \\
              & QE--RE  & -5.750 & 1.996 &   12.828 & 13.700 & 0.916 &  6.674 & 0.957 &  7.703 \\
              & QE--FE  & -7.924 & 2.351 &  -5.808 &  2.505 & 0.862 &  6.140 & 0.924 &  7.087 \\
      \midrule
      $25\%$  & DiVE    &  4.844 & 2.561 &   6.414 & 33.474 & 0.909 &  7.302 & 0.944 &  8.428 \\
              & QE--RE  & -4.174 & 2.386 &   3.573 & 18.090 & 0.914 &  7.326 & 0.944 &  8.456 \\
              & QE--FE  & -7.969 & 3.038 & -29.442 & 9.873 & 0.813 &  6.136 & 0.870 &  7.082 \\
      \midrule
      $50\%$  & DiVE    &  4.209 & 3.582 &   5.675 & 25.434 & 0.914 &  8.973 & 0.942 & 10.357 \\
              & QE--RE  & -2.227 & 3.371 &  2.754 & 25.035 & 0.920 &  8.857 & 0.946 & 10.223 \\
              & QE--FE  & -8.689 & 4.439 & -53.208 & 28.828 & 0.736 &  6.121 & 0.800 &  7.064 \\
      \midrule
      $75\%$  & DiVE    &  4.212 & 6.690 &   1.544 & 25.322 & 0.923 & 12.409 & 0.954 & 14.322 \\
              & QE--RE  &  0.636 & 6.479 &  0.455 & 25.754 & 0.920 & 12.334 & 0.953 & 14.236 \\
              & QE--FE  & -8.660 & 8.328 & -76.454 & 58.599 & 0.600 &  6.138 & 0.663 &  7.085 \\
      \bottomrule
    \end{tabular}
    \begin{tablenotes}[flushleft]\footnotesize
        \item \textit{Note.} Results are based on 1,000 simulation replicates per heterogeneity level.
        \%Bias, relative bias; \%MSE, relative mean-squared error; CP, empirical coverage probability of 95\% confidence intervals; AW, average width.
        $z$-based intervals use the standard normal critical value, while $t$-based intervals use quantiles from a $t$-distribution with $N\!-\!1$ degrees of freedom.
        The true pooled difference is the benchmark for point-estimate metrics; variance metrics are evaluated against method-specific analytic targets defined in Appendix~B.
    \end{tablenotes}
  \end{threeparttable}
\end{table}

\begin{table}[t]
  \centering
  \caption{Performance of point and variance estimators under a log-normal scenario ($N=10$, fixed sample sizes $n=300$).}
  \label{tab:res_F_logN10n300}
  \renewcommand{\arraystretch}{1.15}
  \begin{threeparttable}
    \begin{tabular}{@{}cc*{8}{S}@{}}
      \toprule
      &
      & \multicolumn{2}{c}{\textbf{Point}}
      & \multicolumn{2}{c}{\textbf{Variance}}
      & \multicolumn{2}{c}{\textbf{$z$-based}}
      & \multicolumn{2}{c}{\textbf{$t$-based}} \\
      \cmidrule(lr){3-4}\cmidrule(lr){5-6}\cmidrule(lr){7-8}\cmidrule(lr){9-10}
      \textbf{$I^2$}
      & \textbf{Method}
      & \textbf{\%Bias} & \textbf{\%MSE}
      & \textbf{\%Bias} & \textbf{\%MSE}
      & \textbf{CP}    & \textbf{AW}
      & \textbf{CP}    & \textbf{AW} \\
      \midrule
      $0\%$   & DiVE    &  2.004 & 0.578 &   3.641 & 27.125 & 0.917 & 3.614 & 0.949 & 4.171 \\
              & QE--RE  & -1.853 & 0.566 &  17.744 & 12.462 & 0.950 & 3.947 & 0.975 & 4.556 \\
              & QE--FE  & -2.581 & 0.617 &  0.996 &  0.870 & 0.928 & 3.678 & 0.959 & 4.245 \\
      \midrule
      $25\%$  & DiVE    &  1.898 & 0.737 &   2.436 & 23.708 & 0.916 & 4.161 & 0.953 & 4.803 \\
              & QE--RE  & -1.334 & 0.700 &   6.568 & 14.930 & 0.944 & 4.307 & 0.966 & 4.971 \\
              & QE--FE  & -2.687 & 0.777 & -24.190 &  6.287 & 0.886 & 3.680 & 0.929 & 4.247 \\
      \midrule
      $50\%$  & DiVE    &  1.517 & 1.175 &   1.164 & 22.047 & 0.917 & 5.074 & 0.952 & 5.857 \\
              & QE--RE  & -0.867 & 1.145 &   0.632 & 18.782 & 0.913 & 5.086 & 0.955 & 5.871 \\
              & QE--FE  & -3.117 & 1.286 & -49.815 & 25.030 & 0.793 & 3.667 & 0.844 & 4.232 \\
      \midrule
      $75\%$  & DiVE    & 1.659 & 2.378 &   1.329 & 24.357 & 0.912 & 7.165 & 0.953 & 8.270 \\
              & QE--RE  & 0.400 & 2.321 &   -0.217 & 23.325 & 0.912 & 7.116 & 0.952 & 8.213 \\
              & QE--FE  & -2.951 & 2.463 & -74.832 & 56.053 & 0.638 & 3.672 & 0.708 & 4.238 \\
      \bottomrule
    \end{tabular}
    \begin{tablenotes}[flushleft]\footnotesize
        \item \textit{Note.} Results are based on 1,000 simulation replicates per heterogeneity level.
        \%Bias, relative bias; \%MSE, relative mean-squared error; CP, empirical coverage probability of 95\% confidence intervals; AW, average width.
        $z$-based intervals use the standard normal critical value, while $t$-based intervals use quantiles from a $t$-distribution with $N\!-\!1$ degrees of freedom.
        The true pooled difference is the benchmark for point-estimate metrics; variance metrics are evaluated against method-specific analytic targets defined in Appendix~B.
      \end{tablenotes}
  \end{threeparttable}
\end{table}

\begin{table}[t]
  \centering
  \caption{Performance of point and variance estimators under a log-normal scenario ($N=30$, fixed sample sizes $n=100$).}
  \label{tab:res_F_logN30n100}
  \renewcommand{\arraystretch}{1.15}
  \begin{threeparttable}
    \begin{tabular}{@{}cc*{8}{S}@{}}
      \toprule
      &
      & \multicolumn{2}{c}{\textbf{Point}}
      & \multicolumn{2}{c}{\textbf{Variance}}
      & \multicolumn{2}{c}{\textbf{$z$-based}}
      & \multicolumn{2}{c}{\textbf{$t$-based}} \\
      \cmidrule(lr){3-4}\cmidrule(lr){5-6}\cmidrule(lr){7-8}\cmidrule(lr){9-10}
      \textbf{$I^2$}
      & \textbf{Method}
      & \textbf{\%Bias} & \textbf{\%MSE}
      & \textbf{\%Bias} & \textbf{\%MSE}
      & \textbf{CP}    & \textbf{AW}
      & \textbf{CP}    & \textbf{AW} \\
      \midrule
      $0\%$   & DiVE    &  4.098 & 0.761 &   7.380 & 11.663 & 0.912 & 3.752 & 0.929 & 3.915 \\
              & QE--RE  & -8.273 & 1.261 &  2.334 &  3.141 & 0.779 & 3.694 & 0.797 & 3.854 \\
              & QE--FE  & -9.622 & 1.525 &  -7.913 &  1.310 & 0.692 & 3.512 & 0.720 & 3.665 \\
      \midrule
      $25\%$  & DiVE    &  4.246 & 0.897 &   4.750 & 9.861 & 0.926 & 4.284 & 0.935 & 4.471 \\
              & QE--RE  & -5.825 & 1.032 &  -3.075 &  6.629 & 0.882 & 4.131 & 0.896 & 4.311 \\
              & QE--FE  & -9.255 & 1.615 & -30.555 & 9.698 & 0.696 & 3.522 & 0.711 & 3.675 \\
      \midrule
      $50\%$  & DiVE    &  5.007 & 1.421 &   3.122 &  8.716 & 0.907 & 5.209 & 0.924 & 5.435 \\
              & QE--RE  & -2.010 & 1.154 &  -0.243 &  8.393 & 0.924 & 5.121 & 0.933 & 5.344 \\
              & QE--FE  & -8.839 & 1.996 & -53.691 & 29.002 & 0.657 & 3.522 & 0.675 & 3.676 \\
      \midrule
      $75\%$  & DiVE    &  4.266 & 2.333 &   1.664 &  7.648 & 0.935 & 7.319 & 0.950 & 7.638 \\
              & QE--RE  &  0.511 & 2.141 &   1.090 &  8.714 & 0.946 & 7.290 & 0.955 & 7.607 \\
              & QE--FE  & -9.489 & 3.496 & -76.842 & 59.091 & 0.545 & 3.523 & 0.559 & 3.676 \\
      \bottomrule
    \end{tabular}
    \begin{tablenotes}[flushleft]\footnotesize
        \item \textit{Note.} Results are based on 1,000 simulation replicates per heterogeneity level.
        \%Bias, relative bias; \%MSE, relative mean-squared error; CP, empirical coverage probability of 95\% confidence intervals; AW, average width.
        $z$-based intervals use the standard normal critical value, while $t$-based intervals use quantiles from a $t$-distribution with $N\!-\!1$ degrees of freedom.
        The true pooled difference is the benchmark for point-estimate metrics; variance metrics are evaluated against method-specific analytic targets defined in Appendix~B.
    \end{tablenotes}
  \end{threeparttable}
\end{table}

\begin{table}[t]
  \centering
  \caption{Performance of point and variance estimators under a log-normal scenario ($N=30$, fixed sample sizes $n=300$).}
  \label{tab:res_F_logN30n300}
  \renewcommand{\arraystretch}{1.15}
  \begin{threeparttable}
    \begin{tabular}{@{}cc*{8}{S}@{}}
      \toprule
      &
      & \multicolumn{2}{c}{\textbf{Point}}
      & \multicolumn{2}{c}{\textbf{Variance}}
      & \multicolumn{2}{c}{\textbf{$z$-based}}
      & \multicolumn{2}{c}{\textbf{$t$-based}} \\
      \cmidrule(lr){3-4}\cmidrule(lr){5-6}\cmidrule(lr){7-8}\cmidrule(lr){9-10}
      \textbf{$I^2$}
      & \textbf{Method}
      & \textbf{\%Bias} & \textbf{\%MSE}
      & \textbf{\%Bias} & \textbf{\%MSE}
      & \textbf{CP}    & \textbf{AW}
      & \textbf{CP}    & \textbf{AW} \\
      \midrule
      $0\%$   & DiVE    &  1.558 & 0.213 &   3.167 &  8.401 & 0.935 & 2.128 & 0.946 & 2.220 \\
              & QE--RE  & -2.959 & 0.266 &   8.710 &  3.651 & 0.902 & 2.199 & 0.909 & 2.295 \\
              & QE--FE  & -3.394 & 0.292 &  -0.237 &  0.278 & 0.874 & 2.112 & 0.881 & 2.204 \\
      \midrule
      $25\%$  & DiVE    &  1.796 & 0.284 &   2.083 &  7.550 & 0.925 & 2.445 & 0.938 & 2.552 \\
              & QE--RE  & -1.871 & 0.277 &  0.687 &  5.667 & 0.911 & 2.435 & 0.925 & 2.541 \\
              & QE--FE  & -3.100 & 0.350 & -24.928 &  6.362 & 0.832 & 2.116 & 0.852 & 2.208 \\
      \midrule
      $50\%$  & DiVE    &  1.540 & 0.385 &   1.313 &  7.194 & 0.926 & 2.984 & 0.939 & 3.114 \\
              & QE--RE  & -1.002 & 0.362 &  -0.099 &  7.194 & 0.934 & 2.963 & 0.946 & 3.092 \\
              & QE--FE  & -3.360 & 0.485 & -49.903 & 24.972 & 0.761 & 2.117 & 0.780 & 2.209 \\
      \midrule
      $75\%$  & DiVE    &  1.478 & 0.737 &   2.424 &  7.118 & 0.945 & 4.246 & 0.956 & 4.431 \\
              & QE--RE  &  0.160 & 0.708 &   1.795 &  7.256 & 0.953 & 4.232 & 0.957 & 4.416 \\
              & QE--FE  & -3.575 & 0.866 & -75.020 & 56.297 & 0.624 & 2.114 & 0.645 & 2.206 \\
      \bottomrule
    \end{tabular}
    \begin{tablenotes}[flushleft]\footnotesize
        \item \textit{Note.} Results are based on 1,000 simulation replicates per heterogeneity level.
        \%Bias, relative bias; \%MSE, relative mean-squared error; CP, empirical coverage probability of 95\% confidence intervals; AW, average width.
        $z$-based intervals use the standard normal critical value, while $t$-based intervals use quantiles from a $t$-distribution with $N\!-\!1$ degrees of freedom.
        The true pooled difference is the benchmark for point-estimate metrics; variance metrics are evaluated against method-specific analytic targets defined in Appendix~B.
    \end{tablenotes}
  \end{threeparttable}
\end{table}


\begin{figure}[t]
  \centering
  \includegraphics[width=\linewidth]{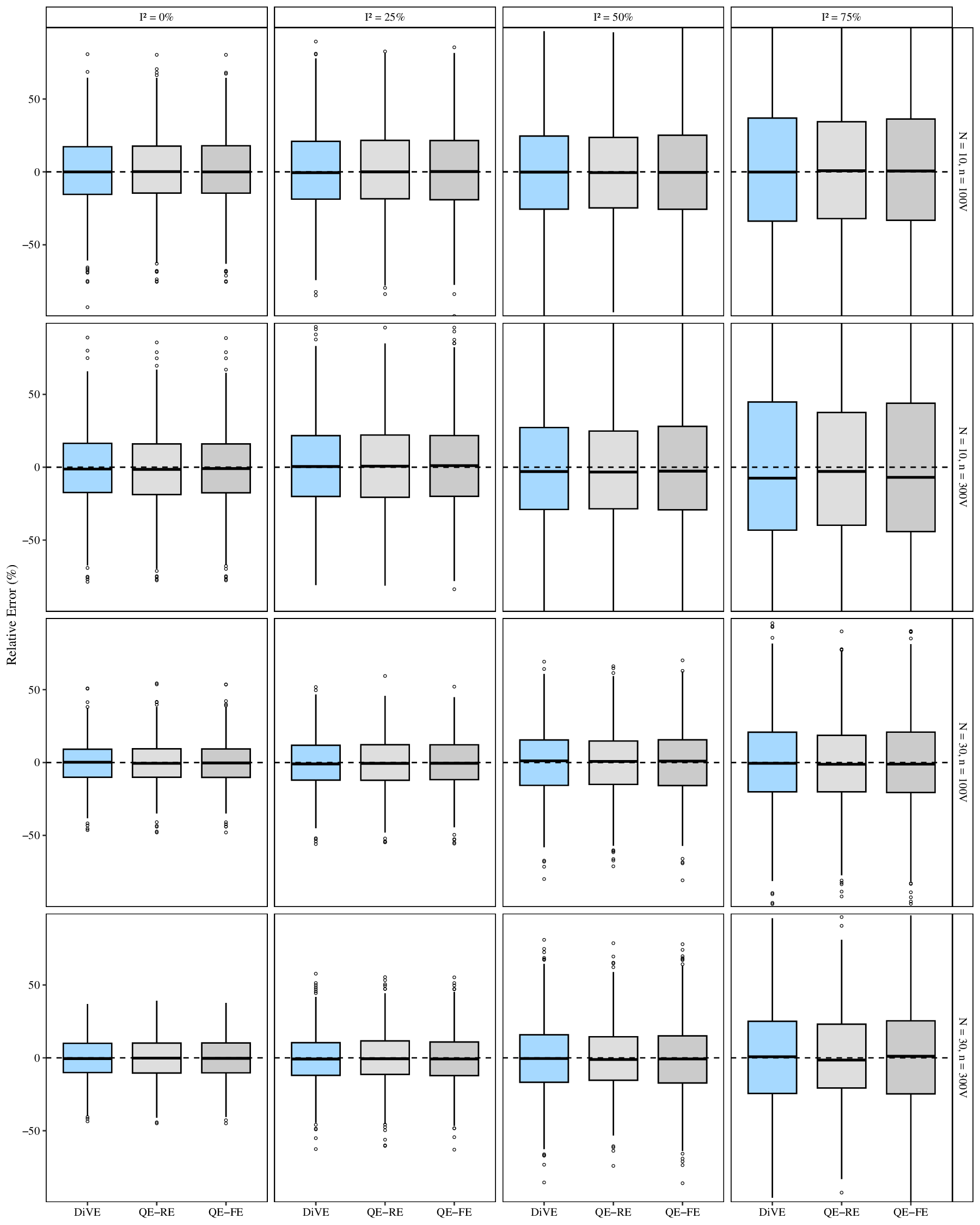}
  \caption{Distribution of relative errors for point estimate (normal scenario with varying sample sizes).\\[0.5em]
    {\footnotesize
      \textit{Note.} Panels show relative errors for the point estimator; methods compared are DiVE, the QE method under a RE model (QE--RE), and the QE method under an FE model (QE--FE).
      The four columns of panels correspond to $I^2\in\{0,25,50,75\}\%$.
      The central line in each box denotes the median error; boxes span the interquartile range; the dashed reference line at zero indicates no error.
    }
  }
  \label{fig:res_box.p_norm_V}
\end{figure}

\begin{figure}[t]
  \centering
  \includegraphics[width=\linewidth]{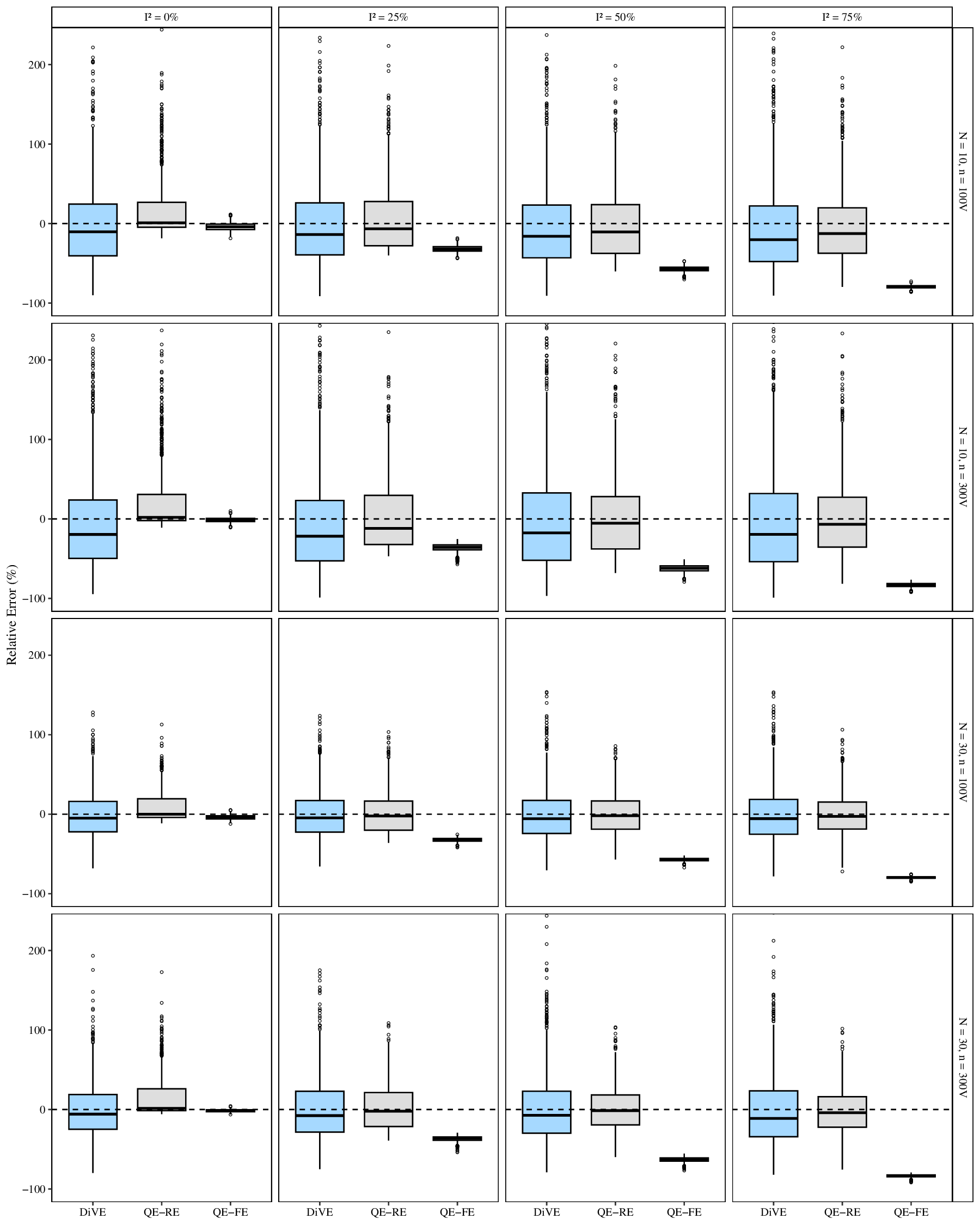}
  \caption{Distribution of relative errors for variance estimate (normal scenario with varying sample sizes).\\[0.5em]
    {\footnotesize
      \textit{Note.} Panels show relative errors for the variance estimator; methods compared are DiVE, the QE method under a RE model (QE--RE), and the QE method under an FE model (QE--FE).
      The four columns of panels correspond to $I^2\in\{0,25,50,75\}\%$.
      The central line in each box denotes the median error; boxes span the interquartile range; the dashed reference line at zero indicates no error.
    }
  }
  \label{fig:res_box.v_norm_V}
\end{figure}

\FloatBarrier
\begin{table}[t]
  \centering
  \caption{Performance of point and variance estimators under a normal scenario ($N=10$, varying sample sizes with average $n=100$).}
  \label{tab:res_V_normN10n100}
  \renewcommand{\arraystretch}{1.15}
  \begin{threeparttable}
    \begin{tabular}{@{}cc*{8}{S}@{}}
      \toprule
      &
      & \multicolumn{2}{c}{\textbf{Point}}
      & \multicolumn{2}{c}{\textbf{Variance}}
      & \multicolumn{2}{c}{\textbf{$z$-based}}
      & \multicolumn{2}{c}{\textbf{$t$-based}} \\
      \cmidrule(lr){3-4}\cmidrule(lr){5-6}\cmidrule(lr){7-8}\cmidrule(lr){9-10}
      \textbf{$I^2$}
      & \textbf{Method}
      & \textbf{\%Bias} & \textbf{\%MSE}
      & \textbf{\%Bias} & \textbf{\%MSE}
      & \textbf{CP}    & \textbf{AW}
      & \textbf{CP}    & \textbf{AW} \\
      \midrule
      $0\%$   & DiVE    &   0.518 &  6.026 &   -0.744 & 32.388 & 0.919 & 0.211 & 0.949 & 0.244 \\
              & QE--RE  &   0.565 &  6.214 &  17.695 & 17.137 & 0.958 & 0.236 & 0.977 & 0.272 \\
              & QE--FE  &   0.599 &  6.159 &  -3.978 &  0.417 & 0.944 & 0.215 & 0.968 & 0.248 \\
      \midrule
      $25\%$  & DiVE    &  0.994 &  9.172 &   -0.608 & 31.859 & 0.908 & 0.251 & 0.945 & 0.289 \\
              & QE--RE  &  0.890 &  9.044 &   6.444 & 18.421 & 0.933 & 0.263 & 0.962 & 0.303 \\
              & QE--FE  &  0.864 &  9.270 & -31.891 & 10.332 & 0.880 & 0.215 & 0.931 & 0.248 \\
      \midrule
      $50\%$  & DiVE    &   -0.911 & 14.544 &   -2.584 & 34.589 & 0.907 & 0.312 & 0.941 & 0.360 \\
              & QE--RE  &   -0.760 & 14.077 &   -2.467 & 19.252 & 0.918 & 0.308 & 0.947 & 0.356 \\
              & QE--FE  &   -0.989 & 14.661 & -57.160 & 32.775 & 0.820 & 0.215 & 0.863 & 0.248 \\
      \midrule
      $75\%$  & DiVE    &   0.934 & 28.475 &   -3.837 & 43.822 & 0.917 & 0.446 & 0.952 & 0.515 \\
              & QE--RE  &   0.970 & 24.299 &   -2.706 & 24.779 & 0.936 & 0.431 & 0.966 & 0.498 \\
              & QE--FE  &   0.921 & 28.802 & -79.575 & 63.360 & 0.646 & 0.215 & 0.711 & 0.249 \\
      \bottomrule
    \end{tabular}
    \begin{tablenotes}[flushleft]\footnotesize
        \item \textit{Note.} Results are based on 1,000 simulation replicates per heterogeneity level.
        \%Bias, relative bias; \%MSE, relative mean-squared error; CP, empirical coverage probability of 95\% confidence intervals; AW, average width.
        $z$-based intervals use the standard normal critical value, while $t$-based intervals use quantiles from a $t$-distribution with $N\!-\!1$ degrees of freedom.
        The true pooled difference is the benchmark for point-estimate metrics; variance metrics are evaluated against method-specific analytic targets defined in Appendix~B.
    \end{tablenotes}
  \end{threeparttable}
\end{table}

\begin{table}[t]
  \centering
  \caption{Performance of point and variance estimators under a normal scenario ($N=10$, varying sample sizes with average $n=300$).}
  \label{tab:res_V_normN10n300}
  \renewcommand{\arraystretch}{1.15}
  \begin{threeparttable}
    \begin{tabular}{@{}cc*{8}{S}@{}}
      \toprule
      &
      & \multicolumn{2}{c}{\textbf{Point}}
      & \multicolumn{2}{c}{\textbf{Variance}}
      & \multicolumn{2}{c}{\textbf{$z$-based}}
      & \multicolumn{2}{c}{\textbf{$t$-based}} \\
      \cmidrule(lr){3-4}\cmidrule(lr){5-6}\cmidrule(lr){7-8}\cmidrule(lr){9-10}
      \textbf{$I^2$}
      & \textbf{Method}
      & \textbf{\%Bias} & \textbf{\%MSE}
      & \textbf{\%Bias} & \textbf{\%MSE}
      & \textbf{CP}    & \textbf{AW}
      & \textbf{CP}    & \textbf{AW} \\
      \midrule
      $0\%$   & DiVE    &  -1.045 &  6.222 &  -1.954 & 49.263 & 0.889 & 0.119 & 0.927 & 0.137 \\
              & QE--RE  &  -1.391 &  6.473 &  23.023 & 23.689 & 0.959 & 0.139 & 0.977 & 0.160 \\
              & QE--FE  &  -1.102 &  6.226 &  -1.330 &  0.105 & 0.947 & 0.126 & 0.972 & 0.145 \\
      \midrule
      $25\%$  & DiVE    &   1.322 & 9.735 &   -2.138 & 59.449 & 0.874 & 0.147 & 0.911 & 0.169 \\
              & QE--RE  &   1.243 &  9.402 &   4.352 & 20.948 & 0.937 & 0.155 & 0.964 & 0.179 \\
              & QE--FE  &   1.368 & 9.716 & -36.390 & 13.508 & 0.886 & 0.126 & 0.933 & 0.145 \\
      \midrule
      $50\%$  & DiVE    &   -0.441 & 17.074 &  4.229 & 71.684 & 0.887 & 0.197 & 0.925 & 0.227 \\
              & QE--RE  &  -1.542 & 15.327 &  1.171 &  24.439 & 0.918 & 0.190 & 0.950 & 0.219 \\
              & QE--FE  &   -0.560 & 17.102 & -62.517 &  39.292 & 0.778 & 0.126 & 0.830 & 0.145 \\
      \midrule
      $75\%$  & DiVE    &  -1.818 & 38.395 &  2.371 &  76.239 & 0.880 & 0.293 & 0.920 & 0.338 \\
              & QE--RE  &  -0.205 & 30.827 &  1.195 &  27.173 & 0.904 & 0.267 & 0.937 & 0.308 \\
              & QE--FE  &  -1.943 & 38.582 & -83.301 &  69.470 & 0.563 & 0.126 & 0.643 & 0.145 \\
      \bottomrule
    \end{tabular}
    \begin{tablenotes}[flushleft]\footnotesize
        \item \textit{Note.} Results are based on 1,000 simulation replicates per heterogeneity level.
        \%Bias, relative bias; \%MSE, relative mean-squared error; CP, empirical coverage probability of 95\% confidence intervals; AW, average width.
        $z$-based intervals use the standard normal critical value, while $t$-based intervals use quantiles from a $t$-distribution with $N\!-\!1$ degrees of freedom.
        The true pooled difference is the benchmark for point-estimate metrics; variance metrics are evaluated against method-specific analytic targets defined in Appendix~B.
    \end{tablenotes}
  \end{threeparttable}
\end{table}

\begin{table}[t]
  \centering
  \caption{Performance of point and variance estimators under a normal scenario ($N=30$, varying sample sizes with average $n=100$).}
  \label{tab:res_V_normN30n100}
  \renewcommand{\arraystretch}{1.15}
  \begin{threeparttable}
    \begin{tabular}{@{}cc*{8}{S}@{}}
      \toprule
      &
      & \multicolumn{2}{c}{\textbf{Point}}
      & \multicolumn{2}{c}{\textbf{Variance}}
      & \multicolumn{2}{c}{\textbf{$z$-based}}
      & \multicolumn{2}{c}{\textbf{$t$-based}} \\
      \cmidrule(lr){3-4}\cmidrule(lr){5-6}\cmidrule(lr){7-8}\cmidrule(lr){9-10}
      \textbf{$I^2$}
      & \textbf{Method}
      & \textbf{\%Bias} & \textbf{\%MSE}
      & \textbf{\%Bias} & \textbf{\%MSE}
      & \textbf{CP}    & \textbf{AW}
      & \textbf{CP}    & \textbf{AW} \\
      \midrule
      $0\%$   & DiVE    &   -0.333 & 2.105 &   -0.960 &  8.516 & 0.944 & 0.125 & 0.949 & 0.130 \\
              & QE--RE  &   -0.397 & 2.110 &  9.105 &  4.467 & 0.954 & 0.132 & 0.964 & 0.138 \\
              & QE--FE  &   -0.372 & 2.110 &  -3.976 &  0.245 & 0.943 & 0.124 & 0.955 & 0.130 \\
      \midrule
      $25\%$  & DiVE    &   -0.156 & 2.922 &   -0.120 & 9.811 & 0.942 & 0.149 & 0.954 & 0.156 \\
              & QE--RE  &   -0.280 & 2.875 &   0.988 &  6.706 & 0.949 & 0.149 & 0.960 & 0.156 \\
              & QE--FE  &   -0.242 & 2.965 & -32.185 &  10.418 & 0.900 & 0.124 & 0.915 & 0.130 \\
      \midrule
      $50\%$  & DiVE    &   0.128 & 5.230 &  -0.062 & 11.886 & 0.932 & 0.187 & 0.941 & 0.196 \\
              & QE--RE  &  0.112 & 4.882 &  0.297 &  7.055 & 0.937 & 0.183 & 0.945 & 0.191 \\
              & QE--FE  &   0.149 & 5.276 & -57.211 & 32.774 & 0.769 & 0.124 & 0.790 & 0.130 \\
      \midrule
      $75\%$  & DiVE    &   -0.309 & 9.682 &  -0.294 & 13.022 & 0.946 & 0.272 & 0.956 & 0.284 \\
              & QE--RE  &   -0.674 &  8.534 &   0.145 &  7.348 & 0.948 & 0.257 & 0.955 & 0.268 \\
              & QE--FE  &   -0.545 & 9.640 & -79.742 & 63.603 & 0.640 & 0.124 & 0.657 & 0.130 \\
      \bottomrule
    \end{tabular}
    \begin{tablenotes}[flushleft]\footnotesize
        \item \textit{Note.} Results are based on 1,000 simulation replicates per heterogeneity level.
        \%Bias, relative bias; \%MSE, relative mean-squared error; CP, empirical coverage probability of 95\% confidence intervals; AW, average width.
        $z$-based intervals use the standard normal critical value, while $t$-based intervals use quantiles from a $t$-distribution with $N\!-\!1$ degrees of freedom.
        The true pooled difference is the benchmark for point-estimate metrics; variance metrics are evaluated against method-specific analytic targets defined in Appendix~B.
    \end{tablenotes}
  \end{threeparttable}
\end{table}

\begin{table}[t]
  \centering
  \caption{Performance of point and variance estimators under a normal scenario ($N=30$, varying sample sizes with average $n=300$).}
  \label{tab:res_V_normN30n300}
  \renewcommand{\arraystretch}{1.15}
  \begin{threeparttable}
    \begin{tabular}{@{}cc*{8}{S}@{}}
      \toprule
      &
      & \multicolumn{2}{c}{\textbf{Point}}
      & \multicolumn{2}{c}{\textbf{Variance}}
      & \multicolumn{2}{c}{\textbf{$z$-based}}
      & \multicolumn{2}{c}{\textbf{$t$-based}} \\
      \cmidrule(lr){3-4}\cmidrule(lr){5-6}\cmidrule(lr){7-8}\cmidrule(lr){9-10}
      \textbf{$I^2$}
      & \textbf{Method}
      & \textbf{\%Bias} & \textbf{\%MSE}
      & \textbf{\%Bias} & \textbf{\%MSE}
      & \textbf{CP}    & \textbf{AW}
      & \textbf{CP}    & \textbf{AW} \\
      \midrule
      $0\%$   & DiVE    &  -0.308 & 2.151 &   -0.290 & 12.403 & 0.933 & 0.072 & 0.940 & 0.075 \\
              & QE--RE  &  -0.441 & 2.212 &  15.319 &  8.571 & 0.965 & 0.078 & 0.976 & 0.082 \\
              & QE--FE  &  -0.355 & 2.168 &  -1.492 &  0.051 & 0.951 & 0.073 & 0.964 & 0.076 \\
      \midrule
      $25\%$  & DiVE    &  -0.818 & 3.126 &  -0.437 & 15.340 & 0.929 & 0.090 & 0.939 & 0.094 \\
              & QE--RE  &  -0.501 & 3.129 &  1.688 &  8.124 & 0.933 & 0.089 & 0.944 & 0.093 \\
              & QE--FE  &  -0.726 & 3.151 & -36.847 & 13.689 & 0.882 & 0.073 & 0.892 & 0.076 \\
      \midrule
      $50\%$  & DiVE    &  0.116 & 5.884 &   1.810 & 21.259 & 0.929 & 0.118 & 0.935 & 0.123 \\
              & QE--RE  &  -0.302 & 4.867 &   0.388 &  7.994 & 0.948 & 0.111 & 0.954 & 0.115 \\
              & QE--FE  &  0.035 & 5.850 & -63.144 & 39.971 & 0.770 & 0.073 & 0.791 & 0.076 \\
      \midrule
      $75\%$  & DiVE    &   0.160 & 12.507 &   -0.564 & 22.380 & 0.938 & 0.175 & 0.951 & 0.183 \\
              & QE--RE  &   0.706 &  9.748 &   -1.584 &  7.711 & 0.938 & 0.154 & 0.947 & 0.160 \\
              & QE--FE  &   0.211 & 12.484 & -83.630 & 69.972 & 0.565 & 0.073 & 0.584 & 0.076 \\
      \bottomrule
    \end{tabular}
    \begin{tablenotes}[flushleft]\footnotesize
        \item \textit{Note.} Results are based on 1,000 simulation replicates per heterogeneity level.
        \%Bias, relative bias; \%MSE, relative mean-squared error; CP, empirical coverage probability of 95\% confidence intervals; AW, average width.
        $z$-based intervals use the standard normal critical value, while $t$-based intervals use quantiles from a $t$-distribution with $N\!-\!1$ degrees of freedom.
        The true pooled difference is the benchmark for point-estimate metrics; variance metrics are evaluated against method-specific analytic targets defined in Appendix~B.
    \end{tablenotes}
  \end{threeparttable}
\end{table}

\begin{figure}[t]
  \centering
  \includegraphics[width=\linewidth]{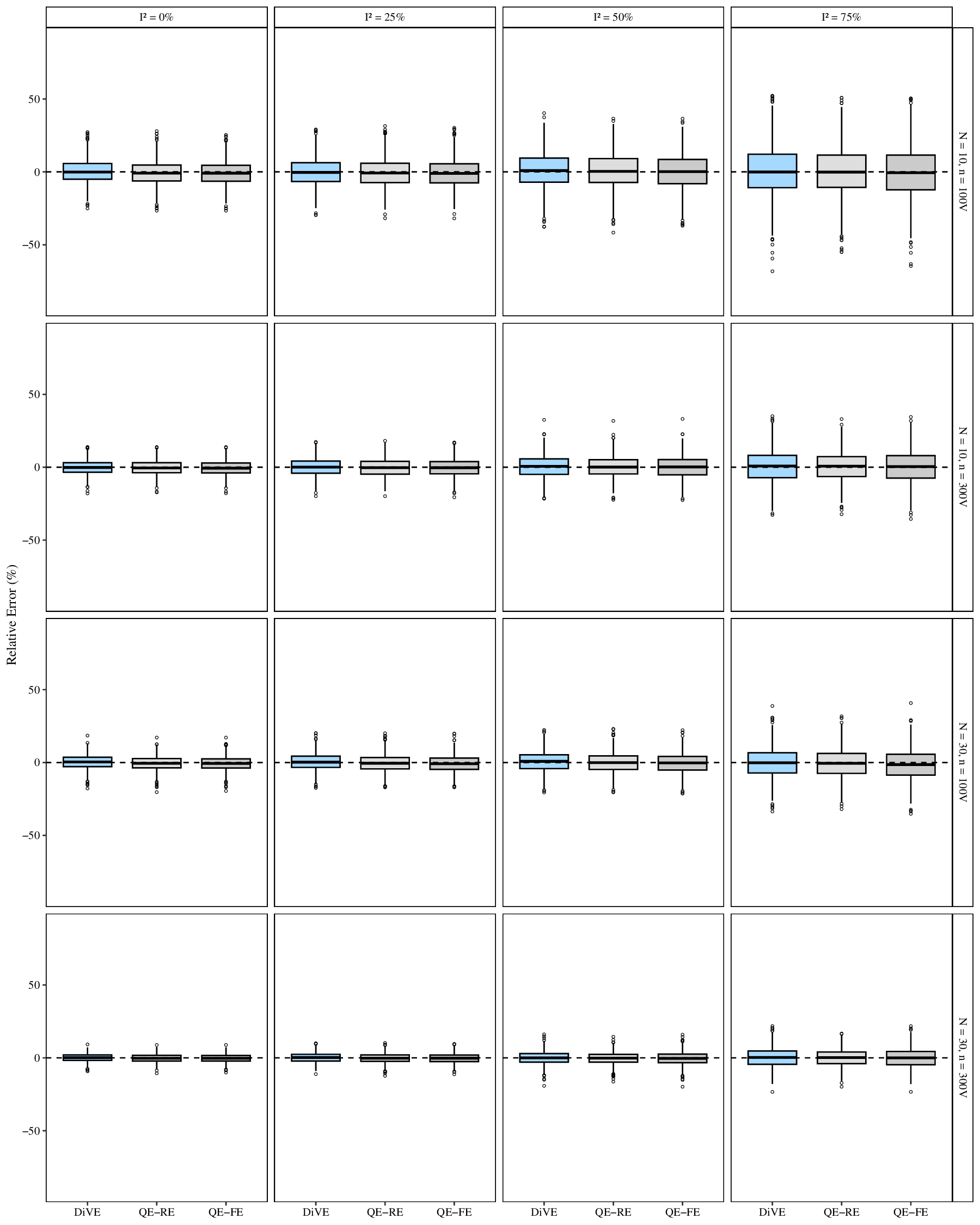}
  \caption{Distribution of relative errors for point estimate (skew-normal scenario with varying sample sizes).\\[0.5em]
    {\footnotesize
      \textit{Note.} Panels show relative errors for the point estimator; methods compared are DiVE, the QE method under a RE model (QE--RE), and the QE method under an FE model (QE--FE).
      The four columns of panels correspond to $I^2\in\{0,25,50,75\}\%$.
      The central line in each box denotes the median error; boxes span the interquartile range; the dashed reference line at zero indicates no error.
    }
  }
  \label{fig:res_box.p_skew_V}
\end{figure}

\begin{figure}[t]
  \centering
  \includegraphics[width=\linewidth]{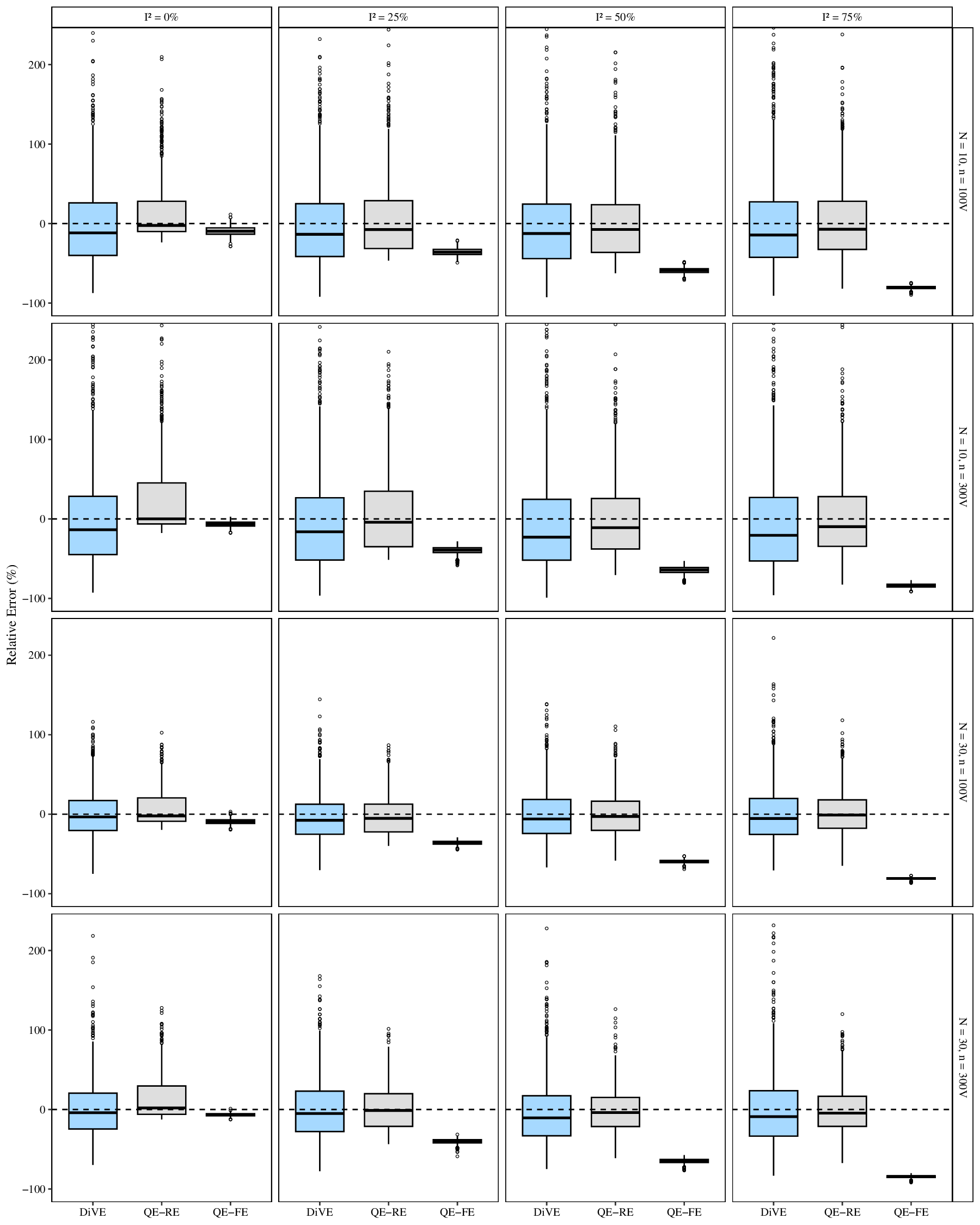}
  \caption{Distribution of relative errors for variance estimate (skew-normal scenario with varying sample sizes).\\[0.5em]
    {\footnotesize
      \textit{Note.} Panels show relative errors for the variance estimator; methods compared are DiVE, the QE method under a RE model (QE--RE), and the QE method under an FE model (QE--FE).
      The four columns of panels correspond to $I^2\in\{0,25,50,75\}\%$.
      The central line in each box denotes the median error; boxes span the interquartile range; the dashed reference line at zero indicates no error.
    }
  }
  \label{fig:res_box.v_skew_V}
\end{figure}

\FloatBarrier
\begin{table}[t]
  \centering
  \caption{Performance of point and variance estimators under a skew-normal scenario ($N=10$, varying sample sizes with average $n=100$).}
  \label{tab:res_V_skewN10n100}
  \renewcommand{\arraystretch}{1.15}
  \begin{threeparttable}
    \begin{tabular}{@{}cc*{8}{S}@{}}
      \toprule
      &
      & \multicolumn{2}{c}{\textbf{Point}}
      & \multicolumn{2}{c}{\textbf{Variance}}
      & \multicolumn{2}{c}{\textbf{$z$-based}}
      & \multicolumn{2}{c}{\textbf{$t$-based}} \\
      \cmidrule(lr){3-4}\cmidrule(lr){5-6}\cmidrule(lr){7-8}\cmidrule(lr){9-10}
      \textbf{$I^2$}
      & \textbf{Method}
      & \textbf{\%Bias} & \textbf{\%MSE}
      & \textbf{\%Bias} & \textbf{\%MSE}
      & \textbf{CP}    & \textbf{AW}
      & \textbf{CP}    & \textbf{AW} \\
      \midrule
      $0\%$   & DiVE    &   0.278 & 0.686 &  -1.736 & 28.554 & 0.909 & 1.044 & 0.940 & 1.205 \\
              & QE--RE  &  -0.604 & 0.717 &  14.792 & 16.794 & 0.948 & 1.155 & 0.982 & 1.333 \\
              & QE--FE  &  -0.828 & 0.704 & -9.369 &  1.211 & 0.930 & 1.038 & 0.965 & 1.198 \\
      \midrule
      $25\%$  & DiVE    &   -0.085 & 1.013 &   -0.391 & 36.042 & 0.890 & 1.242 & 0.926 & 1.434 \\
              & QE--RE  &  -0.784 & 1.032 &   5.529 & 21.309 & 0.927 & 1.294 & 0.953 & 1.494 \\
              & QE--FE  &  -1.150 & 1.031 & -35.759 & 12.997 & 0.873 & 1.038 & 0.916 & 1.198 \\
      \midrule
      $50\%$  & DiVE    &  0.993 & 1.542 &  -0.444 & 41.421 & 0.905 & 1.553 & 0.932 & 1.793 \\
              & QE--RE  &  0.477 & 1.507 &   -0.941 & 22.315 & 0.908 & 1.538 & 0.938 & 1.775 \\
              & QE--FE  &  -0.035 & 1.570 & -59.086 & 35.033 & 0.774 & 1.040 & 0.833 & 1.201 \\
      \midrule
      $75\%$  & DiVE    &   0.366 & 3.151 &   2.266 & 46.853 & 0.908 & 2.281 & 0.940 & 2.633 \\
              & QE--RE  &   0.212 & 2.912 &   1.710 & 25.511 & 0.916 & 2.190 & 0.950 & 2.528 \\
              & QE--FE  &  -0.538 & 3.210 & -80.582 & 64.974 & 0.628 & 1.040 & 0.707 & 1.200 \\
      \bottomrule
    \end{tabular}
    \begin{tablenotes}[flushleft]\footnotesize
        \item \textit{Note.} Results are based on 1,000 simulation replicates per heterogeneity level.
        \%Bias, relative bias; \%MSE, relative mean-squared error; CP, empirical coverage probability of 95\% confidence intervals; AW, average width.
        $z$-based intervals use the standard normal critical value, while $t$-based intervals use quantiles from a $t$-distribution with $N\!-\!1$ degrees of freedom.
        The true pooled difference is the benchmark for point-estimate metrics; variance metrics are evaluated against method-specific analytic targets defined in Appendix~B.
    \end{tablenotes}
  \end{threeparttable}
\end{table}

\begin{table}[t]
  \centering
  \caption{Performance of point and variance estimators under a skew-normal scenario ($N=10$, varying sample sizes with average $n=300$).}
  \label{tab:res_V_skewN10n300}
  \renewcommand{\arraystretch}{1.15}
  \begin{threeparttable}
    \begin{tabular}{@{}cc*{8}{S}@{}}
      \toprule
      &
      & \multicolumn{2}{c}{\textbf{Point}}
      & \multicolumn{2}{c}{\textbf{Variance}}
      & \multicolumn{2}{c}{\textbf{$z$-based}}
      & \multicolumn{2}{c}{\textbf{$t$-based}} \\
      \cmidrule(lr){3-4}\cmidrule(lr){5-6}\cmidrule(lr){7-8}\cmidrule(lr){9-10}
      \textbf{$I^2$}
      & \textbf{Method}
      & \textbf{\%Bias} & \textbf{\%MSE}
      & \textbf{\%Bias} & \textbf{\%MSE}
      & \textbf{CP}    & \textbf{AW}
      & \textbf{CP}    & \textbf{AW} \\
      \midrule
      $0\%$   & DiVE    &   -0.103 & 0.241 &   2.983 & 52.278 & 0.898 & 0.606 & 0.938 & 0.700 \\
              & QE--RE  &  -0.341 & 0.252 &  25.505 & 28.466 & 0.950 & 0.695 & 0.979 & 0.802 \\
              & QE--FE  &  -0.480 & 0.246 &  -6.497 &  0.541 & 0.935 & 0.609 & 0.966 & 0.702 \\
      \midrule
      $25\%$  & DiVE    &   0.060 & 0.363 &   -1.747 & 46.220 & 0.878 & 0.734 & 0.912 & 0.848 \\
              & QE--RE  &  -0.300 & 0.365 &  7.101 & 24.261 & 0.924 & 0.778 & 0.957 & 0.899 \\
              & QE--FE  &  -0.334 & 0.362 & -39.647 & 15.950 & 0.869 & 0.609 & 0.924 & 0.703 \\
      \midrule
      $50\%$  & DiVE    &   0.471 & 0.619 &  -2.238 & 63.608 & 0.874 & 0.952 & 0.914 & 1.098 \\
              & QE--RE  &  0.243 & 0.545 &  -0.579 & 24.829 & 0.914 & 0.936 & 0.943 & 1.080 \\
              & QE--FE  &  0.137 & 0.621 & -64.600 & 41.950 & 0.739 & 0.609 & 0.808 & 0.703 \\
      \midrule
      $75\%$  & DiVE    &   0.420 & 1.389 &  1.390 & 67.901 & 0.887 & 1.447 & 0.924 & 1.670 \\
              & QE--RE  &   0.342 & 1.039 &  0.702 & 25.827 & 0.923 & 1.322 & 0.955 & 1.525 \\
              & QE--FE  &   0.020 & 1.377 & -83.987 & 70.608 & 0.572 & 0.609 & 0.639 & 0.702 \\
      \bottomrule
    \end{tabular}
    \begin{tablenotes}[flushleft]\footnotesize
        \item \textit{Note.} Results are based on 1,000 simulation replicates per heterogeneity level.
        \%Bias, relative bias; \%MSE, relative mean-squared error; CP, empirical coverage probability of 95\% confidence intervals; AW, average width.
        $z$-based intervals use the standard normal critical value, while $t$-based intervals use quantiles from a $t$-distribution with $N\!-\!1$ degrees of freedom.
        The true pooled difference is the benchmark for point-estimate metrics; variance metrics are evaluated against method-specific analytic targets defined in Appendix~B.
    \end{tablenotes}
  \end{threeparttable}
\end{table}

\begin{table}[t]
  \centering
  \caption{Performance of point and variance estimators under a skew-normal scenario ($N=30$, varying sample sizes with average $n=100$).}
  \label{tab:res_V_skewN30n100}
  \renewcommand{\arraystretch}{1.15}
  \begin{threeparttable}
    \begin{tabular}{@{}cc*{8}{S}@{}}
      \toprule
      &
      & \multicolumn{2}{c}{\textbf{Point}}
      & \multicolumn{2}{c}{\textbf{Variance}}
      & \multicolumn{2}{c}{\textbf{$z$-based}}
      & \multicolumn{2}{c}{\textbf{$t$-based}} \\
      \cmidrule(lr){3-4}\cmidrule(lr){5-6}\cmidrule(lr){7-8}\cmidrule(lr){9-10}
      \textbf{$I^2$}
      & \textbf{Method}
      & \textbf{\%Bias} & \textbf{\%MSE}
      & \textbf{\%Bias} & \textbf{\%MSE}
      & \textbf{CP}    & \textbf{AW}
      & \textbf{CP}    & \textbf{AW} \\
      \midrule
      $0\%$   & DiVE    &   0.450 & 0.237 &  0.545 &  8.776 & 0.931 & 0.625 & 0.942 & 0.652 \\
              & QE--RE  &  -0.542 & 0.241 &   7.763 &  5.465 & 0.948 & 0.650 & 0.953 & 0.679 \\
              & QE--FE  &  -0.664 & 0.240 & -9.378 &  0.997 & 0.931 & 0.599 & 0.941 & 0.625 \\
      \midrule
      $25\%$  & DiVE    &   0.325 & 0.339 &   -3.928 & 9.250 & 0.920 & 0.725 & 0.937 & 0.757 \\
              & QE--RE  &  -0.484 & 0.338 &   -2.347 &  6.253 & 0.930 & 0.727 & 0.942 & 0.759 \\
              & QE--FE  &  -0.793 & 0.347 & -35.912 & 12.964 & 0.866 & 0.599 & 0.882 & 0.625 \\
      \midrule
      $50\%$  & DiVE    &  0.598 & 0.493 &   -0.305 & 11.786 & 0.940 & 0.929 & 0.949 & 0.969 \\
              & QE--RE  &  0.002 & 0.467 &  -0.266 &  7.664 & 0.945 & 0.904 & 0.953 & 0.943 \\
              & QE--FE  &  -0.439 & 0.499 & -59.547 & 35.501 & 0.792 & 0.600 & 0.813 & 0.626 \\
      \midrule
      $75\%$  & DiVE    &   -0.076 & 1.042 &   0.937 & 14.533 & 0.945 & 1.356 & 0.958 & 1.415 \\
              & QE--RE  &   -0.415 & 0.938 &   1.312 &  7.726 & 0.951 & 1.282 & 0.961 & 1.338 \\
              & QE--FE  &  -1.241 & 1.071 & -80.891 & 65.450 & 0.602 & 0.599 & 0.622 & 0.625 \\

      \bottomrule
    \end{tabular}
    \begin{tablenotes}[flushleft]\footnotesize
        \item \textit{Note.} Results are based on 1,000 simulation replicates per heterogeneity level.
        \%Bias, relative bias; \%MSE, relative mean-squared error; CP, empirical coverage probability of 95\% confidence intervals; AW, average width.
        $z$-based intervals use the standard normal critical value, while $t$-based intervals use quantiles from a $t$-distribution with $N\!-\!1$ degrees of freedom.
        The true pooled difference is the benchmark for point-estimate metrics; variance metrics are evaluated against method-specific analytic targets defined in Appendix~B.
    \end{tablenotes}
  \end{threeparttable}
\end{table}

\begin{table}[t]
  \centering
  \caption{Performance of point and variance estimators under a skew-normal scenario ($N=30$, varying sample sizes with average $n=300$).}
  \label{tab:res_V_skewN30n300}
  \renewcommand{\arraystretch}{1.15}
  \begin{threeparttable}
    \begin{tabular}{@{}cc*{8}{S}@{}}
      \toprule
      &
      & \multicolumn{2}{c}{\textbf{Point}}
      & \multicolumn{2}{c}{\textbf{Variance}}
      & \multicolumn{2}{c}{\textbf{$z$-based}}
      & \multicolumn{2}{c}{\textbf{$t$-based}} \\
      \cmidrule(lr){3-4}\cmidrule(lr){5-6}\cmidrule(lr){7-8}\cmidrule(lr){9-10}
      \textbf{$I^2$}
      & \textbf{Method}
      & \textbf{\%Bias} & \textbf{\%MSE}
      & \textbf{\%Bias} & \textbf{\%MSE}
      & \textbf{CP}    & \textbf{AW}
      & \textbf{CP}    & \textbf{AW} \\
      \midrule
      $0\%$   & DiVE    &   0.174 & 0.079 &  1.610 & 13.520 & 0.925 & 0.361 & 0.935 & 0.377 \\
              & QE--RE  &  -0.222 & 0.082 &  14.582 & 9.480 & 0.956 & 0.387 & 0.964 & 0.403 \\
              & QE--FE  &  -0.253 & 0.080 &  -6.587 &  0.470 & 0.941 & 0.351 & 0.948 & 0.367 \\
      \midrule
      $25\%$  & DiVE    &   0.114 & 0.119 &  0.755 & 15.470 & 0.936 & 0.448 & 0.948 & 0.467 \\
              & QE--RE  &  -0.244 & 0.117 &  1.800 &  8.225 & 0.937 & 0.443 & 0.949 & 0.463 \\
              & QE--FE  &  -0.328 & 0.121 & -40.060 & 16.147 & 0.859 & 0.351 & 0.882 & 0.367 \\
      \midrule
      $50\%$  & DiVE    &   0.042 & 0.202 &  -2.155 & 21.819 & 0.919 & 0.574 & 0.927 & 0.599 \\
              & QE--RE  &  -0.221 & 0.173 &  -2.013 &  8.055 & 0.936 & 0.542 & 0.943 & 0.566 \\
              & QE--FE  &  -0.380 & 0.203 & -64.940 & 42.261 & 0.774 & 0.351 & 0.793 & 0.367 \\
      \midrule
      $75\%$  & DiVE    &   0.218 & 0.448 &  0.455 & 24.271 & 0.941 & 0.873 & 0.950 & 0.911 \\
              & QE--RE  &   0.096 & 0.350 &  -1.117 &  8.171 & 0.941 & 0.763 & 0.959 & 0.796 \\
              & QE--FE  &  -0.137 & 0.450 & -84.441 & 71.333 & 0.550 & 0.352 & 0.569 & 0.367 \\
      \bottomrule
    \end{tabular}
    \begin{tablenotes}[flushleft]\footnotesize
        \item \textit{Note.} Results are based on 1,000 simulation replicates per heterogeneity level.
        \%Bias, relative bias; \%MSE, relative mean-squared error; CP, empirical coverage probability of 95\% confidence intervals; AW, average width.
        $z$-based intervals use the standard normal critical value, while $t$-based intervals use quantiles from a $t$-distribution with $N\!-\!1$ degrees of freedom.
        The true pooled difference is the benchmark for point-estimate metrics; variance metrics are evaluated against method-specific analytic targets defined in Appendix~B.
    \end{tablenotes}
  \end{threeparttable}
\end{table}

\begin{figure}[t]
  \centering
  \includegraphics[width=\linewidth]{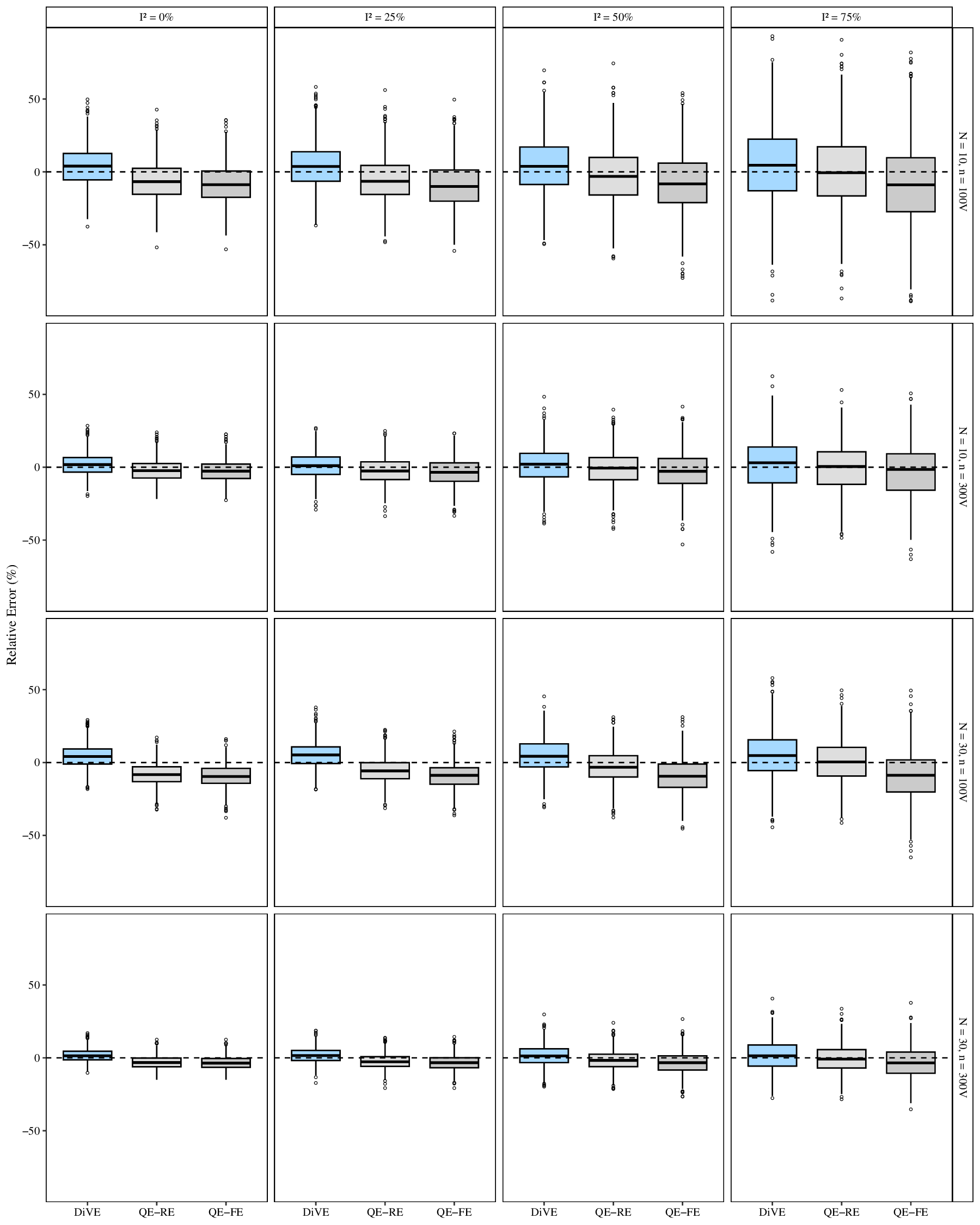}
  \caption{Distribution of relative errors for point estimate (log-normal scenario with varying sample sizes).\\[0.5em]
    {\footnotesize
      \textit{Note.} Panels show relative errors for the point estimator; methods compared are DiVE, the QE method under a RE model (QE--RE), and the QE method under an FE model (QE--FE).
      The four columns of panels correspond to $I^2\in\{0,25,50,75\}\%$.
      The central line in each box denotes the median error; boxes span the interquartile range; the dashed reference line at zero indicates no error.
    }
  }
  \label{fig:res_box.p_log_V}
\end{figure}

\begin{figure}[t]
  \centering
  \includegraphics[width=\linewidth]{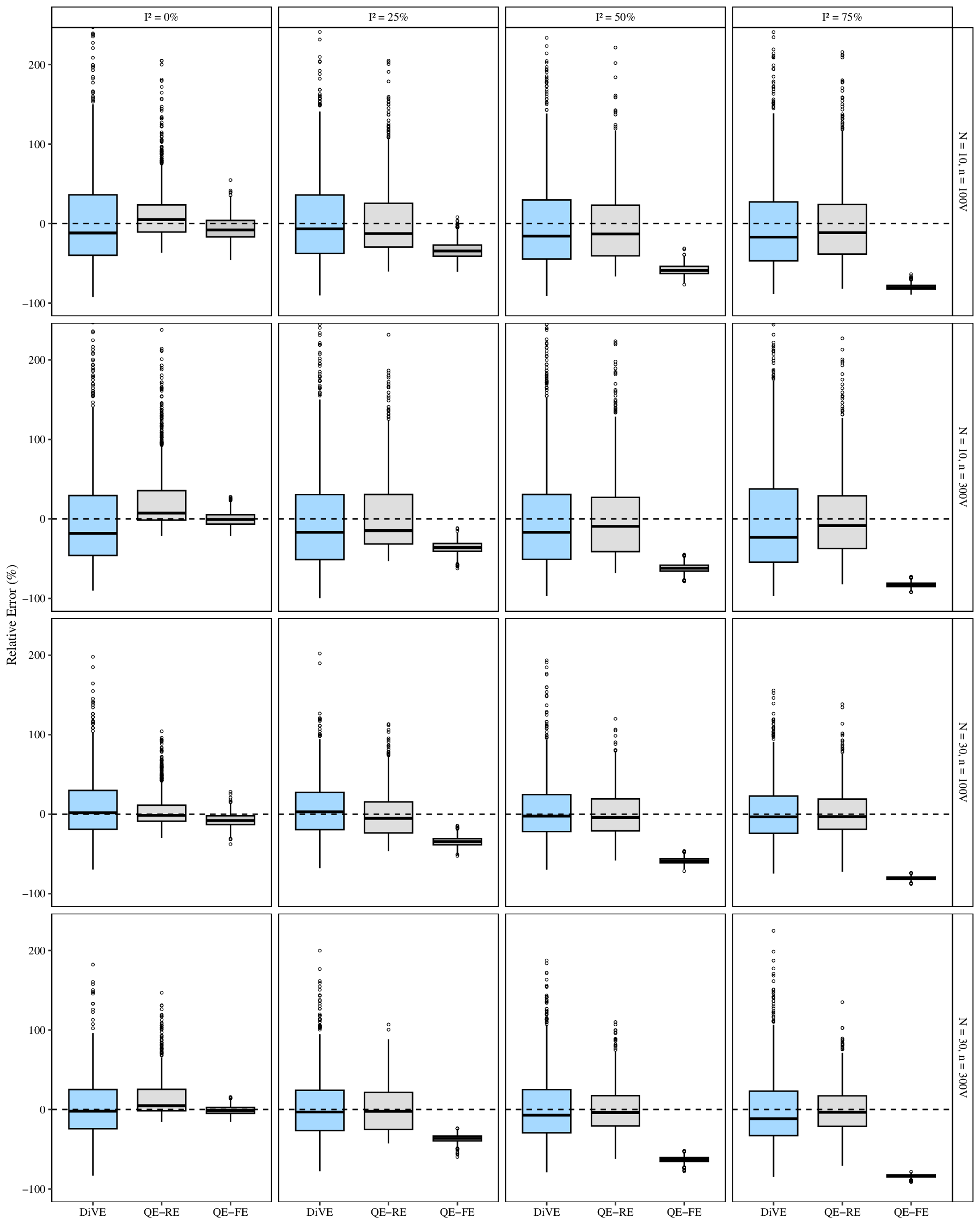}
  \caption{Distribution of relative errors for variance estimate (log-normal scenario with varying sample sizes).\\[0.5em]
    {\footnotesize
      \textit{Note.} Panels show relative errors for the variance estimator; methods compared are DiVE, the QE method under a RE model (QE--RE), and the QE method under an FE model (QE--FE).
      The four columns of panels correspond to $I^2\in\{0,25,50,75\}\%$.
      The central line in each box denotes the median error; boxes span the interquartile range; the dashed reference line at zero indicates no error.
    }
  }
  \label{fig:res_box.v_log_V}
\end{figure}

\FloatBarrier
\begin{table}[t]
  \centering
  \caption{Performance of point and variance estimators under a log-normal scenario ($N=10$, varying sample sizes with average $n=100$).}
  \label{tab:res_V_logN10n100}
  \renewcommand{\arraystretch}{1.15}
  \begin{threeparttable}
    \begin{tabular}{@{}cc*{8}{S}@{}}
      \toprule
      &
      & \multicolumn{2}{c}{\textbf{Point}}
      & \multicolumn{2}{c}{\textbf{Variance}}
      & \multicolumn{2}{c}{\textbf{$z$-based}}
      & \multicolumn{2}{c}{\textbf{$t$-based}} \\
      \cmidrule(lr){3-4}\cmidrule(lr){5-6}\cmidrule(lr){7-8}\cmidrule(lr){9-10}
      \textbf{$I^2$}
      & \textbf{Method}
      & \textbf{\%Bias} & \textbf{\%MSE}
      & \textbf{\%Bias} & \textbf{\%MSE}
      & \textbf{CP}    & \textbf{AW}
      & \textbf{CP}    & \textbf{AW} \\
      \midrule
      $0\%$   & DiVE    &   4.186 & 1.988 &   5.344 & 44.041 & 0.886 &  6.230 & 0.923 &  7.190 \\
              & QE--RE  &  -6.314 & 2.135 &  14.362 & 17.790 & 0.902 &  6.706 & 0.937 &  7.740 \\
              & QE--FE  &  -8.239 & 2.461 &  -6.320 &  2.616 & 0.850 &  6.123 & 0.902 &  7.067 \\
      \midrule
      $25\%$  & DiVE    &   3.773 & 2.520 &   6.734 & 43.136 & 0.915 &  7.479 & 0.944 &  8.632 \\
              & QE--RE  &  -5.414 & 2.573 &   3.186 & 19.802 & 0.909 &  7.448 & 0.946 &  8.596 \\
              & QE--FE  &  -8.958 & 3.265 & -33.623 & 12.455 & 0.794 &  6.121 & 0.859 &  7.064 \\
      \midrule
      $50\%$  & DiVE    &   4.453 & 3.853 &  0.608 & 38.437 & 0.911 &  9.132 & 0.946 & 10.540 \\
              & QE--RE  &  -2.596 & 3.580 &  -3.883 & 20.736 & 0.927 &  8.824 & 0.948 & 10.185 \\
              & QE--FE  &  -7.940 & 4.539 & -58.145 & 34.239 & 0.746 &  6.127 & 0.796 &  7.071 \\
      \midrule
      $75\%$  & DiVE    &   4.371 & 7.477 &   -0.811 & 47.522 & 0.910 & 13.066 & 0.946 & 15.081 \\
              & QE--RE  &  0.278 & 6.341 &  -1.281 & 27.727 & 0.923 & 12.515 & 0.952 & 14.445 \\
              & QE--FE  & -8.808 & 9.114 & -80.087 & 64.259 & 0.590 &  6.126 & 0.658 &  7.071 \\
      \bottomrule
    \end{tabular}
    \begin{tablenotes}[flushleft]\footnotesize
        \item \textit{Note.} Results are based on 1,000 simulation replicates per heterogeneity level.
        \%Bias, relative bias; \%MSE, relative mean-squared error; CP, empirical coverage probability of 95\% confidence intervals; AW, average width.
        $z$-based intervals use the standard normal critical value, while $t$-based intervals use quantiles from a $t$-distribution with $N\!-\!1$ degrees of freedom.
        The true pooled difference is the benchmark for point-estimate metrics; variance metrics are evaluated against method-specific analytic targets defined in Appendix~B.
    \end{tablenotes}
  \end{threeparttable}
\end{table}

\begin{table}[t]
  \centering
  \caption{Performance of point and variance estimators under a log-normal scenario ($N=10$, varying sample sizes with average $n=300$).}
  \label{tab:res_V_logN10n300}
  \renewcommand{\arraystretch}{1.15}
  \begin{threeparttable}
    \begin{tabular}{@{}cc*{8}{S}@{}}
      \toprule
      &
      & \multicolumn{2}{c}{\textbf{Point}}
      & \multicolumn{2}{c}{\textbf{Variance}}
      & \multicolumn{2}{c}{\textbf{$z$-based}}
      & \multicolumn{2}{c}{\textbf{$t$-based}} \\
      \cmidrule(lr){3-4}\cmidrule(lr){5-6}\cmidrule(lr){7-8}\cmidrule(lr){9-10}
      \textbf{$I^2$}
      & \textbf{Method}
      & \textbf{\%Bias} & \textbf{\%MSE}
      & \textbf{\%Bias} & \textbf{\%MSE}
      & \textbf{CP}    & \textbf{AW}
      & \textbf{CP}    & \textbf{AW} \\
      \midrule
      $0\%$   & DiVE    &   1.751 & 0.588 &   2.380 & 53.141 & 0.896 & 3.514 & 0.928 & 4.055 \\
              & QE--RE  &  -2.243 & 0.613 &  25.105 & 26.552 & 0.951 & 4.046 & 0.971 & 4.670 \\
              & QE--FE  &  -2.728 & 0.625 &  -0.270 &  0.833 & 0.930 & 3.655 & 0.955 & 4.218 \\
      \midrule
      $25\%$  & DiVE    &   1.118 & 0.814 &   2.130 & 64.002 & 0.899 & 4.331 & 0.925 & 4.999 \\
              & QE--RE  &  -2.408 & 0.847 &   4.812 & 23.230 & 0.938 & 4.484 & 0.958 & 5.175 \\
              & QE--FE  &  -3.373 & 0.928 & -35.927 & 13.474 & 0.858 & 3.648 & 0.908 & 4.210 \\
      \midrule
      $50\%$  & DiVE    &   1.614 & 1.526 &   4.757 & 71.479 & 0.878 & 5.712 & 0.927 & 6.593 \\
              & QE--RE  &  -1.029 & 1.361 &  0.621 & 26.873 & 0.918 & 5.460 & 0.946 & 6.302 \\
              & QE--FE  &  -2.806 & 1.623 & -62.013 & 38.772 & 0.751 & 3.661 & 0.811 & 4.225 \\
      \midrule
      $75\%$  & DiVE    &   1.915 & 3.259 & 3.225 & 77.568 & 0.880 & 8.479 & 0.912 & 9.786 \\
              & QE--RE  &  -0.400 & 2.599 &  1.467 & 29.088 & 0.899 & 7.713 & 0.943 & 8.902 \\
              & QE--FE  &  -2.683 & 3.375 & -83.097 & 69.160 & 0.548 & 3.656 & 0.620 & 4.219 \\
      \bottomrule
    \end{tabular}
    \begin{tablenotes}[flushleft]\footnotesize
        \item \textit{Note.} Results are based on 1,000 simulation replicates per heterogeneity level.
        \%Bias, relative bias; \%MSE, relative mean-squared error; CP, empirical coverage probability of 95\% confidence intervals; AW, average width.
        $z$-based intervals use the standard normal critical value, while $t$-based intervals use quantiles from a $t$-distribution with $N\!-\!1$ degrees of freedom.
        The true pooled difference is the benchmark for point-estimate metrics; variance metrics are evaluated against method-specific analytic targets defined in Appendix~B.
    \end{tablenotes}
  \end{threeparttable}
\end{table}


\begin{table}[t]
  \centering
  \caption{Performance of point and variance estimators under a log-normal scenario ($N=30$, varying sample sizes with average $n=300$).}
  \label{tab:res_V_logN30n300}
  \renewcommand{\arraystretch}{1.15}
  \begin{threeparttable}
    \begin{tabular}{@{}cc*{8}{S}@{}}
      \toprule
      &
      & \multicolumn{2}{c}{\textbf{Point}}
      & \multicolumn{2}{c}{\textbf{Variance}}
      & \multicolumn{2}{c}{\textbf{$z$-based}}
      & \multicolumn{2}{c}{\textbf{$t$-based}} \\
      \cmidrule(lr){3-4}\cmidrule(lr){5-6}\cmidrule(lr){7-8}\cmidrule(lr){9-10}
      \textbf{$I^2$}
      & \textbf{Method}
      & \textbf{\%Bias} & \textbf{\%MSE}
      & \textbf{\%Bias} & \textbf{\%MSE}
      & \textbf{CP}    & \textbf{AW}
      & \textbf{CP}    & \textbf{AW} \\
      \midrule
      $0\%$   & DiVE    &   1.634 & 0.218 &   3.429 & 14.535 & 0.926 & 2.117 & 0.935 & 2.209 \\
              & QE--RE  &  -3.027 & 0.281 &  15.186 &  8.929 & 0.899 & 2.258 & 0.914 & 2.356 \\
              & QE--FE  &  -3.320 & 0.299 &  -0.931 &  0.290 & 0.861 & 2.105 & 0.881 & 2.196 \\
      \midrule
      $25\%$  & DiVE    &   1.679 & 0.296 &  2.076 & 16.460 & 0.932 & 2.621 & 0.941 & 2.735 \\
              & QE--RE  &  -2.526 & 0.328 &  1.777 &  8.792 & 0.918 & 2.579 & 0.926 & 2.691 \\
              & QE--FE  &  -3.264 & 0.377 & -36.540 & 13.567 & 0.811 & 2.103 & 0.831 & 2.195 \\
      \midrule
      $50\%$  & DiVE    &   1.436 & 0.521 &   3.069 & 20.890 & 0.928 & 3.441 & 0.937 & 3.590 \\
              & QE--RE  &  -1.758 & 0.453 &   -0.378 &  8.632 & 0.923 & 3.181 & 0.933 & 3.319 \\
              & QE--FE  &  -3.454 & 0.623 & -62.870 & 39.667 & 0.704 & 2.107 & 0.724 & 2.199 \\
      \midrule
      $75\%$  & DiVE    &   1.546 & 1.091 &   -0.200 & 22.503 & 0.931 & 5.073 & 0.942 & 5.293 \\
              & QE--RE  &  -0.697 & 0.817 &   -0.633 &  8.244 & 0.943 & 4.454 & 0.950 & 4.647 \\
              & QE--FE  &  -3.280 & 1.202 & -83.547 & 69.841 & 0.526 & 2.104 & 0.545 & 2.196 \\
      \bottomrule
    \end{tabular}
    \begin{tablenotes}[flushleft]\footnotesize
        \item \textit{Note.} Results are based on 1,000 simulation replicates per heterogeneity level.
        \%Bias, relative bias; \%MSE, relative mean-squared error; CP, empirical coverage probability of 95\% confidence intervals; AW, average width.
        $z$-based intervals use the standard normal critical value, while $t$-based intervals use quantiles from a $t$-distribution with $N\!-\!1$ degrees of freedom.
        The true pooled difference is the benchmark for point-estimate metrics; variance metrics are evaluated against method-specific analytic targets defined in Appendix~B.
    \end{tablenotes}
  \end{threeparttable}
\end{table}


\end{document}